\documentclass[12pt,a4paper]{article}
\usepackage{amsfonts}
\usepackage[countmax]{subfloat}
\usepackage{algorithm, algorithmicx}
\usepackage[noend]{algpseudocode}
\usepackage{graphics}
\usepackage{tikz}
\usetikzlibrary{decorations.markings,quantikz}
\usepackage{enumerate}
\usepackage{hyperref}
\usepackage{ulem}
\usepackage{mathtools}
\mathtoolsset{showonlyrefs}
\usepackage{amsmath}
\usepackage{graphicx}
\usepackage{enumerate}
\usepackage{amssymb}
\usepackage{color}
\usepackage{amsthm}
\usepackage{caption}
\usepackage{subcaption}
\usepackage{authblk}

\usepackage[left=0.75in,right=0.75in,top=0.75in,bottom=0.75in]{geometry}
\newcommand{\appref}[1]{\hyperref[#1]{{Appendix~\ref*{#1}}}}
\newcommand{\be}{\begin{eqnarray} \begin{aligned}}
\newcommand{\ee}{\end{aligned} \end{eqnarray} }
\newcommand{\benn}{\begin{eqnarray*} \begin{aligned}}
\newcommand{\eenn}{\end{aligned} \end{eqnarray*}}
\newcommand*{\textfrac}[2]{{{#1}/{#2}}}

\newcommand*{\cC}{\mathcal{C}}
\newcommand*{\cE}{\mathcal{E}}

\newcommand*{\cG}{\mathcal{G}}

\newcommand*{\cI}{\mathcal{I}}

\newcommand*{\cL}{\mathcal{L}}

\newcommand*{\cN}{\mathcal{N}}

\newcommand*{\cQ}{\mathcal{Q}}
\newcommand*{\cR}{\mathcal{R}}
\newcommand*{\cO}{\mathcal{O}}

\newcommand*{\tr}{\mathop{\mathrm{tr}}\nolimits}

\newcommand*{\supp}{\mathrm{supp}}

\newcommand{\bc}{\begin{center}}
\newcommand{\ec}{\end{center}}

\newtheorem{theorem}{Theorem}[section]
\newtheorem{lemma}[theorem]{Lemma}

\newtheorem{corollary}[theorem]{Corollary}

\usepackage{amsfonts}

\def\01{\{0,1\}}

\newcommand{\eps}{\varepsilon}

\definecolor{myred}{RGB}{204,51,17}
\definecolor{myblue}{RGB}{0,119,187}
\definecolor{mygrey}{RGB}{187,187,187}

\tikzset{bluepauli/.style={
        draw,
        thick,
        circle,
        fill=myblue!50,
        minimum size=1cm }}
\tikzset{bluecliff/.style={
        draw,
        thick,
        rectangle,
        fill=myblue!90,
        minimum size=1cm }}
\tikzset{greypauli/.style={
        draw,
        thick,
        circle,
        fill=mygrey!50,
        minimum size=1cm }}
\tikzset{redcliff/.style={
        draw,
        thick,
        rectangle,
        fill=myred!90,
        minimum size=1cm }}
\tikzset{redpauli/.style={
        draw,
        thick,
        circle,
        fill=myred!50,
        minimum size=1cm }}
\tikzset{greycliff/.style={
        draw,
        thick,
        rectangle,
        fill=white,
        minimum size=1cm
 }}

\newcommand*{\ExpE}{\mathbb{E}}

\newcommand*{\Utelep}{U^{\mathsf{Telep}}}

\newcommand*{\telep}{\mathsf{Telep}}

\newcommand*{\restrictedtelepR}{R_{\Utelep_n}\restriction_\xi}
\newcommand*{\telepR}{R_{\Utelep_n}}

\begin{document}

\newcommand*{\AC}{\mathsf{AC}}
\newcommand*{\NC}{\mathsf{NC}}
\newcommand*{\QNC}{\mathsf{QNC}}
\renewcommand*{\P}{\mathsf{P}}
\renewcommand*{\L}{\mathsf{L}}
\newcommand*{\TC}{\mathsf{TC}}
\newcommand*{\NL}{\mathsf{NL}}

\newcommand*{\DT}{\mathsf{DT}}
\newcommand*{\poly}{\mathsf{poly}}

\newcommand*{\Cliff}{\mathsf{Clifford}}

\newcommand{\ZZ}{\mathbb{Z}}

\newcommand{\labelgroup}[5]{\POS"#1,#5"."#2,#5"."#1,#5"."#2,#5", \POS"#1,#5"."#2,#5"."#1,#5"."#2,#5"*!C!<1em,#3>=<0em>{#4}}
\newcommand*{\cfinal}{C_{\textrm{tot}}}

\DeclareDocumentCommand{\makereversebit}{O{+45}O{}m}{
	\arrow[arrows,line cap=round,to path={(\tikztostart) -- ($(\tikztostart)!{+0.5/cos(#1)}!#1:(\tikztotarget)$) node [anchor=west,style={#2}]{#3} -- (\tikztotarget)}]{d}
}
\newcommand*{\localfunctions}{\mathsf{Local}}

\newcommand*{\teleport}{\mathsf{Telep}}

\newcommand{\coml}[1]{{\textcolor{blue}{L: #1}}}
\newcommand{\comx}[1]{{\textcolor{red}{X: #1}}}

\newcommand*{\cliff}{\mathsf{Cliff}}
\newcommand*{\pauli}{\mathsf{Pauli}}
\newcommand*{\CNOT}{\mathsf{CNOT}}
\newcommand*{\pthres}{p_{\textrm{thres}}}
\newcommand*{\manc}{{m_{\textrm{aux}}}}
\newcommand*{\rec}{\mathsf{Rec}}
\newcommand*{\rep}{\mathsf{Rep}}
\newcommand*{\parity}{\mathsf{Parity}}
\newcommand*{\dec}{\mathsf{Dec}}
\newcommand*{\Uext}{U^{\textrm{ext}}}
\newcommand*{\cdepth}{\mathsf{depth}} %
\newcommand*{\csize}{\mathsf{size}} %
\newcommand*{\enc}{\mathsf{Enc}}

\title{A colossal advantage:  3D-local noisy shallow quantum circuits defeat unbounded fan-in classical circuits}
\author{Libor Caha}
\author{Xavier Coiteux-Roy}
\author{Robert K\"onig}
\affil{\small School of Computation, Information and Technology, Technical University of Munich \& \\
Munich Center for Quantum Science and Technology, Munich, Germany.}

\maketitle

\begin{abstract}
We present a computational problem with the following properties:
(i) Every instance can be solved with near-certainty by a constant-depth quantum circuit using only nearest-neighbor gates in $3D$ even when its implementation is corrupted by noise.
(ii) Any constant-depth classical circuit composed of unbounded fan-in AND, OR, as well as NOT gates, i.e., an $\AC^0$-circuit,  of size smaller than a certain subexponential, fails to solve a uniformly random instance with probability greater than a certain constant. 
Such an advantage against unbounded fan-in classical circuits was previously only known in the noise-free case or without locality constraints. We overcome these limitations, proposing a quantum advantage demonstration amenable to  experimental realizations. 
  Subexponential circuit-complexity lower bounds have traditionally been referred to as exponential. We use the term colossal since our fault-tolerant $3D$ architecture resembles a certain Roman monument.
\end{abstract}

\tableofcontents

\section{Introduction}
In the present era of near- and intermediate-term quantum devices, fully fault-tolerant, universal, scalable quantum computers belong to the realm of science fiction. Fortunately, even with significantly more modest, imperfect resources, quantum information-processing can be superior to purely classical protocols. A key challenge is to identify such cases and characterize the potential of limited quantum devices. In the context of computation, this not only involves identifying computational problems that straddle the fine line of being amenable to quantum algorithmic solutions while being beyond the reach of comparable classical devices. It also requires designing computational architectures that can withstand noise, and can realistically be built. We make a proposal in this direction, establishing the strongest known complexity-theoretic separation between the computational power of noisy, $3D$-local shallow quantum circuits as opposed to that of shallow classical circuits with unbounded fan-in gates.

The study of circuit complexity has a long history in classical computer science, where capabilities and limitations of different classes of circuits can be accurately characterized. Celebrated results include for example the statement that the parity of $n$~bits cannot be computed in~$\AC^0$, i.e., by polynomial-size, constant-depth circuits with unbounded fan-in AND, OR, and NOT gates~\cite{furst_parity_1984,Ajtai1983,HastadThesis}. It is also known for example that~$\AC^0$ is richer than the class~$\NC^0$ of problems solved by constant-depth bounded fan-in circuits. Trivially, a computational problem separating these two classes is that of computing the AND of $n$~bits.

Motivated by the successes of classical circuit complexity, analogous quantum circuit classes have been investigated from the early days of quantum computing, see e.g.,~\cite{hoyerspalek}.  More recent work studied how quantum circuit complexity classes compare to classical ones. In~\cite{BGK}, it was shown that there is a relation problem that is solvable by a constant-depth (so-called shallow) quantum circuit, but whose solution by a classical circuit with bounded fan-in gates requires at least logarithmic depth. This provides an unconditional separation between~$\QNC^0$ (the class of shallow quantum circuits) and~$\NC^0$. A stronger complexity-theoretic result 
separating~$\QNC^0$ from~$\AC^0$
was subsequently established by Bene Watts et~al.~\cite{BeneWattsKothariSchaefferTalAC0}: They showed that a certain problem, the so-called relaxed parity-halving problem, has the property of being solvable by a constant-depth quantum circuit, yet any classical circuit with unbounded fan-in gates requires superpolynomial (in fact subexponential) size for this task.

{\bf Quantum advantage in the presence of noise.} While these results establish an advantage of certain shallow quantum circuits compared to similarly defined classical circuits (respectively associated complexity classes), it is important to study whether or not these findings translate to an experimentally observable  difference between quantum and classical devices in  the real world.  That is, is it possible to benefit from quantum information-processing when all operations and building blocks are noisy? Not surprisingly, this central question was also posed immediately after Shor's discovery of the factoring algorithm and other early quantum algorithms. In that case, the (theoretical) resolution is the fault-tolerance threshold theorem~\cite{Aharonov1997}: It  ensures that quantum computations can be rendered fault-tolerant while incurring only a polynomial overhead in resources assuming the noise strength is below some threshold. 

When studying shallow (i.e., constant-depth) quantum circuits, the standard fault-tolerance constructions from the fault-tolerance threshold theorem do not apply: This is because these typically do not preserve shallowness. This means that rendering complexity-theoretic separations fault-tolerant generally requires new,  non-standard error-correction techniques. A first result in this direction was obtained in~\cite{BGKT}, where it was shown that the separation 
between shallow quantum and bounded fan-in classical circuits (first demonstrated in~\cite{BGK}) can be made  robust to noise: Even noisy shallow quantum circuits beat  such classical $\NC^0$-circuits at a certain task. 

\subsection{Our result}
Here we seek a stronger separation: We show that noisy shallow $3D$-local quantum circuits solve a computational task with higher probability than (ideal) unbounded fan-in classical~$\AC^0$-circuits of subexponential size. This can thus be seen as a fault-tolerant counterpart to  the work~\cite{BeneWattsKothariSchaefferTalAC0}. More precisely, we show the following:
\begin{theorem}[Fault-tolerant quantum advantage against $\AC^0$, informal version]\label{thm:mainresultinformal}
There is a computational problem with the following properties:
\begin{enumerate}[(i)]
\item
The problem  is beyond the reach of $AC^0$-circuits: Any $AC^0$-circuit solving the problem with probability at least $0.9888$ on average over a randomly chosen instance has superpolynomial (in fact subexponential) size.
\item
The problem  can be solved with average probability at least~$0.99$ by a $3D$-local shallow quantum circuit even in the presence of local stochastic noise.  That is, the quantum advantage can be observed using a shallow, noisy quantum circuit which only involves nearest-neighbor gates on qubits arranged on a regular $3D$~lattice.
\end{enumerate}
 \end{theorem}
Our result thus strengthens the result of~\cite{BGKT}: While requiring a comparable amount of (imperfect) quantum resources/capabilities, and only local operations in~$3D$, it establishes a quantum advantage against $\AC^0$ instead of $\NC^0$.
  
To establish our main result, we borrow some of the fault-tolerance techniques from Ref.~\cite{BGKT}.  We note that these  techniques have also been applied in Ref.~\cite{grier2021interactiveNoisy}  to render various quantum advantage proposals fault-tolerant,
including interactive settings (as considered in the noise-free setup in~\cite{grier2020interactive}).
In particular, one of the claims of~\cite{grier2021interactiveNoisy} combined with the noise-free result from Ref.~\cite{BeneWattsKothariSchaefferTalAC0} may superficially appear identical to our main claim. Namely that noisy shallow quantum circuits are superior to~$\AC^0$-circuits. However, our result is stronger than that of~\cite{grier2021interactiveNoisy} because it only requires a $3D$-local quantum circuit. As we argue in Section~\ref{sec:localityconsiderationsnecessary}, locality considerations are essential not only for technological reasons, but also in order to ensure that the noise model is meaningful.

\subsection{The need for considering locality\label{sec:localityconsiderationsnecessary}} 
Ref.~\cite{grier2021interactiveNoisy} shows a quantum advantage of noisy shallow quantum circuits against~$\AC^0$, but the results obtained in that paper  (including additional results for interactive settings) do not take into account the problem of (non)-locality of operations.  Let us briefly comment on how our contribution goes beyond the current state-of-the-art, and  why the consideration of geometric locality is essential. The noise-tolerant quantum circuit constructed in~\cite{grier2021interactiveNoisy} for demonstrating an advantage against $\AC^0$ is non-local: It  does not specify a geometric layout of the qubits, and thus essentially assumes all-to-all-connectivity. 
More precisely, when considering such non-local circuits, the noise model of local stochastic noise~\cite{gottesmanoverhead,fawzi2018constant} occurring between circuit gate layers (see Section~\ref{sec:localstochasticnoise})   does not accurately reflect typical physical setups where a (large) number of qubits are spatially arranged, e.g., on a lattice. In such a setting, gates between distant qubits typically need to be implemented by a sequence of SWAP-gates. This requires a more detailed error analysis: The local stochastic noise model  is only meaningful if  these additional SWAP-gates (and hence the geometric layout) are explicitly incorporated in the analysis. Further complications may arise (when geometrically embedding qubits into space in some manner)  because such SWAP-gates will increase the circuit depth, thus potentially going outside the regime of shallow, that is, constant-depth circuits.

We note that a fault-tolerant quantum advantage experiment robust to  an interesting alternative noise model involving 
adversarial (instead of stochastic) 
qubit loss was proposed in~\cite{hasegawalegallArbitrarycorruption}. The corresponding circuit is local on an expander graph, but presumably not local for any embedding in~$3D$.

This motivates the question of whether or not a quantum advantage against $\AC^0$ can be demonstrated by means of a {\em geometrically local} noise-resilient quantum circuit, where qubits are arranged 
 (with constant density) in space and gates act locally on constant-size neighborhood of qubits. We answer this affirmatively by 
 Theorem~\ref{thm:mainresultinformal} above. Clearly, the fact that the corresponding advantage can be observed by a {\em local}, error-resilient quantum circuit makes this result particularly attractive for potential experimental implementations. 

\subsection{From single-qubit gate teleportation to complexity theory}
 
The key difference of our work to~\cite{grier2021interactiveNoisy} is in the choice of underlying problem: Whereas the authors of~\cite{grier2021interactiveNoisy} 
rely on the  so-called relaxed parity-halving problem, our construction is derived from what we call the single-qubit gate-teleportation circuit (see Fig.~\ref{fig:gateteleportationcircuit} in Sec.~\ref{sec:singlequbitteleportationrelation}). 
The latter is a concatenation of multiple applications of the standard gate-teleportation procedure~\cite{GottesmanChuangNature} for single-qubit Clifford gates, but without the application of the corresponding Pauli corrections: It is a classically controlled Clifford circuit (taking $n$~single-qubit Clifford elements as input), and outputting the result of $n$~Bell measurements (equivalently, a sequence of $n$~Pauli observables), see Section~\ref{sec:singlequbitteleportationrelation}. The geometrically $1D$-local nature ultimately yields  a fault-tolerant quantum advantage proposal with a $3D$-local quantum circuit.

On our route to establishing a quantum advantage of noisy shallow quantum circuits against~$\AC^0$, we establish the following result for (ideal) circuits.
\begin{theorem}[Single-qubit gate teleportation yields a quantum advantage against~$\AC^0$, informal version]\label{thm:gateteleportationqadvantage}
Let $\cC$ be an $\AC^0$-circuit which, for a uniformly  chosen sequence~$C=(C_0,\ldots,C_{n-1})\in\cliff^n$ of $n$~single-qubit Clifford gates, produces --- with probability at least~$0.9888$ on average --- a string~$x\in \{0,1\}^{2n}$ which  occurs with non-zero probability in the output distribution of the single-qubit gate-teleportation circuit on input~$C$. Then the size of~$\cC$ is superpolynomial (in fact subexponential). 
\end{theorem}
In the terminology of~\cite{wang2021possibilistic}, Theorem~\ref{thm:gateteleportationqadvantage} states that 
the computational problem of ``possibilistically'' simulating the single-qubit gate-teleportation circuit is  infeasible for~$\AC^0$-circuits of polynomial size. It is arguably one of the most basic computational problems that one can think of, and has a considerably simpler structure than e.g., the parity-halving problem~\cite{BeneWattsKothariSchaefferTalAC0}.

This single-qubit gate-teleportation problem shares a few attractive average-case hardness features with prior work such as~\cite{LeGallAveragecase,BeneWattsKothariSchaefferTalAC0}: The bound on the classical circuits considered here involves the average over a fully random input. In contrast, the result of~\cite{BGK,BGKT} (as well as our results for the noise-tolerant setup) require restricting to a subset of inputs corresponding to valid problem instances.

The role of gate teleportation in complexity theory has already been recognized in  seminal work by Terhal and DiVincenzo~\cite{terhal2002adaptive}. We believe  Theorem~\ref{thm:gateteleportationqadvantage} adds to this by establishing a new,  unconditional complexity-theoretic result (stronger than our prior work~\cite{teleppaper}, which only provided a separation for~$\NC^0$). In a sense, our work is complementary to considerations with the well-recognized importance of gate teleportation for quantum fault tolerance:  We ultimately obtain a fault-tolerant quantum advantage by exploiting gate teleportation, not for noise-resilience, but to obtain a simple and natural computational problem.

\section{The magic-square game and its variants}
Our demonstration of quantum advantage is built upon a connection between the single-qubit teleportation circuit and pseudo-telepathy games.

\subsection{Pseudo-telepathy games and the magic square}
Pseudo-telepathy games are non-local games that cannot be won with certainty by classical players with shared randomness, yet can be won with certainty by quantum players sharing entanglement. The celebrated magic-square game~\cite{PERES1990107,mermin} is an example of such a game in the bipartite setting. In this game, Alice's input is an element $\alpha \in\{1,2,3\}$, whereas Bob's input is $\beta \in\{1,2,3\}$. Alice outputs $\left(x_1, x_2, x_3\right) \in\{-1,1\}^3$ and Bob outputs $\left(y_1, y_2, y_3\right) \in\{-1,1\}^3$. The two players win the game if
\begin{align}
\begin{aligned}
& x_1 x_2 x_3=\ \, 1 \\
& y_1 y_2 y_3=-1
\end{aligned}\qquad \text { and } \qquad y_\alpha x_\beta=\ 1\ . \label{eq:winningconditionmagicsquare}
\end{align}

\begin{figure}
\centering
$$
\begin{array}{lll}
\hphantom{-}X \otimes I & \hphantom{-}I \otimes X & \hphantom{-}X \otimes X \\
\hphantom{-}I \otimes Z & \hphantom{-}Z \otimes I & \hphantom{-}Z \otimes Z \\
-X \otimes Z & -Z \otimes X & \hphantom{-}Y \otimes Y
\end{array}
$$
\caption{These Pauli observables form a magic square: the product of any row of observables is $I\otimes I$, while the product of any column of observables is $-I\otimes I$.\label{fig:magicsquarestandard}}
\end{figure}

Two non-communicating classical players can win this game for a uniformly chosen instance $(\alpha, \beta) \in$ $\{1,2,3\}^2$ with probability at most $8 / 9$. In contrast, the game can be won with certainty by players sharing two copies $\Phi_{A_1 B_1} \otimes \Phi_{A_2 B_2}$ of the maximally entangled state $|\Phi\rangle=\frac{1}{\sqrt{2}}(|00\rangle+|11\rangle)$. A corresponding protocol can be obtained from the ``magic square,'' see Fig.~\ref{fig:magicsquarestandard}. On input $\alpha \in$ $\{1,2,3\}$, Alice performs a measurement of the three commuting observables in the $\alpha$-th row of the square, outputting the corresponding eigenvalues. Similarly, on input $\beta \in\{1,2,3\}$, Bob measures the three commuting observables in the $\beta$-th column of the square to obtain $\left(y_1, y_2, y_3\right)$.

\subsection{Pseudo-telepathy games with single-qubit control\label{sec:singlequbitcontrolpseudo}}
 Here we need a special variant of the magic-square game. It has the special feature that there is a quantum strategy that always allows Alice and Bob to choose the correct measurement basis by making a Bell measurement preceded by a single-qubit rotation depending only on their respective input.

In more detail, we consider two sets~$\{U_\alpha\}_{\alpha=1}^3$ and $\{V_\beta\}_{\beta=1}^3$ of single-qubit  unitaries defined as follows. For $P\in\{X,Y,Z\}$, define the Clifford gate
\begin{align}
R_P&=\frac{1}{\sqrt{2}}(I-iP)\ .
\end{align}
(Conjugation by~$R_P$ realizes a counterclockwise rotation of the Bloch sphere around the axis defined by~$P$ with angle~$\pi/2$.) We then define the single-qubit Clifford unitaries
\begin{align}
(U_1,U_2,U_3)&=(R_X, ZR_Y^\dagger,R_Z^\dagger) 
\end{align}
and
\begin{align}
(V_1,V_2,V_3)&=(I,R_ZR_Y^\dagger,R_X^\dagger R_Y)\ .
\end{align}
Consider the following procedure: 
\begin{enumerate}[(i)]
\item
Alice and Bob share two copies~$\Phi_{A_1B_1}\otimes\Phi_{A_2B_2}$ of the maximally entangled state~$\Phi$.
\item
On input~$\alpha\in \{1,2,3\}$, Alice applies the unitary~$U_\alpha$ to system~$A_1$.
\item
Similarly, on input~$\beta\in \{1,2,3\}$, Bob applies the unitary~$V_\beta$ to system~$B_2$.
\item\label{it:bellmeasurementalice}
Alice performs a Bell basis change by applying $(H\otimes I)_{A_1A_2}\CNOT_{A_1A_2}$. She then measures in the computational basis, obtaining an outcome~$(u_1,u_2)\in \{0,1\}^2$.
\item\label{it:bellmeasurementbob}
Bob similarly applies 
$(H\otimes I)_{B_1B_2}\CNOT_{B_1B_2}$ and subsequently measures in the computational basis, obtaining an outcome~$(v_1,v_2)\in \{0,1\}^2$.
\end{enumerate}

\begin{figure}
\centering
\begin{tikzpicture}
\node[] (img1) at (-0.02,0){
\begin{quantikz}
&\lstick{$\alpha$} &\cwbend{1}& & &\\
&\qw & \gate{U_\alpha}&\ctrl{1}&\gate{H}& \meter{}&\rstick[1]{$\hspace{-15pt}=u_1$} \\
&\qw &\qw&\targ{} &\qw &\meter{}&\rstick[1]{$\hspace{-15pt}=u_2$}
\end{quantikz}};
\node[] (img2) at (-0.05,-4){
\begin{quantikz}
&\qw &\qw &\ctrl{1}&\gate{H}& \meter{}&\rstick[1]{$\hspace{-15pt}=v_1$} \\
&\qw & \gate{V_\beta}&\targ{} &\qw &\meter{}&\rstick[1]{$\hspace{-15pt}=v_2$}\\
&\lstick{$\beta$} &\cwbend{-1}& & &
\end{quantikz}};
\draw[thick] (-3,0.25) to (-4.625,-1.375) to (-3,-3);
\draw[thick] (-3,-1) to (-4.625,-2.625) to (-3,-4.25);
\node[] (phi1) at (-5,-1.375){$\Phi$};
\node[] (phi2) at (-5,-2.625){$\Phi$};
\node[] (a1) at (-2.75,0.5){\small $A_1$};
\node[] (a2) at (-2.75,-0.75){\small $A_2$};
\node[] (b1) at (-2.75,-2.75){\small $B_1$};
\node[] (b2) at (-2.75,-4){\small $B_2$};
\end{tikzpicture}
\caption{Quantum strategy for the pseudotelepathy game~$\cG$\label{fig:magicsquarecircuit}.}
\end{figure}
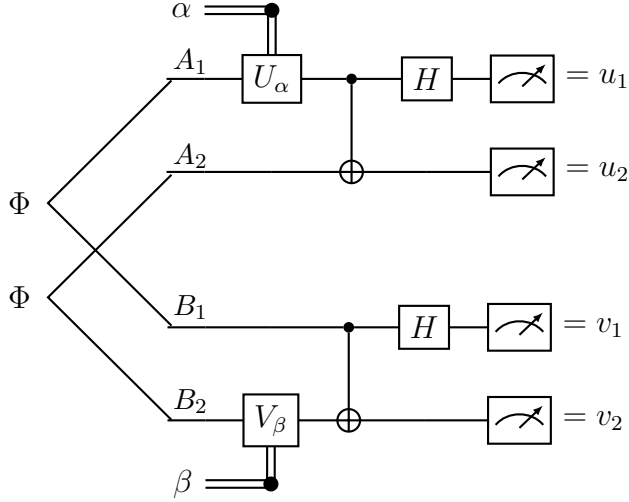

This strategy is illustrated in Fig.~\ref{fig:magicsquarecircuit}.
We argue the following:
\begin{lemma}\label{lem:classicalwinningG}
The described protocol is a quantum strategy of a pseudotelepathy game~$\cG$ which achieves winning probability~$1$. For $(\alpha,\beta)\in \{1,2,3\}^3$ chosen uniformly at random, two non-communicating classical players can win~$\cG$ with probability at most~$8/9$.
\end{lemma}

\begin{proof}
We will argue that the outputs of Alice and Bob  produced by the described protocol can be post-processed locally to provide a solution to the standard magic-square problem. That is, there are functions~$\{f_\alpha\}_{\alpha\in\{1,2,3\}}$ with~$f_\alpha:\{0,1\}^2\rightarrow \{-1,1\}^3$ and functions $\{g_\beta\}_{\beta\in \{1,2,3\}}$ where $g_\beta:\{0,1\}^2\rightarrow\{-1,1\}^3$ such that the following holds. Suppose the protocol produces outputs~$(u_1,u_2),(v_1,v_2)$ on  input~$(\alpha,\beta)\in \{1,2,3\}^2$. Set 
\begin{align}
x=f_\alpha(u_1,u_2)\qquad\textrm{ and }\qquad y=g_\beta (v_1,v_2)\ .\label{eq:xyoutputpostprocessed}
\end{align} Then $x=(x_1,x_2,x_3)$ and $y=(y_1,y_2,y_3)$ satisfy the winning condition~\eqref{eq:winningconditionmagicsquare} of the magic-square game.

The fact that any output of the described quantum protocol  yields a solution to the magic-square game by local processing implies that producing such an output with certainty is infeasible  for non-communicating classical players. In other words, we can define a non-local game~$\cG$ where the task is to ``possibilistically'' simulate the circuit in Fig.~\ref{fig:magicsquarecircuit}: Given~$(\alpha,\beta)\in \{0,1\}^3$, the task is to output a pair~$(u,v)\in \{0,1\}^2\times \{0,1\}^2$ that occurs with non-zero probability in the output distribution of the circuit. Because of the existence of the functions~$\{f_\alpha\}_\alpha, \{g_\beta\}_\beta$, i.e., the relation with the magic-square game, the game~$\cG$ can only be won with probability~$8/9$ by classical strategies, yet can (trivially) be won with certainty by the  quantum strategy illustrated in Fig.~\ref{fig:magicsquarecircuit}.

It remains to show the existence of functions~$\{f_\alpha\}_{\alpha\in\{1,2,3\}}$ and $\{g_\beta\}_{\beta\in\{1,2,3\}}$  with the desired property. To do so, we relate the measurement result to the magic square of operators given in Fig.~\ref{fig:magicsquare}.
\begin{figure}\centering
        \begin{tabular}{ c c c }
	        $X_{A_1}X_{A_2}$ & $\mspace{3mu}Y_{A_1}Z_{A_2}$ & $\mspace{2mu}Z_{A_1}Y_{A_2}$\\
        	$Y_{A_1}Y_{A_2}$ & $\mspace{5mu}Z_{A_1}X_{A_2}$ & $X_{A_1}Z_{A_2}$\\
        	$Z_{A_1}Z_{A_2}$ & $X_{A_1}Y_{A_2}$ & $\mspace{4mu}Y_{A_1}X_{A_2}$
        \end{tabular}
        \caption{We consider a modified version of the magic square as given by these operators. Again, the product of any row of observables is~$I$ and the product of any column of observables is~$-I$. We consider operators acting on Alice's system only here.}
        \label{fig:magicsquare}
   \end{figure}

Let us first consider Alice's measurement results $(u_1,u_2)\in \{0,1\}^2$ on input~$\alpha\in\{1,2,3\}$. They correspond to the 
eigenvalues $(-1)^{u_1}$ and $(-1)^{u_2}$ of the observables~$\tilde{U}^\dagger_\alpha (Z\otimes I) \tilde{U}_\alpha$ and 
$\tilde{U}^\dagger_\alpha (I\otimes Z) \tilde{U}_\alpha$, where 
\begin{align}
\tilde{U}_\alpha &=(H\otimes I)\CNOT (U_\alpha\otimes I)\ . 
\end{align} 
Straightforward computation then gives the following table:
\begin{center}
\begin{tabular}{l|c|cc}
$\alpha$ & $U_\alpha$ & $(-1)^{u_1}$ eigenvalue of & $(-1)^{u_2}$ eigenvalue of\\
\hline
1 & $R_X$ & $X_{A_1}X_{A_2}$ & $Y_{A_1}Z_{A_2}$\\
2 & $ZR_Y^\dagger$ & $Z_{A_1}X_{A_2}$ & $X_{A_1}Z_{A_2}$\\
3 & $R_Z^\dagger$ & $Y_{A_1}X_{A_2}$ & $Z_{A_1}Z_{A_2}$
\end{tabular}
\end{center}
Observe that for each~$\alpha\in \{1,2,3\}$, the measurement outcomes~$(u_1,u_2)$ directly determine the eigenvalues of two observables in row~$\alpha$ of the table given in Fig.~\ref{fig:magicsquare}. The eigenvalue of the third observable is obtained by taking the product~$(-1)^{u_1+u_2}$  since the product of these observables is~$I$.
In other words, the list of 
eigenvalues is given by
\begin{center}
\begin{tabular}{c|ccc|c}
$\alpha=1$ &  $X_{A_1}X_{A_2}$ & $Y_{A_1}Z_{A_2}$ & $Z_{A_1}Y_{A_2}$ & $\left((-1)^{u_1},(-1)^{u_2},(-1)^{u_1+u_2}\right)=:f_1(u_1,u_2)$\\
\hline
$\alpha=2$ &        	$Y_{A_1}Y_{A_2}$ & $\mspace{5mu}Z_{A_1}X_{A_2}$ & $X_{A_1}Z_{A_2}$ &$\left((-1)^{u_1+u_2},(-1)^{u_1},(-1)^{u_2}\right)=:f_2(u_1,u_2)$\\
\hline
$\alpha=3$ & $Z_{A_1}Z_{A_2}$ & $X_{A_1}Y_{A_2}$ & $\mspace{4mu}Y_{A_1}X_{A_2}$        &$\left((-1)^{u_2},(-1)^{u_1+u_2},(-1)^{u_1}\right)=:f_3(u_1,u_2)$
 \end{tabular}
\end{center}
Here we defined the functions~$\{f_\alpha\}_{\alpha\in \{1,2,3\}}$ in such a way that $f_\alpha(u_1,u_2)$ is the triple of eigenvalues associated with the row~$\alpha$ in Fig.~\ref{fig:magicsquare}.

Similarly, the measurement results $(v_1,v_2)\in \{0,1\}^2$ of Bob correspond to 
eigenvalues $(-1)^{v_1}$ and $(-1)^{v_2}$ of the observables~$\tilde{V}^\dagger_\beta (Z\otimes I) \tilde{V}_\beta$ and 
$\tilde{V}^\dagger_\beta (I\otimes Z) \tilde{V}_\beta$, where 
\begin{align}
\tilde{V}_\beta &=(H\otimes I)\CNOT (I\otimes V_\beta)\ ,
\end{align}
that is,
\begin{center}
\begin{tabular}{l|c|cc}
$\beta$ & $V_\beta$ & $(-1)^{v_1}$ eigenvalue of & $(-1)^{v_2}$ eigenvalue of\\
\hline
1 & I & $X_{B_1}X_{B_2}$ & $Z_{B_1}Z_{B_2}$\\
2 & $R_ZR_Y^\dagger$ & $-X_{B_1}Y_{B_2}$ & $Z_{B_1}X_{B_2}$\\
3 & $R_X^\dagger R_Y$ & $X_{B_1}Z_{B_2}$ & $-Z_{B_1}Y_{B_2}$
\end{tabular}
\end{center}
Similarly as before, we can define functions~$\{g_\beta\}_{\beta\in\{1,2,3\}}$ such that~$g_\beta(v_1,v_2)$ corresponds to a triple of eigenvalues of commuting operators whose product is~$-I$: We have 
\begin{center}
\begin{tabular}{c|ccc|c}
$\beta=1$ &  $X_{B_1}X_{B_2}$ & $Y_{B_1}Y_{B_2}$ & $Z_{B_1}Z_{B_2}$ & $\left((-1)^{v_1},-(-1)^{v_1+v_2},(-1)^{v_2}\right)=:g_1(v_1,v_2)$\\
\hline
$\beta=2$ &        	$-Y_{B_1}Z_{B_2}$ & $\mspace{5mu}Z_{B_1}X_{B_2}$ & $-X_{B_1}Y_{B_2}$ &$\left(-(-1)^{v_1+v_2},(-1)^{v_1},(-1)^{v_2}\right)=:g_2(v_1,v_2)$\\
\hline
$\beta=3$ & $-Z_{B_1}Y_{B_2}$ & $X_{B_1}Z_{B_2}$ & $\mspace{4mu}-Y_{B_1}X_{B_2}$        &$\left((-1)^{v_2},(-1)^{v_2},-(-1)^{v_1+v_2}\right)=:g_3(v_1,v_2)$
 \end{tabular}
\end{center}
It is clear from these definitions that~$(x,y)$ computed according to~\eqref{eq:xyoutputpostprocessed} satisfy 
the conditions $x_1x_2x_3=1$ and $y_1y_2y_3=-1$.

To check the remaining condition~$y_\alpha x_\beta=1$, observe that for any~$\beta\in \{1,2,3\}$, the three relevant observables corresponds to the $\beta$-th column of the magic square in Fig.~\ref{fig:magicsquare} up to signs depending on the number (parity) of Pauli-$Y$-operators, and the fact that these act on $B_1B_2$ instead of~$A_1A_2$. In our quantum strategy, the state being measured is~$\Phi_{A_1B_1}\otimes\Phi_{A_2B_2}$. Because the maximally entangled state~$\Phi$ satisfies~
\begin{align}
(\Gamma\otimes I)\Phi=(I\otimes \Gamma^T)\Phi\qquad\textrm{ for any }\qquad \Gamma\in\mathsf{Mat}_{2\times 2}(\mathbb{C})\  ,\label{eq:maximallyentangledstatetransofrmation}
\end{align} and $Y^T=-Y$ is antisymmetric (whereas $X^T=X$ and $Z^T=Z$ are symmetric), it follows that the eigenvalues computed from~$(v_1,v_2)$ using~$g_\beta$ actually correspond to the operators in the~$\beta$-th column in Fig.~\ref{fig:magicsquare} (i.e., the operators act on~$A_1A_2$ and there are no signs). The claim follows from this.
\end{proof}
Let us state the definition of the game~$\cG$ more explicitly. For inputs~$(\alpha,\beta)\in \{1,2,3\}^2$, the task is to output~$((u_1,u_2),(v_1,v_2))\in \{0,1\}^2\times\{0,1\}^2$ such that the outcome probability 
\begin{align}
p_{\alpha,\beta}(u_1,u_2,v_1,v_2)&=
\left|(\langle u_1,u_2|\otimes \langle v_1,v_2|) (\tilde{U}^{A_1A_2}_\alpha\otimes \tilde{V}^{B_1B_2}_\beta) (|\Phi_{A_1B_1}\rangle\otimes |\Phi_{A_2B_2}\rangle)\right|^2\ \label{eq:outcomeprobabilityu1u2v1v2}
\end{align}
is non-zero.

\begin{figure}
\centering
\begin{tikzpicture}

\node[] (img1) at (0,0){
\begin{quantikz}
&\ctrl{1}&\gate{H}& \meter{}&\rstick[1]{$\hspace{-15pt}=s_1$} \\
&\targ{} &\qw &\meter{}&\rstick[1]{$\hspace{-15pt}=s_2$}
\end{quantikz}};

\node at (3.25,0) {$\equiv$};

\node[] (img2) at (6,0.15){
\begin{quantikz}
&\gate{X^{s_2}Z^{s_1}}& \qw \makereversebit[45][black]{$\Phi$} \\[0.5cm]
 &\qw & \qw
\end{quantikz}};
\end{tikzpicture}
\caption{Bell measurement: Instead of labeling outcomes by~$(s_1,s_2)\in \{0,1\}^2$, we may equivalently use $P\in \{I,X,Y,Z\}$. The outcome~$(s_1,s_2)$ corresponds to the Pauli $X^{s_2}Z^{s_1}$. Overall signs can be ignored.
\label{fig:bellmeasurement}}
\end{figure}
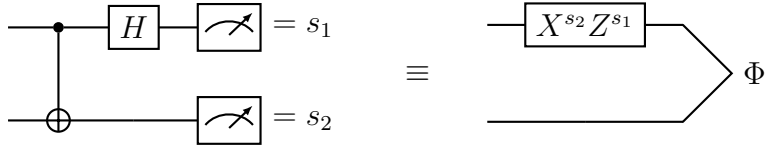

In the following, we will rephrase the game~$\cG$ slightly in order to simplify the notation. To this end, observe that the steps~\eqref{it:bellmeasurementalice} and~\eqref{it:bellmeasurementbob} realize Bell measurements: We have
\begin{align}
(\bra{s_1}\otimes\bra{s_2})\CNOT (H\otimes I)&\propto \bra{\Phi}(X^{s_2} Z^{s_1}\otimes I)\qquad\textrm{ for all }\qquad (s_1,s_2)\in \{0,1\}^2\ ,
\end{align}
see Fig.~\ref{fig:bellmeasurement}. That is,  every outcome~$(s_1,s_2)$ is uniquely associated with a Pauli operator~$P\propto X^{s_2}Z^{s_1}$. The overall phase of $P$ is irrelevant here and we may without loss of generality assume that~$P\in\{I,X,Y,Z\}=:\pauli$. 
In other words, we are performing a von Neumann measurement with operators~$\{\proj{\Phi_P}\}_{P\in\pauli}$, where $\Phi_P:=(P\otimes I)\Phi$ for $P\in \pauli$ defines the Bell basis. 

In the following, we directly use the set~$\pauli$ to label measurement outcomes instead of using two bits for each Pauli operator. For the game~$\cG$, this means that instead of asking for an output of the form~$((u_1,u_2),(v_1,v_2))\in \{0,1\}^2\times\{0,1\}^2$, we seek~$(P,Q)\in \pauli^2$ (where $P$ is associated with $(u_1,u_2)$ and $Q$ with~$(v_1,v_2)$, respectively). The corresponding output probability~\eqref{eq:outcomeprobabilityu1u2v1v2} then takes the form
\begin{align}
p_{\alpha,\beta}(P,Q)&=\left|(\langle\Phi_P|_{A_1A_2} \langle \Phi_Q|_{B_1B_2})
(U_\alpha^{A_1}\otimes I_{A_2}\otimes I_{B_1}\otimes V_\beta^{B_2}) \left(\ket{\Phi}_{A_1B_1}\ket{\Phi}_{A_2B_2}\right)\right|^2
\end{align}
for $(P,Q)\in \pauli^2$.  A diagrammatic representation of the probability~$p_{\alpha,\beta}(P,Q)$ is given in Fig.~\ref{fig:outcomeprobability}.
It shows that the probability of interest can be written as 
\begin{align}
p_{\alpha,\beta}(P,Q)=\frac{1}{16}\left|\tr(U_\alpha P V_\beta Q)\right|^2\label{eq:convenientexpressionalphabeta}
\end{align}
Formally, this follows from the substitution rule~\eqref{eq:maximallyentangledstatetransofrmation} and the fact that $P^T\in \{\pm P\}$ for every $P\in \pauli$.
\begin{center}
\begin{figure}
\centering
\includegraphics[width=0.9\textwidth]{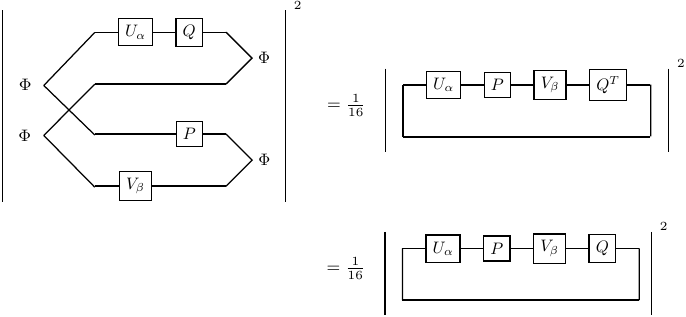}
\caption{Probability~$p_{\alpha,\beta}(P,Q)$
 of the outcome~$(P,Q)\in\pauli^2$ for the circuit~\eqref{fig:magicsquarecircuit} when expressing Bell measurements in terms of Pauli outcomes.
\label{fig:outcomeprobability}}
\end{figure}
\end{center}

In the following, let~$\cliff$ denote the $1$-qubit Clifford group  modulo global phases. Such phases are irrelevant in our considerations. (As for Paulis, we often represent elements of~$\cliff$ by bitstrings but leave this implicit to avoid clutter. Since~$|\cliff|=24$, $5$~bits are sufficient.)
Then we can formulate the following corollary, which will be an essential tool in our analysis.
\begin{corollary}\label{cor:mainMS}
Let $F:\{1,2,3\}\rightarrow\pauli$ and $G:\{1,2,3\}\rightarrow\pauli$ be arbitrary functions.
Then there exists at least one pair~$(\alpha,\beta)\in \{1,2,3\}^2$ such that
\begin{align}
\tr(U_\alpha F(\alpha)V_\beta G(\beta))&=0\ .
\end{align}
\end{corollary}
\begin{proof}
For the sake of contradiction, suppose that
\begin{align}
\tr(U_\alpha F(\alpha)V_\beta G(\beta))&\neq 0\qquad\textrm{ for all }\qquad (\alpha,\beta)\in \{1,2,3\}^2\ .\label{eq:contradictionassumption}
\end{align}
Consider the classical, local (non-communicating) strategy where Alice, on input~$\alpha$, outputs~$F(\alpha)$, and Bob, on input~$\beta$, outputs~$G(\beta)$. By assumption~\eqref{eq:contradictionassumption} and expression~\eqref{eq:convenientexpressionalphabeta}, we then get $p_{\alpha,\beta}(F(\alpha),G(\beta))=0$ for all $(\alpha,\beta)\in \{1,2,3\}^2$. That is, this strategy succeeds with certainty, contradicting Lemma~\ref{lem:classicalwinningG}. 
\end{proof}

We generalize the previous corollary to uniformly random Clifford inputs.
\begin{lemma}\label{lem:mainCliff}
Let $F':\cliff \rightarrow\pauli$ and $G':\cliff \rightarrow\pauli$ be arbitrary functions. Then
\begin{align}
\Pr_{(U,V)\in \cliff^2}\left[\tr(U F'(U)V G'(V))=0\right]\ge \frac{1}{36}  \ .\label{eq:lowerboundclifftwofpgp}
\end{align}
\end{lemma}
\begin{proof}
Define the set 
\begin{align}
\Gamma:=\left\{(U,V)\in\cliff^2\ |\ \tr\left(UF'(U)VG'(V)\right)=0\right\}\ .
\end{align}
By Corollary~\ref{cor:mainMS}, there is at least one pair~$(\alpha,\beta)\in \{1,2,3\}^2$ such that $(U_\alpha,V_\beta)\in\Gamma$, i.e., $|\Gamma|\geq 1$. 
We can improve this lower bound as follows.

Let $(P,Q)\in \pauli^2$ be arbitrary. Define
$F,G:\{1,2,3\}\rightarrow\pauli$ by setting
\begin{align}
F(\alpha):=PF'(U_\alpha P)\qquad\textrm{ and }\qquad
G(\beta):=QG'(V_\beta Q)
\end{align}
for $(\alpha,\beta)\in \{1,2,3\}^2$. By Corollary~\ref{cor:mainMS}, it follows that there exist~$(\alpha,\beta)\in \{1,2,3\}^2$ such that
\begin{align}
\tr(U_\alpha F(\alpha) V_\beta G(\beta))=\tr(U_\alpha P F'(U_\alpha P)V_\beta QG'(V_\beta Q))=0\ .
\end{align}
This shows that $(U_\alpha P,V_\beta Q)\in \Gamma$ for any pair $(P,Q)\in \pauli^2$. It is easy to check that
\begin{align}
|\left\{(U_\alpha P,V_\beta Q)\ |\ (P,Q)\in\pauli^2\right\}|&=16\ ,
\end{align} thus
\begin{align}
    |\Gamma| &\geq 16\ .
\end{align}
The claim~\eqref{eq:lowerboundclifftwofpgp} now follows because
\begin{align}
    \Pr_{(U,V)\in\cliff^2}\left[(U,V)\in\Gamma\right] &=\frac{|\Gamma|}{24^2}\geq \frac{16}{24^2}\ .
\end{align}
\end{proof}

\section{Single-qubit gate teleportation as a  relation problem\label{sec:singlequbitteleportationrelation}}
The computational problem we use to establish complexity-theoretic separations is a relation (also called search) problem. 
Such a problem is defined by a subset~$R\subset \cI\times\cO$ of input/output pairs.  For a given input~$i\in\cI$, the task is to output $o\in\cO$ such that $(i,o)\in R$. We call any $o\in\cO$ with this property a valid solution to the input instance~$i\in\cI$.

In our case, the relation is  indexed by its size~$n\in\mathbb{N}$. We denote it by~$R_{\Utelep_n}$ for reasons that will be clarified below. The associated input set is the set
\begin{align}
\cliff^n &=\left\{(C_0,\ldots,C_{n-1})\ |\ C_j\in \cliff\textrm{ for any }j\in \mathbb{Z}_n\right\}
\end{align} 
of $n$-tuples of single-qubit Clifford group elements. Here and below we use indices $j\in \mathbb{Z}_n$ instead of $[n]=\{1,\ldots,n\}$ to emphasize a certain cyclic symmetry of the problem. Similarly, the associated output set of the problem is
the set of $n$-tuples
\begin{align}
\pauli^n &=\left\{(P_0,\ldots,P_{n-1})\ |\ P_j\in \pauli\textrm{ for any }j\in \mathbb{Z}_n\right\}
\end{align}
of single-qubit Pauli elements.  The relation $R_{\Utelep_n}\subset\cliff^n\times\pauli^n$ is now defined as follows: We have 
\begin{align}
\left((C_0,\ldots,C_{n-1}),(P_0,\ldots,P_{n-1})\right)\in R_{\Utelep_n}\qquad\textrm{ if and only if }\qquad
\tr(P_{n-1}C_{n-1}\cdots P_0C_0)\neq 0\ .
\end{align}

Equivalently, we may say that the computational problem is to output, for a given $n$-tuple~$(C_0,\ldots,C_{n-1})$ of single-qubit Clifford unitaries, a sequence~$(P_{0},\ldots,P_{n-1})$ of Pauli operators such that the quantity
\begin{align}
p(P_0,\ldots,P_{n-1}|C_0,\ldots,C_{n-1})&:=4^{-n}\left|\tr(P_{n-1}C_{n-1}\cdots P_0C_0)\right|^2\label{eq:outcomeprobabilitygateteleportation} 
\end{align}
is non-zero. According to this reformulation and simple algebra, the computational problem given by the relation~$R_{\Utelep_n}$ is that of ``possibilistically'' simulating the output distribution of the gate-teleportation circuit~$\Utelep_n$ given in Fig.~\ref{fig:gateteleportationcircuit}. Indeed, Eq.~\eqref{eq:outcomeprobabilitygateteleportation} is the outcome distribution of~$\Utelep_n$ when the state prepared by~$\Utelep_n$ with input~$(C_0,\ldots,C_{n-1})$ is measured using Bell measurements. This means that~$R_{\Utelep_n}$ is the problem of producing, for a given input~$(C_0,\ldots,C_{n-1})$, an output~$(P_0,\ldots,P_{n-1})$ that occurs with non-zero probability in the output distribution of the circuit~$\Utelep_n$.  For this reason, we will refer to the problem defined by~$R_{\Utelep_n}$ as the possibilistic simulation of the single-qubit gate-teleportation circuit or single-qubit gate-teleportation problem. We previously introduced this computational problem in Ref.~\cite{teleppaper}.

Because the circuit~$\Utelep_n$ is a constant-depth quantum circuit with nearest-neighbor gates arranged on a circle, we can observe the following:
\begin{theorem}[Relation problem $R_{\Utelep_n}$ is contained in $1D$-local $\QNC^0$]\label{thm:mainquantumcircuit}
There is a classically controlled Clifford circuit~$\Utelep_n$ which is $1D$-local and of constant depth, and which solves the relation problem~$R_{\Utelep_n}$ with certainty for any input~$C=(C_0,\ldots,C_{n-1})\in\cliff^n$, i.e., it produces $P=(P_0,\ldots,P_{n-1})\in\pauli^n$ such that $(C,P)\in R_{\Utelep_n}$.
\end{theorem}

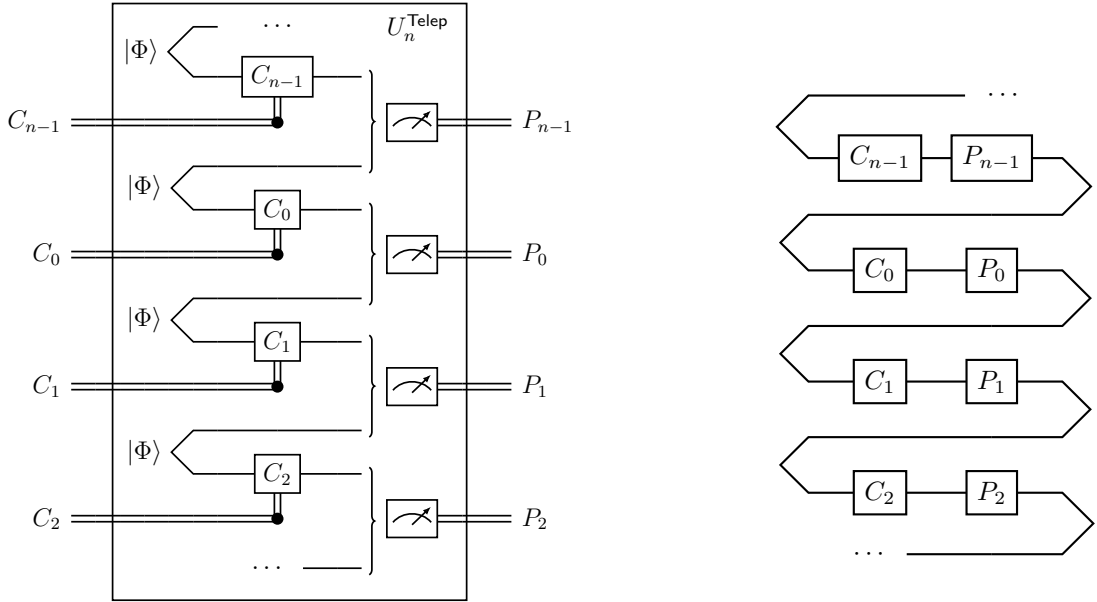
\begin{figure}[h]
\centering
\begin{subfigure}[t]{0.48\textwidth}
         \centering
         \resizebox{1.05\textwidth}{!}{
\begin{tikzpicture}[]
\centering
\draw[black,thick] (-2.5,-4.75) rectangle (3.3,5.1);
\node[] (txt) at (2.5,4.7){$\Utelep_n$};
\node[] (img1) at (0,0){
\begin{quantikz}[row sep=0.4cm,column sep=0.4cm,transparent]
				&&&&& \makeebit[-45][black]{$\ket{\Phi}$} & \qw &  \mathbf{\cdots} \\
				&&&&&                                                 & \qw & \gate{C_{n-1}} &\qw & \qw\rstick[wires=3]{} \\[-0.3cm]
				\lstick{\qquad $C_{n-1}$}&\cw&\cw&\cw&\cw&\cw & \cw &\cwbend{-1} &                           & &|[meter]| &\cw&\cw&\cw\rstick{$P_{n-1}$}\\
				&&&&& \makeebit[-45][black]{$\ket{\Phi}$} & \qw & \qw  &   \qw        &\qw \\
				&&&&&                                                 & \qw & \gate{C_0} &\qw & \qw\rstick[wires=3]{} \\[-0.3cm]
				\lstick{\qquad $C_0$}&\cw&\cw&\cw&\cw&\cw & \cw &\cwbend{-1} &                           & &|[meter]| &\cw&\cw&\cw\rstick{$P_0$}\\
				&&&&& \makeebit[-45][black]{$\ket{\Phi}$} & \qw & \qw  &   \qw        &\qw \\
				&&&&&                                                 & \qw & \gate{C_1}  &\qw & \qw\rstick[wires=3]{} \\[-0.3cm]
				\lstick{\qquad $C_1$}&\cw&\cw&\cw& \cw&\cw&\cw &\cwbend{-1} &                           & &|[meter]| &\cw&\cw&\cw\rstick{$P_1$}\\
				&&&&& \makeebit[-45][black]{$\ket{\Phi}$} & \qw & \qw  &   \qw        &\qw \\
				&&&&&                                                 & \qw & \gate{C_2}  &\qw & \qw\rstick[wires=3]{} \\[-0.3cm]
				\lstick{\qquad $C_2$}&\cw&\cw&\cw& \cw&\cw&\cw &\cwbend{-1} &                           & &|[meter]| &\cw&\cw&\cw\rstick{$P_2$}\\
				&&&& & & &\mathbf{\cdots~~} & \qw       &\qw \\
			\end{quantikz}};
\end{tikzpicture}}
\caption{The gate-teleportation circuit~$\Utelep_n$ is a classically controlled Clifford circuit. It takes as input $n$~Clifford group elements~$(C_0,\ldots,C_{n-1})$. Each of these is  applied to half of maximally entangled state~$\Phi$ (i.e., this constitutes a classically controlled single-qubit Clifford gate.)
If Bell measurements are performed on pairs of qubits (shifted by one), the output~$(P_0,\ldots,P_{n-1})$ is an $n$-tuple of Paulis jointly distributed according to Eq.~\eqref{eq:outcomeprobabilitygateteleportation}. \label{fig:gateteleportationcircuit}}
\end{subfigure}\hfill
\begin{subfigure}[t]{0.48\textwidth}
         \centering
\begin{tikzpicture}[]
\node[] (img1) at (0,0){
  \begin{quantikz}[row sep=0.45cm,column sep=0.4cm,transparent,font=\footnotesize]
          \makeebit[-45][black]{} & \qw & \qw \mathbf{~~\cdots}&  \\
          \qw & \gate{C_{n-1}} & \gate{P_{n-1}}  &\qw \makeebit[45][black]{}\\
	 \makeebit[-45][black]{} & \qw & \qw & \qw  \\
	 \qw & \gate{C_0} & \gate{P_0}  &\qw \makeebit[45][black]{}\\
	 \makeebit[-45][black]{} & \qw & \qw & \qw  \\
	 \qw & \gate{C_1} & \gate{P_1}  &\qw \makeebit[45][black]{}\\
	 \makeebit[-45][black]{} & \qw & \qw & \qw  \\
	 \qw & \gate{C_2} & \gate{P_2}  &\qw \makeebit[45][black]{}\\
     &\mathbf{\cdots~~} & \qw  & \qw \\ 
			\end{quantikz}
   };
\end{tikzpicture}
\caption{Tensor-network representation of the gate-teleportation circuit~$\Utelep_n$. The Clifford gates $\{C_i\}$ represent the input to the circuit, while the Pauli gates $\{P_i\}$ represent the output.}
\end{subfigure}
\caption{Single-qubit gate teleportation circuit.}
\end{figure}

\section{Quantum advantage against $\NC^0$\label{sec:NC0}}
In the following, we consider the hardness of the possibilistically simulating the gate-teleportation circuit by classical circuits.
Without loss of generality, we can restrict to deterministic circuits. A classical circuit~$\cC_n$ for the problem $\telepR$ can then be seen as a function~$\cC_n:\cliff^n\rightarrow\pauli^n$. 

\subsection{Average-case hardness of~$\telepR$ for $\NC^0$-circuit}
In this section, we show a lower bound on the circuit depth of any classical circuit with bounded fan-in gates solving $\telepR$ with high probability on average. 
\begin{theorem}\label{thm:cteleportation}
Let $\cC_n:\cliff^n\rightarrow\pauli^n$
be a circuit with gates of fan-in bounded by a constant~$K$, such that
\begin{align}
    \Pr_{C\in \cliff^n}\left[(C,\cC_n(C))\in\telepR\right]> \frac{35}{36}\ .
\end{align}
Then the circuit depth of~$\cC_n$ satisfies
\begin{align}
\cdepth(\cC_n)\in \Omega(\log n)\ .
 \end{align}
\end{theorem}
\noindent An immediate implication of Theorem~\ref{thm:cteleportation} is the following corollary. 
It  constitutes an extension  of the worst-case result proven by the present authors in Ref.~\cite{teleppaper} to the average case. We note, however, that~\cite{teleppaper} uses different proof techniques, relying on a reduction to known circuit-complexity lower bounds for computing parity. 
\begin{corollary}
No $\NC^0$ circuit can solve the
relation problem~$R_{\Utelep_n}$ of possibilistic simulation of the single-qubit gate-teleportation circuit
for a uniformly chosen instance with average probability greater than $35/36$.
In particular, the relation problem~$R_{\Utelep_n}$ separates the (relational) complexity classes~$\NC^0$ and $\QNC^0$ of constant-depth classical and constant-depth quantum circuits, i.e., we have~$\QNC^0\not\subseteq \NC^0$.
\end{corollary}
Theorem~\ref{thm:cteleportation} is a consequence of a no-signaling argument similar in spirit to those used in~\cite{BGK,BGKT}. The main idea is to show that any classical, limited-depth circuit achieving an average success probability above some threshold value implies a local strategy for a certain generalization of the magic-square game captured in Lemma~\ref{lem:mainCliff}.

Our proof relies on non-signaling properties of $\NC^0$ circuits. Let us formalize the non-signaling property. A circuit~$\cC_n:\cliff^n\rightarrow\pauli^n$  as considered in this section has $n$~inputs, where the $r$-th input~$C_r$ takes a single-qubit Clifford~$C_r\in \cliff$. Similarly, there are $n$~outputs, with the $s$-th output labeled as $P_s$ (outputting a Pauli~$P_s\in\pauli$)\footnote{For convenience, we will slightly abuse the notation and use inputs and outputs for the input nodes and output nodes of the circuit.}.
Given an input~$C_r$, $r\in \mathbb{Z}_n$ to the circuit, we define the forward light cone~$\cL^{\rightarrow}_{\cC_n}(C_r)$ of~$C_r$ as the set of outputs whose associated output depends non-trivially on the input~$C_r$. That is, the set $\cL^{\rightarrow}_{\cC_n}(C_r)$ contains every output~$P_s$, $s\in\mathbb{Z}_n$ with the following property: There are two $n$-tuples~$C,C'\in \cliff^n$ of 
Cliffords such that $C$ and $C'$ agree in each  entry~$m\in\mathbb{Z}_n\backslash \{r\}$, yet
lead to two outputs~$\cC_n(C),\cC_n(C')\in\pauli^n$ differing in the entry~$s\in\mathbb{Z}_n$, i.e., $\cC_n(C)_s\neq \cC_n(C')_s$. Similarly, we can define the backward light cone $\cL_{\cC_n}^\leftarrow(P_s)$ of output $P_s$ as the set of inputs that non-trivially affect output $P_s$. That is, the set $\cL_{\cC_n}^\leftarrow(P_s)$ contains every input $C_r, r\in\mathbb{Z}_n$ such that $P_s\in\cL^\rightarrow_{C_n}(C_r)$.

The relevant no-signaling property can now be concisely stated in these terms, as follows. We say that inputs $C_j$ and $C_k$ have non-intersecting light cones or are non-signaling if and only if
\begin{align}
\cL^\rightarrow_{\cC_n}(C_j)\cap \cL^\rightarrow_{\cC_n}(C_k)&=\emptyset\ .
\end{align}
This property implies
that the output Paulis can be partitioned into three disjoint sets namely: 
\begin{enumerate}[(i)]
\item
output Paulis that depend on~$C_j$ (indexed by $P_s\in \cL^\rightarrow_{\cC_n}(C_j)$),
\item
output Paulis that depend on~$C_k$ 
(indexed by $P_s\in \cL^\rightarrow_{\cC_n}(C_k)$), and 
\item
 output Paulis that do not depend on either~$C_j$ or $C_k$ (indexed by
 $P_s\in \left(\cL^\rightarrow_{\cC_n}(C_j)
 \cup\cL^\rightarrow_{\cC_n}(C_k)\right)^c)$.
\end{enumerate}
(Here $A^c$ denotes the complement of a subset~$A\subset\mathbb{Z}_n$.) 
Importantly, there are no output entries depending simultaneously on both~$(C_j,C_k)$ if the inputs  $C_j$ and $C_k$ are non-signaling.

    The following statement shows that any limited-depth circuit with bounded fan-in gates  has at least two inputs~$C_j$, $C_k$ that are non-signaling. We refer to~\cite{BGKT} for a proof, where an essentially analogous setup is considered.
\begin{lemma}[{\cite[Lemma~7]{BGKT}}]\label{lem:basicprobabilistic}
There is a constant~$c>0$ such that the following holds.
Let $\cC_n:\cliff^n\rightarrow\pauli^n$
be a circuit with gates of fan-in bounded by a constant~$K$. If
\begin{align}
    \cdepth(\cC_n)\leq c\frac{\log n}{\log K}\ ,
\end{align}
then there are two inputs~$C_j,C_k$ with $(j<k)$ that are non-signaling.
\end{lemma}

\definecolor{myred}{RGB}{204,51,17}
\definecolor{myblue}{RGB}{0,119,187}
\definecolor{mygrey}{RGB}{187,187,187}

\tikzset{bluepauli/.style={
        draw,
        thick,
        circle,
        fill=myblue!50,
        minimum size=1cm }}
\tikzset{bluecliff/.style={
        draw,
        thick,
        rectangle,
        fill=myblue!70,
        minimum size=1cm }}
\tikzset{greypauli/.style={
        draw,
        thick,
        circle,
        fill=mygrey!50,
        minimum size=1cm }}
\tikzset{redcliff/.style={
        draw,
        thick,
        rectangle,
        fill=myred!70,
        minimum size=1cm }}
\tikzset{redpauli/.style={
        draw,
        thick,
        circle,
        fill=myred!50,
        minimum size=1cm }}
\tikzset{greycliff/.style={
        draw,
        thick,
        rectangle,
        fill=white,
        minimum size=1cm
 }}
 \tikzset{invisiblepauli/.style={
        draw,
        thick,
        circle,
        opacity=0,
        minimum size=1cm }}

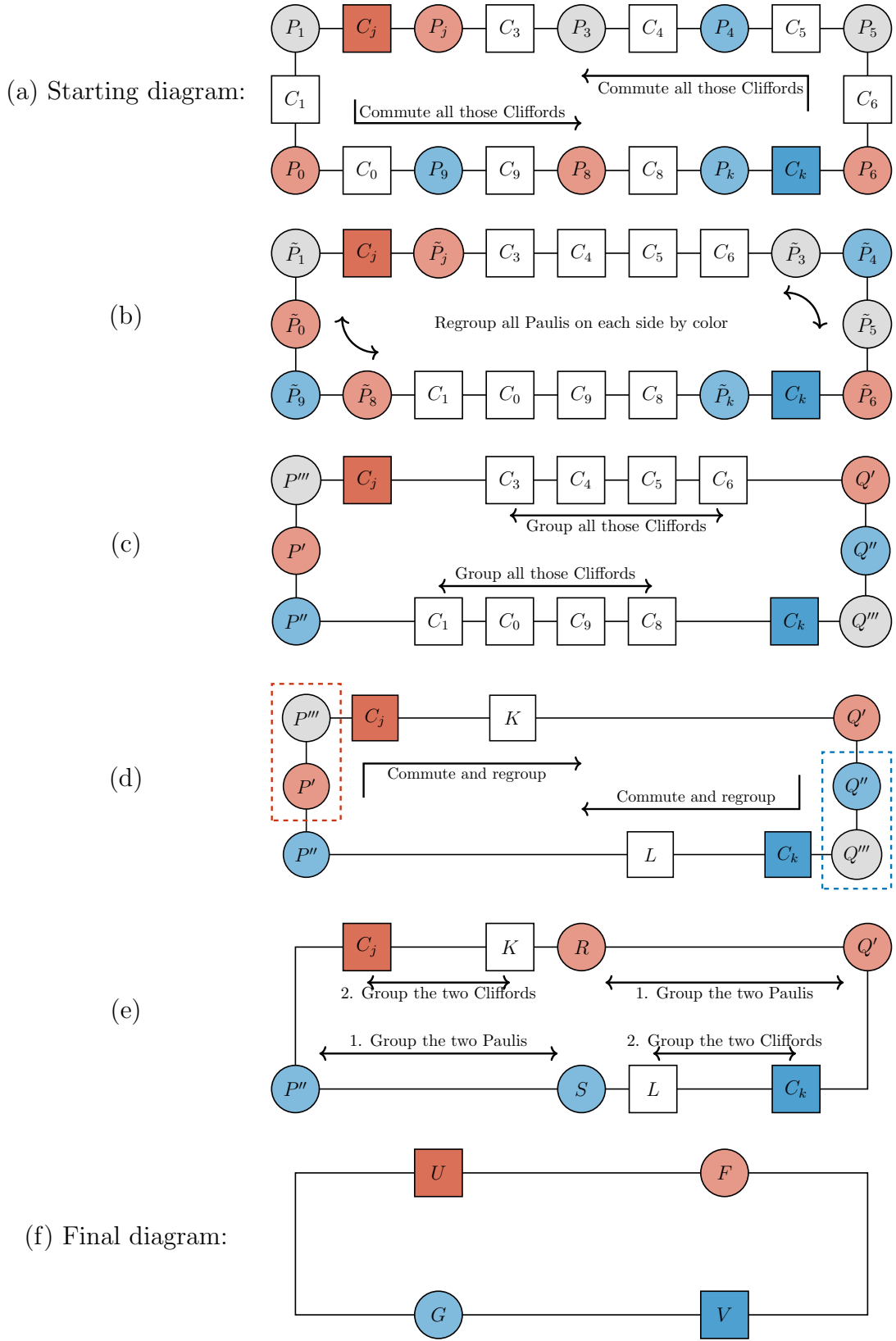
\begin{figure}
\centering
\begin{subfigure}{\textwidth}
\centering
\makebox[3cm]{(a) Starting diagram:}\qquad\raisebox{-0.5\height}{\resizebox{0.58\textwidth}{!}{\begin{tikzpicture}
\centering
\draw[black,thick] (0,0) rectangle (12,3);
\node[redpauli] at (0,0) {$P_0$};
\node[greycliff] at (1.5,0) {$C_0$};
\node[bluepauli] at (3,0) {$P_9$};
\node[greycliff] at (4.5,0) {$C_9$};
\node[redpauli] at (6,0) {$P_8$};
\node[greycliff] at (7.5,0) {$C_8$};
\node[bluepauli] at (9,0) {$P_{k}$};
\node[bluecliff] at (10.5,0) {$C_{k}$};
\node[redpauli] at (12,0) {$P_6$};
\node[greycliff] at (12,1.5) {$C_6$};
\node[greypauli] at (12,3) {$P_5$};
\node[greycliff] at (10.5,3) {$C_5$};
\node[bluepauli] at (9,3) {$P_4$};
\node[greycliff] at (7.5,3) {$C_4$};
\node[greypauli] at (6,3) {$P_3$};
\node[greycliff] at (4.5,3) {$C_3$};
\node[redpauli] at (3,3) {$P_j$};
\node[redcliff] at (1.5,3) {$C_j$};
\node[greypauli] at (0,3) {$P_1$};
\node[greycliff] at (0,1.5) {$C_1$};

 \draw[->, very thick] (1.25,1.5) -- (1.25,1) -- (6,1);
 \draw[->, very thick] (12-1.25,1.25) -- (12-1.25,2) -- (12-6,2);
  \node[align=center] at (3.5,1.25) {\footnotesize Commute all those Cliffords};
   \node[align=center] at (12-3.5,1.75) {\footnotesize Commute all those Cliffords};
 
\end{tikzpicture}}}
\end{subfigure}\vspace{0.5cm}

\begin{subfigure}{\textwidth}
\centering
\makebox[3cm]{(b)}\qquad\raisebox{-0.5\height}{\resizebox{0.58\textwidth}{!}{\begin{tikzpicture}
\centering
\draw[black,thick] (0,0) rectangle (12,3);

\node[redpauli] at (1.5,0) {$\tilde{P}_8$};
\node[bluepauli] at (0,0) {$\tilde{P}_9$};
\node[redpauli] at (0,1.5) {$\tilde{P}_0$};
\node[greypauli] at (0,3) {$\tilde{P}_1$};

\node[greycliff] at (3,0) {$C_1$};
\node[greycliff] at (4.5,0) {$C_0$};
\node[greycliff] at (6,0) {$C_9$};
\node[greycliff] at (7.5,0) {$C_8$};

\node[bluepauli] at (9,0) {$\tilde{P}_{k}$};
\node[bluecliff] at (10.5,0) {$C_{k}$};

\node[redpauli] at (12,0) {$\tilde{P}_6$};
\node[greypauli] at (12,1.5) {$\tilde{P}_5$};
\node[bluepauli] at (12,3) {$\tilde{P}_4$};
\node[greypauli] at (10.5,3) {$\tilde{P}_3$};

\node[greycliff] at (9,3) {$C_6$};
\node[greycliff] at (7.5,3) {$C_5$};
\node[greycliff] at (6,3) {$C_4$};
\node[greycliff] at (4.5,3) {$C_3$};
\node[redpauli] at (3,3) {$\tilde{P}_j$};
\node[redcliff] at (1.5,3) {$C_j$};

 \node[align=center] at (6,1.5) {\footnotesize Regroup all Paulis on each side by color};
 
  \draw[very thick,<->] ({10+0.7*cos(70)} , {1.5+0.7*sin(70)}) arc (95:-5:0.7);
    \draw[very thick,<->] ({1.5+0.7*cos(70)} , {0.25+0.7*sin(70)}) arc (275:175:0.7);

\end{tikzpicture}}}

\end{subfigure}\vspace{0.5cm}
~

\begin{subfigure}{\textwidth}
\centering
\makebox[3cm]{(c)}\qquad\raisebox{-0.5\height}{\resizebox{0.58\textwidth}{!}{\begin{tikzpicture}
\centering
\draw[black,thick] (0,0) rectangle (12,3);

\node[bluepauli] at (0,0) {$P''$};
\node[redpauli] at (0,1.5) {$P'$};
\node[greypauli] at (0,3) {$P'''$};

\node[greycliff] at (3,0) {$C_1$};
\node[greycliff] at (4.5,0) {$C_0$};
\node[greycliff] at (6,0) {$C_9$};
\node[greycliff] at (7.5,0) {$C_8$};

\node[bluecliff] at (10.5,0) {$C_k$};

\node[greypauli] at (12,0) {$Q'''$};
\node[bluepauli] at (12,1.5) {$Q''$};
\node[redpauli] at (12,3) {$Q'$};

\node[greycliff] at (9,3) {$C_6$};
\node[greycliff] at (7.5,3) {$C_5$};
\node[greycliff] at (6,3) {$C_4$};
\node[greycliff] at (4.5,3) {$C_3$};
\node[redcliff] at (1.5,3) {$C_j$};

 \node[align=center] at (6.75,2) {\footnotesize Group all those Cliffords};
  \node[align=center] at (5.25,1) {\footnotesize Group all those Cliffords};
 \draw[<->, very thick] (4.5,2.25) -- (9,2.25);
  \draw[<->, very thick] (3,0.75) -- (7.5,0.75);

\end{tikzpicture}}}

\end{subfigure}\vspace{0.5cm}
~

\begin{subfigure}{\textwidth}
\centering
\makebox[3cm]{(d)}\qquad\raisebox{-0.5\height}{\resizebox{0.58\textwidth}{!}{\begin{tikzpicture}
\centering
\draw[black,thick] (0,0) rectangle (12,3);

\node[bluepauli] at (0,0) {$P''$};
\node[redpauli] at (0,1.5) {$P'$};
\node[greypauli] at (0,3) {$P'''$};

\node[greycliff] at (7.5,0) {$L$};
\node[bluecliff] at (10.5,0) {$C_{k}$};

\node[greypauli] at (12,0) {$Q'''$};
\node[bluepauli] at (12,1.5) {$Q''$};
\node[redpauli] at (12,3) {$Q'$};

\node[greycliff] at (4.5,3) {$K$};
\node[redcliff] at (1.5,3) {$C_j$};

\draw[dashed, very thick, myred] (-0.75,0.75) rectangle (0.75,3.75);
\draw[->, very thick] (1.25,1.25) -- (1.25,2) -- (6,2);
  \node[align=center] at (3.5,1.75) {\footnotesize Commute and regroup};

  \draw[dashed, very thick, myblue] (-0.75+12,0.75-1.5) rectangle (0.75+12,3.75-1.5);
\draw[->, very thick] (12-1.25,3-1.25) -- (12-1.25,3-2) -- (12-6,3-2);
  \node[align=center] at (12-3.5,3-1.75) {\footnotesize Commute and regroup};

\end{tikzpicture}}}

\end{subfigure}\vspace{0.5cm}
 
\begin{subfigure}{\textwidth}
\centering
\makebox[3cm]{(e)}\qquad\raisebox{-0.5\height}{\resizebox{0.58\textwidth}{!}{\begin{tikzpicture}
\centering
\draw[black,thick] (0,0) rectangle (12,3);

\node[bluepauli] at (0,0) {$P''$};
\node[redpauli] at (6,3) {$R$};

\node[greycliff] at (7.5,0) {$L$};
\node[bluecliff] at (10.5,0) {$C_{k}$};

\node[bluepauli] at (6,0) {$S$};
\node[redpauli] at (12,3) {$Q'$};

\node[greycliff] at (4.5,3) {$K$};
\node[redcliff] at (1.5,3) {$C_j$};
 \node[align=center] at (9,2) {\footnotesize 1. Group the two Paulis};
  \node[align=center] at (3,1) {\footnotesize 1. Group the two Paulis};
 \draw[<->, very thick] (6.5,2.25) -- (11.5,2.25);
  \draw[<->, very thick] (0.5,0.75) -- (5.5,0.75);

     \node[align=center] at (9,1) {\footnotesize 2. Group the two Cliffords};
   \node[align=center] at (3,2) {\footnotesize 2. Group the two Cliffords};
  \draw[<->, very thick] (7.5,0.75) -- (10.5,0.75);
   \draw[<->, very thick] (1.5,2.25) -- (4.5,2.25);

\end{tikzpicture}}}

\end{subfigure}\vspace{0.5cm}
~

\begin{subfigure}{\textwidth}
\centering
\makebox[3cm]{(f) Final diagram:}\qquad\raisebox{-0.5\height}{\resizebox{0.58\textwidth}{!}{\begin{tikzpicture}
\centering
\draw[black,thick] (0,0) rectangle (12,3);

\node[redpauli] at (9,3) {${F}$};
\node[bluecliff] at (9,0) {$V$};
\node[bluepauli] at (3,0) {${G}$};
\node[redcliff] at (3,3) {$U$};
\node[invisiblepauli] at (0,0) {};
\node[invisiblepauli] at (12,3) {};

\end{tikzpicture}}}

\end{subfigure}

\caption{Tensor-network representation illustrating the various commutation relations used in the proof of Theorem~\ref{thm:cteleportation}. We give an example of dimension $n=10$, with $j=2$ and $k=7$. The input Cliffords are represented by squares, while the output Paulis are represented by circles of various colours: (i) the Paulis in red belong to the forward light cone of the Clifford $C_j$ in red. (ii) the Paulis in blue belong to the forward light cone of the Clifford $C_k$ in blue. (iii) The Paulis in grey are in neither of those light cones.}\label{fig:commutations}
\end{figure}

\begin{proof}[Proof of Theorem~\ref{thm:cteleportation}]
Fig.~\ref{fig:commutations} illustrates the main steps of the proof.
We first rewrite the condition~$(C,P)\in\telepR$ in a form that is convenient for our analysis. For this purpose, let $j,k\in \mathbb{Z}_n$ with $j<k$ be arbitrary but fixed. Observe that $(C,P)\in\telepR$ if and only if
\begin{align}
A:=\left|\tr\left(\left(\prod_{r=k+1}^{n-1}P_{r}C_{r}\right) P_{k} C_k \left(\prod_{s=j+1}^{k-1} P_{s} C_{s}\right)P_{j} C_j \left(\prod_{t=0}^{j-1} P_{t}C_{t}\right) \right)\right|>0\ . \label{eq:traceteleprestriction}
\end{align}
Using the cyclicity of trace, we can rewrite the quantity~$A$ as
\begin{align}
A=\left|\tr\left( C_j \left(\prod_{s\in[k+1,\ldots,n-1,0,1,\ldots, j-1]} P_{s}C_{s}\right) P_{k} C_k\left(\prod_{t\in[j+1,\ldots, k-1]} P_{t} C_{t}\right)P_{j} \right)\right|\ ,\label{eq:traceteleprestrictionstep1}
\end{align}
where we used the product indexed by elements of $\mathbb{Z}_n$:
\begin{align}
\prod_{s\in[k+1,\ldots,n-1,0,1,\ldots, j-1]} P_{s}C_{s} &= (P_{j-1}C_{j-1})\cdots (P_{n-1}C_{n-1})(P_0C_0)(P_1C_1)\cdots(P_{k+1}C_{k+1})\\
\prod_{t\in[j+1,\ldots, k-1]} P_{t} C_{t} &=(P_{k-1}C_{k-1})\cdots(P_{j+1}C_{j+1}) \ . \label{eq:cyclicproduct}
\end{align}
We will simply write $[k+1,\ldots, j-1]$ instead of $[k+1,\ldots,n-1,0,1,\ldots, j-1]$ when using the product indexed by elements of $\mathbb{Z}_n$.

We now use the fact that Clifford group normalizes the Pauli group. This allows us to commute all Paulis in the two products 
in Eq.~\eqref{eq:traceteleprestrictionstep1} to the right. That is, we define Pauli 
\begin{align}
    \tilde{P}_s&:=\left(\prod_{t\in[k+1,\ldots, s]} C_t \right)^\dagger P_s \left( \prod_{r\in [k+1,\ldots, s]}C_r \right)\qquad\textrm{ for all }\qquad s\in[k+1,\ldots,j-1]\label{eq:paulicommutingkp1jm1}\\
    \tilde{P}_t&:=\left(\prod_{t\in[j+1,\ldots, s]} C_t \right)^\dagger P_s \left( \prod_{r\in [j+1,\ldots, s]}C_r \right)\qquad\textrm{ for all }\qquad  t\in[j+1,\ldots, k-1]\label{eq:paulicommutingjp1km1}\\
    \tilde{P}_j&:=P_j\\
    \tilde{P}_k&:=P_k\ \label{eq:paulitildelast} 
\end{align}
such that \begin{align}
    \prod_{s\in[k+1,\ldots, j-1]} P_{s}C_{s} & = \left(\prod_{u\in[k+1,\ldots, j-1]}C_u\right)\left(\prod_{v\in[k+1,\ldots, j-1]} \tilde{P}_v\right) \label{eq:productUp}\\
    \prod_{t\in[j+1,\ldots, k-1]} P_{t} C_{t} &= \left(\prod_{u\in[j+1,\ldots, k-1]}C_u\right)\left(\prod_{v\in[j+1,\ldots, k-1]} \tilde{P}_v\right) \ . \label{eq:productVp}
\end{align}
Substituting~\eqref{eq:productUp} and \eqref{eq:productVp} into  Eq.~\eqref{eq:traceteleprestrictionstep1}, we obtain the necessary and sufficient condition
\begin{align}
A&=\left|\tr\left( 
\left(\prod_{u\in [k+1,\ldots,j]} C_u\right)
\left(\prod_{s\in[k,\ldots, j-1]} \tilde{P}_{s}\right) 
\left(\prod_{u\in [j+1,\ldots,k]}C_u\right)
\left(\prod_{t\in[j,\ldots, k-1]} \tilde{P}_{t} \right)\right)\right|>0\ \label{eq:traceteleprestrictionstep2}
\end{align}
for $(C,P)\in\telepR$. For later use, we note that each~$\tilde{P}_s$, $s\in\mathbb{Z}_n$
defined by Eqs.~\eqref{eq:paulicommutingkp1jm1}--~\eqref{eq:paulitildelast} is obtained by commuting Paulis past Clifford operators belonging to the $(n-2)$-tuple
\begin{align}
C_{j,k}^c:=\left(C_0,\ldots, C_{j-1},C_{j+1},\ldots,C_{k-1},C_{k+1},\ldots,C_{n-1}\right) \ ,\label{eq:excltuple}
\end{align}
that is, we have 
\begin{align}
\tilde{P}_s&=\tilde{P}_s(P_s,C_{j,k}^c)\label{eq:functionaldependencetildeps}
    \end{align}
    for a certain function~$\tilde{P}_s$, for every~$s\in\mathbb{Z}_n$. In particular, given~$P_s$, $\tilde{P}_s$ is determined by~\eqref{eq:excltuple} and there is no dependence on~$(C_j,C_k)$.

To prove Theorem~\ref{thm:cteleportation}, suppose for the sake of contradiction that there is a circuit~$\cC_n$ with gates of bounded fan-in~$K$ satisfying
\begin{align}
    \cdepth(\cC_n) &\leq c\frac{\log n}{\log K}\ ,
\end{align}
where $c$ is the constant from Lemma~\ref{lem:basicprobabilistic}, and achieving an average success probability
\begin{align}
    \Pr_{C\in\cliff^n}    \left[(C,\cC_n(C))\in\telepR\right] >\frac{35}{36}\ .\label{eq:assumptioncontradiction}
\end{align}
By
Lemma~\ref{lem:basicprobabilistic}, there are $j<k$ such that
\begin{align}
\cL^{\rightarrow}_{\cC_n}(C_j)\cap \cL^{\rightarrow}_{\cC_n}(C_k)&=\emptyset\ .\label{eq:cncjckprop}
\end{align} As discussed  before Lemma~\ref{lem:nosignalingbiasmain}, 
this means that for every $s\in \mathbb{Z}_n$, the outputs have the following property:  Either
\begin{enumerate}[(i)]
\item\label{it:firstcasedependence}
$P_s\in \cL_{\cC_n}^\rightarrow(C_j)$, 
\item\label{it:secondcasedependence}
 $P_s\in \cL_{\cC_n}^\rightarrow(C_k)$,  or 
 \item\label{it:thirdcasedependence}
 $P_s\not\in \left(\cL_{\cC_n}^\rightarrow(C_j)\cup \cL_{\cC_n}^\rightarrow(C_k)\right)$. 
 \end{enumerate}
This implies the following:
 \begin{enumerate}[(i)]
 \item \label{it:casea}
In case~\eqref{it:firstcasedependence}, the Pauli~$P_s$ depends non-trivially on the input~$C_k$, i.e., there is a function $P_s:\cliff\times(\cliff)^{n-2}\rightarrow\pauli$ such that~$P_s=P_s(C_k,C_{j,k}^c)$.

Combined with Eq.~\eqref{eq:functionaldependencetildeps}, this implies that there is a function
$\tilde{P}_s:\cliff\times(\cliff)^{n-2}\rightarrow\pauli$ such that~$\tilde{P}_s=\tilde{P}_s(C_k,C_{j,k}^c)$.
\item 
In case~\eqref{it:secondcasedependence}, the Pauli~$P_s$ depends non-trivially on~$C_j$, i.e., there is a function $P_s:\cliff\times(\cliff)^{n-2}\rightarrow\pauli$ and write~$P_s(C_j,C_{j,k}^c)$.

Combined with Eq.~\eqref{eq:functionaldependencetildeps}, this implies that there is a function
$\tilde{P}_s:\cliff\times(\cliff)^{n-2}\rightarrow\pauli$ such that~$\tilde{P}_s=\tilde{P}_s(C_j,C_{j,k}^c)$.
\item\label{it:casec}
In  case~\eqref{it:thirdcasedependence}, the Pauli $P_s$ depends only on $C_{j,k}^c$, i.e., there is a function $P_s:(\cliff)^{n-2}\rightarrow\pauli$ such that $P_s=P_s(C_{j,k}^c)$.

Because of Eq.~\eqref{eq:functionaldependencetildeps}, there is a function
$\tilde{P}_s:(\cliff)^{n-2}\rightarrow\pauli$ such that~$\tilde{P}_s=\tilde{P}_s(C_{j,k}^c)$.
\end{enumerate}
Now consider the products of Pauli operators appearing in Eq.~\eqref{eq:traceteleprestrictionstep2}. We can reorder these according to the three cases~\eqref{it:casea}--\eqref{it:casec} above since Pauli operators either commute or anticommute. This gives factorizations of the form
\begin{align}
\prod_{t\in [j,\ldots,k-1]} \tilde{P}_t(C) &=\pm P''_{C^c_{j,k}}(C_k)P'_{C^c_{j,k}}(C_j)P'''_{C^c_{j,k}}
\ \textrm{ where }\ \begin{matrix}
P'_{C^c_{j,k}}(C_j)&:=&\prod_{t\in [j,\ldots,k-1]\cap \cL^{\rightarrow}_{\cC_n}(C_j)} \tilde{P}_t(C_j,C_{j,k}^c)\\
P''_{C^c_{j,k}}(C_k)&:=&\prod_{t\in [j,\ldots,k-1]\cap \cL^{\rightarrow}_{\cC_n}(C_k)} \tilde{P}_t(C_k,C_{j,k}^c)\\
P'''_{C^c_{j,k}}&:=&\prod_{t\in [j,\ldots,k-1]\backslash (\cL^{\rightarrow}_{\cC_n}(C_j)\cup\cL^{\rightarrow}_{\cC_n}(C_k))} \tilde{P}_t(C_{j,k}^c)\ 
\end{matrix}
\end{align}
and
\begin{align}
\prod_{s\in [k,\ldots,j-1]} \tilde{P}_s(C) &=\pm Q'_{C^c_{j,k}}(C_j)Q''_{C^c_{j,k}}(C_k)Q'''_{C^c_{j,k}}
\ \textrm{where }\ 
\begin{matrix}
Q'_{C^c_{j,k}}(C_j)&:=&\prod_{s\in [k,\ldots,j-1]\cap \cL^{\rightarrow}_{\cC_n}(C_j)} \tilde{P}_s(C_j,C_{j,k}^c)\\
Q''_{C^c_{j,k}}(C_k)&:=&\prod_{s\in [k,\ldots,j-1]\cap \cL^{\rightarrow}_{\cC_n}(C_k)} \tilde{P}_s(C_k,C_{j,k}^c)\\
Q'''_{C^c_{j,k}}&:=&\prod_{s\in [k,\ldots,j-1]\backslash (\cL^{\rightarrow}_{\cC_n}(C_j)\cup\cL^{\rightarrow}_{\cC_n}(C_k))} \tilde{P}_s(C_{j,k}^c)\ ,
\end{matrix}
\end{align}
where we choose a slightly different order of the three factors in the two products for convenience only. Inserting this into~\eqref{eq:traceteleprestrictionstep2}, 
and using curly brackets to emphasize associativity we obtain
\begin{align}
A&=\left|\tr\left( 
\left(\prod_{u\in [k+1,\ldots,j]} C_u\right)
\left(\prod_{s\in[k,\ldots, j-1]} \tilde{P}_{s}\right) 
\left(\prod_{u\in [j+1,\ldots,k]}C_u\right)
\left(\prod_{t\in[j,\ldots, k-1]} \tilde{P}_{t} \right)\right)\right|\\
&=
\left|\tr\left(
\left\{C_j K_{C_{j,k}^c}\right\}
Q'_{C_{j,k}^c}(C_j)Q''_{C_{j,k}^c}(C_k)Q'''_{C_{j,k}^c}
\left\{C_k L_{C_{j,k}^c}\right\}
P''_{C_{j,k}^c}(C_k)P'_{C_{j,k}^c}(C_j)P'''_{C_{j,k}^c}\right)\right|\ ,\label{eq:pauliopAdependence}
\end{align}
where we defined the Clifford operators
\begin{align}
    K_{C_{j,k}^c} := \prod_{u\in [k+1,\ldots,j-1]} C_u\qquad\textrm{ and }\qquad 
    L_{C_{j,k}^c} := \prod_{u\in [j+1,\ldots,k-1]} C_u\ .
\end{align}
In the following, we will temporarily suppress the dependence on the tuple~$C_{j,k}^c$, omitting this subscript.  Again using the fact that Paulis commute or anticommute and the cyclicity of the trace, we rearrange~\eqref{eq:pauliopAdependence} by moving all factors depending on $C_j$ and $C_k$, respectively, next to each other. We obtain
\begin{align}
A&=\left|\tr\left(
\left\{C_j K\right\}
Q'(C_j)Q''(C_k)Q'''
\left\{C_k L\right\}
P''(C_k)\left\{P'(C_j)P'''\right\}\right)\right| \\
&=\left|\tr\left(
\left\{P'(C_j)P'''\right\}
\left\{C_j K\right\}
Q'(C_j)
Q''(C_k)Q'''
\left\{C_k L\right\}
P''(C_k)\right)\right|\quad\textrm{by the cyclicity of the trace}\\
&=
\left|\tr\left(
\left\{P'(C_j)P'''\right\}
\left\{C_j K\right\}
Q'(C_j)
\left\{Q''(C_k)Q'''\right\}
\left\{C_k L\right\}
P''(C_k)\right)\right|\\
&=\left|\tr\left(
\left\{C_j K\right\}
R(C_j)
Q'(C_j)
\left\{C_k L\right\}
S(C_k)
P''(C_k)
\right)\right|\label{eq:Aexpressionm}
\end{align}
where we introduced the Paulis
\begin{align}
R(C_j)&= \left\{C_j K\right\}^\dagger \left\{P'(C_j)P'''\right\}\left\{C_j K\right\}\\
S(C_k)&= \left\{C_k L\right\}^\dagger \left\{Q''(C_k)Q'''\right\}\left\{C_k L\right\}
\end{align}
by commuting~$P'(C_j)P'''$ to the right of the factor $C_jK$, and $Q''(C_k)Q'''$ to the right of the factor~$C_jL$.
Finally, we can use associativity and define Pauli operators $\tilde{F}$ and $\tilde{G}$ by
Eq.~\eqref{eq:Aexpressionm} as follows
\begin{align}
A&=\left|\tr\left(
C_j K\underbrace{
R(C_j)
Q'(C_j)}_{=:\tilde{F}}
C_k L\underbrace{S(C_k)
P''(C_k)}_{=:\tilde{G}}
\right)\right|\\
&=\left|\tr\left(C_jK \tilde{F}C_kL\tilde{G}\right)\right|
\end{align}
Let us briefly summarize what we have obtained recalling the different dependencies. We have, for each
$(n-2)$~tuple $C^c_{j,k}\in \cliff^{n-2}$, two functions
\begin{align}
\begin{matrix}
\tilde{F}_{C^c_{j,k}}:&\cliff& \rightarrow &\pauli\\
&C_j & \mapsto &  \pm R_{C_{j,k}^c}(C_j)Q'_{C_{j,k}^c}(C_j)
\end{matrix}
\end{align}
and
\begin{align}    
\begin{matrix}
\tilde{G}_{C^c_{j,k}}:&\cliff& \rightarrow &\pauli\\
&C_k & \mapsto & \pm S_{C_{j,k}^c}(C_k)P''_{C_{j,k}^c}(C_k)
\end{matrix}
\end{align}
where the signs are chosen appropriately such that the expression belongs to~$\pauli$, as well as Clifford operators~$K_{C^c_{j,k}},L_{C^c_{j,k}}\in\cliff$ such that~$(C,\cC_n(C))\in\telepR$ if and only if
\begin{align}
\left|
\tr\left(C_j
K_{C^c_{j,k}}\tilde{F}_{C^c_{j,k}}(C_j)C_kL_{C^c_{j,k}}\tilde{G}_{C^c_{j,k}}(C_k)
\right)
\right|>0\ . \label{eq:cjkdefinitionone}
\end{align}

Now consider the probability~$p_{|C^c_{j,k}}$  that, conditioned on a fixed choice of~$C^c_{j,k}\in\cliff^{n-2}$, 
the circuit~$\cC_n$ produces a valid solution to
a randomly chosen instance from the set of $n$-tuples of Cliffords wich coincide with~$C^c_{j,k}$ in entries~$\ell\not\in\{j,k\}$, and are arbitrary otherwise (i.e., in the entries $j$ and $k$). By~\eqref{eq:cjkdefinitionone}, the probability~$p_{|C^c_{j,k}}$ is given by
\begin{align}
    p_{|C^c_{j,k}}&=\frac{1}{|
    \cliff|^2}\left|\left\{
    (U,V)\in\cliff^2\ :\ 
\left|
\tr\left(U
K_{C^c_{j,k}}\tilde{F}_{C^c_{j,k}}(U)VL_{C^c_{j,k}}\tilde{G}_{C^c_{j,k}}(V)
\right)
\right|>0    
    \right\}\right|\\
&=\frac{1}{|
    \cliff|^2}\left|\left\{
    (U,V)\in\cliff^2\ :\ 
\left|
\tr\left(U
\tilde{F}_{C^c_{j,k}}(UK_{C^c_{j,k}}^{-1})V\tilde{G}_{C^c_{j,k}}(VL^{-1}_{C^c_{j,k}})
\right)
\right|>0    
    \right\}\right|\ ,\label{eq:lowerboundclffm}
        \end{align}
        where the last line is obtained by variable substition.

By definition of $p_{|C^c_{j,k}}$
and assumption~\eqref{eq:assumptioncontradiction}, we have 
\begin{align}
\ExpE_{C^c_{j,k}}[p_{|C^c_{j,k}}]&=
        \Pr_{C\in\cliff^n}    \left[(C,\cC_n(C))\in\telepR\right] >\frac{35}{36}
        \end{align}
        for a uniformly random choice of the $(n-2)$-tuple $C^c_{j,k}\in\cliff^{n-2}$. This shows that there is at least one choice of $C^c_{j,k}$ such that $p_{|C^c_{j,k}}>\frac{35}{36}$. In particular, for such a choice, we obtain
        \begin{align}
            \frac{1}{|
    \cliff|^2}\left|\left\{
    (U,V)\in\cliff^2\ |\ 
\left|
\tr\left(U
\tilde{F}_{C^c_{j,k}}(UK_{C^c_{j,k}}^{-1})V\tilde{G}_{C^c_{j,k}}(VL^{-1}_{C^c_{j,k}})
\right)
\right|>0    
    \right\}\right| &> \frac{35}{36}\label{eq:fivexiseq}
        \end{align}
        by Eq.~\eqref{eq:lowerboundclffm}. Defining the functions
        \begin{align}
            \begin{matrix}
                F : & \cliff & \rightarrow & \pauli\\
                & U & \mapsto & \tilde{F}_{C^c_{j,k}}(UK_{C^c_{j,k}}^{-1})
            \end{matrix}\qquad\textrm{ and }\qquad 
                        \begin{matrix}
                G : & \cliff & \rightarrow & \pauli\\
                & V & \mapsto & \tilde{G}_{C^c_{j,k}}(VL_{C^c_{j,k}}^{-1})\ ,
            \end{matrix}
        \end{align}
Eq.~\eqref{eq:fivexiseq} becomes
\begin{align}
    \Pr_{(U,V)\in\cliff^2}\left[
    \left|\tr(UF(U)VG(V)|\right|>0\right]>\frac{35}{36}\ .
\end{align}
The existence of two such functions~$F,G$ contradicts Lemma~\ref{lem:mainCliff}. This concludes the proof of the theorem.
\end{proof}

\section{Towards proving an advantage against~$\AC^0$: variants of the quantum advantage against $\NC^0$}
In this section, we prove a few auxiliary results extending the  results for $\NC^0$-circuits  from the previous section. These are needed to show the $\AC^0$-circuit size lower bound we prove by restriction based methods~\cite{HastadThesis,Beame1994ASL} in Section~\ref{sec:AC0}.
In particular, to prepare for the proof that the possibilistic simulation of the gate-teleportation circuit is not solvable by an $\AC^0$ circuit of polynomial size, we show here that the gate-teleportation problem cannot be solved by an $\NC^0$ circuit even when some of the inputs of the $\NC^0$ circuit are fixed.

\subsection{Average-case hardness of~$\restrictedtelepR$ for $\NC^0$-circuit\label{sec:averagecasehardnessrestrict}}

\begin{figure}
\centering
\begin{tikzpicture}
\centering
 \foreach \y in {0.5,...,6.5} {
\draw[black, fill=black] (1,\y) circle (1pt);
\draw[black, fill=black] (2,\y) circle (1pt);
 }

\draw[double, double distance=1pt] (1,0.5) -- (0,0.5);
\draw[double, double distance=1pt] (1,1.5) -- (0,1.5);
\draw[double, double distance=1pt] (1,2.5) -- (-3,2.5);
\draw[double, double distance=1pt] (1,3.5) -- (0,3.5);
\draw[double, double distance=1pt] (1,4.5) -- (0,4.5);
\draw[double, double distance=1pt] (1,5.5) -- (-3,5.5);
\draw[double, double distance=1pt] (1,6.5) -- (0,6.5);

\draw[double, double distance=1pt] (2,0.5) -- (4,0.5);
\draw[double, double distance=1pt] (2,1.5) -- (4,1.5);
\draw[double, double distance=1pt] (2,2.5) -- (4,2.5);
\draw[double, double distance=1pt] (2,3.5) -- (4,3.5);
\draw[double, double distance=1pt] (2,4.5) -- (4,4.5);
\draw[double, double distance=1pt] (2,5.5) -- (4,5.5);
\draw[double, double distance=1pt] (2,6.5) -- (4,6.5);
\node at (1,7.25) {\textbf{in}};
\node at (2,7.25) {\textbf{out}};
\node at (2.5,8.3) {$\cC_n$};

 \foreach \y in {0,...,6} {
\node[anchor=west] at (4,6.5-\y) {$P_{\y}$};
}

\draw[black,thick] (0.25,0) rectangle (2.75,8);

\draw[black,thick] (-2.5,-0.5) rectangle (3.5,9);
\node at (3,9.3) {$\cC_n\restriction_\xi$};

\node[anchor=east] at (0,6.5) { $\xi(0)=C_0$};
\node[rectangle, fill=white, anchor=east] at (0,5.5) { $\xi(1)=*$};
\node[anchor=east] at (0,4.5) { $\xi(2)=C_2$};
\node[anchor=east] at (0,3.5) { $\xi(3)=C_3$};
\node[rectangle, fill=white, anchor=east] at (0,2.5) { $\xi(4)=*$};
\node[anchor=east] at (0,1.5) { $\xi(5)=C_5$};
\node[anchor=east] at (0,0.5) { $\xi(6)=C_6$};

\node[anchor=east] at (-3,5.5) { $D_0$};
\node[anchor=east] at (-3,2.5) { $D_1$};
\end{tikzpicture}

\caption{Schematic representation of the effect of a restriction $\xi:\mathbb{Z}_n\to\cliff\cup\{*\}$ on the circuit $\cC_n$. The inputs $C_i$ are fixed while the inputs $D_i$ are active; it gives rise to the restricted circuit $\cC_n\restriction_\xi$}.
\label{fig:rr}
\end{figure}
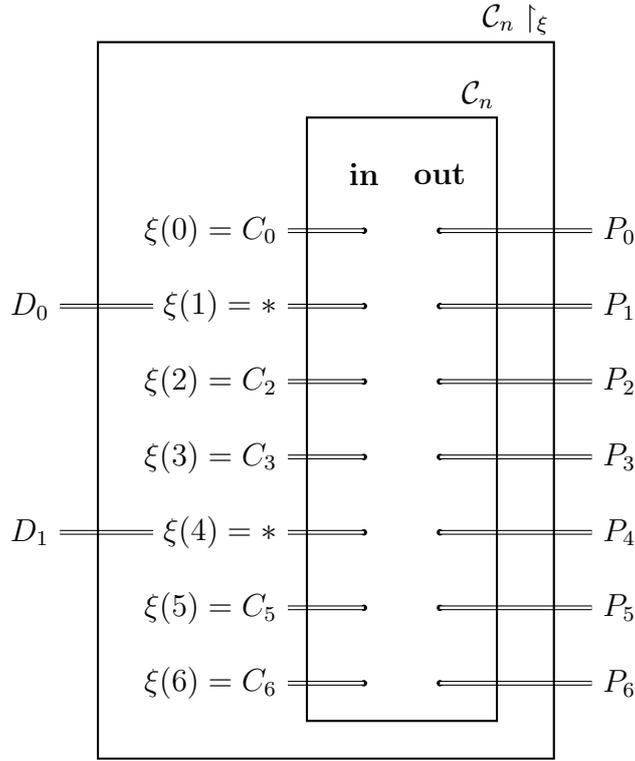

The first strengthening of the result for $\NC^0$-circuits from the previous section concerns the problem of possibilistically simulating gate-teleportation circuit with some fixed inputs.  To formalize this, we define a \emph{restriction} of $n$ Clifford inputs as a function $\xi:\mathbb{Z}_n\rightarrow \cliff\cup\{*\}$ which fixes some Clifford elements and leaves the remaining inputs marked by ``$*$'' active (or free).
Let $m=|\xi^{-1}(\{*\})|$ be the number of active inputs. We use $\xi^{-1}(\{*\})=\{\iota(0)<\iota(1)<\ldots<\iota(m-1)\}\subseteq\mathbb{Z}_n$ to denote the active inputs. The application of restriction $\xi$ to a circuit $\cC$ is defined as
\begin{align}
\begin{matrix}
    \cC_n\restriction_\xi\colon & \cliff^m &\to & \pauli^n\\
     &(x_0,\ldots,x_{m-1}) & \mapsto &\cC(x'_0,\ldots, x'_{n-1})\ ,
\end{matrix}
\end{align}
where for $j\in\mathbb{Z}_m$
\begin{align}
    x'_j=
    \begin{cases}
        x_{\iota^{-1}(j)} & \text{if } \xi(j)=*\\
        \xi(j) & \text{ otherwise} \ .
    \end{cases}
\end{align}
See Fig.~\ref{fig:rr} for an illustration of a restriction applied to a classical circuit for of the single-qubit gate-teleportation problem. We define a restriction of Clifford inputs here, but restrictions can be naturally defined on any type of inputs. In Section~\ref{sec:NC0restrictedproblembitencoding} and Section~\ref{sec:fromac0tonc0} we will use restrictions that fixes bits and blocks of bits, i.e., strings.

We will use restriction to define simulation problem with some inputs fixed.
Let $\xi:\mathbb{Z}_n\to\cliff\cup\{*\}$ be a restriction. We define the restricted problem of possibilistic simulation of gate-teleportation circuit $\restrictedtelepR \subset \cliff^m\times \pauli^n$ as the possibilistic simulation of the quantum circuit $U_n^{\telep}\restriction_{\xi}$. Formally, let $D=(D_0,\ldots, D_{m-1})\in\cliff^m$ and let $P=(P_0,\ldots, P_{n-1})\in \pauli^n$, then
\begin{align}
    (D,P)\in R_{U^\telep_n}\restriction_\xi \qquad \text{ if and only if }\qquad
    (C,P)\in \telepR\ , 
\end{align}
where $C=(C_0,\ldots,C_{n-1})$ is such that for all $\ell\in\mathbb{Z}_n$
\begin{align}
    C_\ell=\begin{cases}
        \xi(\ell) \qquad &\textrm{ if }\qquad \xi(\ell)\not=*\\
        D_{\iota^{-1}(\ell)}\qquad &\textrm{ otherwise }\ .
    \end{cases}
\end{align}

We extend the bound for $\NC^0$-circuits for the problem $\telepR$ given in Theorem~\ref{thm:cteleportation} to the restricted gate-teleportation problem $\restrictedtelepR$ with some inputs fixed.
\begin{theorem}\label{thm:maintheoremrestrictednc0} 
Let $0<\gamma\le 1$ and $1/2<\delta \le 1$ be constants. 
Let $\xi:\mathbb{Z}_n\rightarrow \cliff\cup\{*\}$ be a 
restriction that keeps $m\ge\gamma n^\delta$ inputs active. Then any $\NC^0$-circuit~ $\cC_n:\cliff^m\rightarrow\pauli^n$ solves the restricted gate-teleportation problem $R_{\Utelep_n}\restriction_\xi$ with average probability at most 
\begin{align}
\Pr_{D\in\cliff^m}\left[
\left(D, \cC_n(D)\right)\in R_{\Utelep_n}\restriction_\xi
\right]< \frac{35}{36}\ \label{eq:averagessuccessNC0restrictedproblem}
\end{align}
for an instance $D$ chosen uniformly from the set~$\cliff^m$. 
\end{theorem}
The proof of this theorem follows the proof of Theorem~\ref{thm:cteleportation} and also relies on a light-cone argument. There is one important difference we need to consider. The $\NC^0$-circuit solving the restricted problem has an imbalanced number of inputs and outputs. We show that if such an imbalance is ``reasonable'' the light-cone argument leading to non-signaling inputs still holds. We capture this in the following lemma.

\begin{lemma}\label{lem:nosignalingbiasmain} Let $0<\gamma\le 1$ and $1/2 <\delta\le 1$. Let $\cC_n$ be an $\NC^0$-circuit with constant depth $d$ and bounded fan-in $K$, and with $m\ge\gamma n^\delta$ inputs and $n$ outputs (the inputs and the outputs can be bits or have higher alphabet size).
Then there exist two inputs $j,k\in \mathbb{Z}_n$ with non-intersecting light cones, i.e.,
\begin{align}
    \cL^\rightarrow_{\cC_n} (j)\cap \cL^{\rightarrow}_{\cC_n} (k)=\emptyset\ .
\end{align}
\end{lemma}
\begin{proof}

We first bound the expected number of outputs in the forward light cone of a randomly chosen input.
Let the locality $\ell$ be defined as $\ell:=K^d$ where $d:=\cdepth(\cC_n)$ is the depth of $\cC_n$. Notice that the locality is the maximum number of inputs that can be in a backward light cone $\cL_{\cC_n}^\leftarrow(o)$ of any output $o\in\mathbb{Z}_n$ that is
\begin{align}
    |\cL_{\cC_n}^\leftarrow(o)|\le \ell \qquad \textrm{ for all }\qquad o\in\mathbb{Z}_n \ . \label{eq:localitybound}
\end{align}
Consider a bipartite graph $G=(V_I\cup V_O, E)$, where the set of vertices $V_I=\mathbb{Z}_m$ is the set of inputs and the set $V_O=\mathbb{Z}_n$ of outputs. There is an edge connecting vertices $v_i$ and $v_o$ whenever $v_o\in \cL^\rightarrow (v_i)$.
By~\eqref{eq:localitybound} the total number of edges is bounded by
\begin{align}
    |E|\le \ell\cdot n\ .
\end{align}
This allows us to bound the expected degree of an input node $j \in V_I$ picked uniformly at random. By definition of $G$ this is the expected size of the forward light cone of an input $j$ picked uniformly at random, i.e.,
\begin{align}
   \mathop{\mathbb{E}}_{j} [|\cL^\rightarrow(j)|]=\mathop{\mathbb{E}}_i [{\rm deg}(v_i)] \le \ell n^{1-\delta}\gamma^{-1}\ , \label{eq:expectedsizeforwardlightcone}
\end{align}
where we used that $m=|V_I|\le\gamma n^\delta$ by assumption.
From Eq.~\eqref{eq:expectedsizeforwardlightcone}, we conclude that there exists at least one input $j\in V_I$ for which \begin{align}
    |\cL^\rightarrow(j)|\le \ell n^{1-\delta}\gamma^{-1}\ .
\end{align}
We note that for any output $o$, $|\cL^\leftarrow(o)|\le \ell$, and that therefore for the previously mentioned input $j$ we have
\begin{align}
    |\cL^\leftarrow(\cL^\rightarrow(j))|\le \ell^2 n^{1-\delta}\gamma^{-1}\ .
\end{align}
We conclude the proof by picking any element $k$ from the set
\begin{align}
S=\{k\in V_I: \cL^\rightarrow_{\cC_n} (j)\cap \cL^{\rightarrow}_{\cC_n} (k)=\emptyset\}\ ,\end{align}
and by remarking that it is non-empty for sufficiently large $n$ because
\begin{align}
|S|\ge \gamma n^\delta -\ell^2 n^{1-\delta}\gamma^{-1}  \ . 
\end{align}
\end{proof}

\begin{proof}[Proof of Theorem~\ref{thm:maintheoremrestrictednc0}] 
The argument is almost the same as in the proof of Theorem~\ref{thm:cteleportation}, but with every input instance having some of the Clifford inputs fixed.
By definition of $\restrictedtelepR$ for any instance $D\in\cliff^m$ there is $C=(C_0,\ldots, C_{n-1})\in\cliff^n$ where for $\ell\in\mathbb{Z}_n$
\begin{align}
C_\ell=\begin{cases}
        \xi(\ell) \qquad &\textrm{ if }\qquad \xi(\ell)\not=*\\
        D_{\iota^{-1}(\ell)}\qquad &\textrm{ otherwise }\ ,
\end{cases}
\end{align}
such that $(D,\cC_n(D))\in\restrictedtelepR$ if and only if  $(C,\cC_n(D))\in\telepR$.

Let $\xi^{-1}(\{*\})=\{\iota(0)<\iota(1)<\ldots<\iota(m-1)\}$ be the active inputs of the restriction $\xi$.
Assume $j'<k'\in \mathbb{Z}_m$ are arbitrary but fixed and define $j:=\iota(j')$ and $k:=\iota(k')$.
In the proof of Theorem~\ref{thm:cteleportation} (Eqs.~\eqref{eq:traceteleprestriction} --~\eqref{eq:traceteleprestrictionstep2}) we derived a necessary and sufficient condition for $P\in\pauli^n$ to be the solution of $\telepR$ problem for instance $C$. Following this argument and using the definition of~$\restrictedtelepR$ for an instance~$D$ we similarly obtain 
\begin{align}
\left|\tr\left( 
\left(\prod_{u\in [k+1,\ldots,j]} C_u\right)
\left(\prod_{s\in[k,\ldots, j-1]} \tilde{P}_{s}\right) 
\left(\prod_{u\in [j+1,\ldots,k]}C_u\right)
\left(\prod_{t\in[j,\ldots, k-1]} \tilde{P}_{t} \right)\right)\right|>0\ ,
\end{align}
where $\tilde{P}_s$ and $\tilde{P}_t$ are defined in Eq.~\eqref{eq:paulicommutingkp1jm1}--\eqref{eq:paulitildelast}.

For the sake of contradiction, assume there is an $\NC^0$ circuit $\cC_n$ satisfying
\begin{align}
\Pr_{D\in \cliff^m}\left[
\left(D, \cC_n(D)\right)\in \restrictedtelepR
\right]\ge \frac{35}{36}\ . \label{eq:mainthmrestrictednc0assumption}
\end{align}
By Lemma~\ref{lem:nosignalingbiasmain} we have that any $\NC^0$ circuit $\cC_n$ with at least $m\ge\gamma n^\delta$ inputs, where $0<\gamma\le 1$ and $1/2 <\delta\le 1$, and $n$ outputs necessarily has two inputs $D_{j'}, D_{k'}$ for $j',k'\in\mathbb{Z}_m$ with the non-signalling property. 
By definition, these correspond to $C_j$ and $C_k$ in the instance $C$ of the $\telepR$ problem. By exactly following the light-cone argument below~\eqref{eq:cncjckprop} in the proof of Theorem~\ref{thm:cteleportation}, we obtain the same contradiction with Lemma~\ref{lem:mainCliff}.
\end{proof}

\subsection{From Clifford inputs to input strings}\label{sec:FromCliffordsToBits}
To establish our results for~$\AC^0$, it is more convenient to work with circuits taking bitstrings as input. Since $|\cliff|=24$, 
there is an injective map~$\iota:\cliff\rightarrow\{0,1\}^5$ that we fix throughout the following. Consider a map $\enc:\{0,1\}^5\rightarrow\cliff$ which 
satisfies 
\begin{align}
\begin{matrix}
    \enc:&\{0,1\}^5 & \rightarrow &\cliff\\
     & x & \mapsto & \enc(x):=\begin{cases}
         \iota^{-1}(x)\qquad &\textrm{ if } x\in \iota(\cliff)\\
         C_x &\textrm{ otherwise }
     \end{cases}
    \end{matrix}\label{eq:encodingmapiotacx}
    \end{align}
    for a family~$\{C_x\}_{x\in \{0,1\}^5\backslash \iota(\cliff)}\subset \cliff$ of Clifford group elements. We note such a map is surjective, and completely specified by~$\iota$ and the family~$\{C_x\}_{x\in\{0,1\}^5\backslash \iota(\cliff)}$. Correspondingly, we sometimes write~$\enc=\enc_{\{C_x\}_{x\in\{0,1\}^5\backslash \iota(\cliff)}}$.  We will denote by~$\cE$ the set of such maps (each one parametrized by the associated family of Clifford group elements), and call this the set of valid encodings.

In the following, we consider $n$-tuples $\enc=(\enc^{(0)},\ldots,\enc^{(n-1)})\in\cE^n$ of valid encoding maps. 
Such a tuple defines a map
\begin{align}
\begin{matrix}
\enc:& (\{0,1\}^5)^n & \rightarrow &  \cliff^n\\
    &x=(x^{(0)},\ldots,x^{(n-1)})& \mapsto &  \left(\enc^{(0)}(x^{(0)}),\ldots,\enc^{(n-1)}(x^{(n-1)})\right)\ .
    \end{matrix}
\end{align}
which we again denote by~$\enc$ by a slight abuse of notation.

We can now show that bounds on the average success probability over uniformly random sequences of Cliffords and uniformly random bitstrings encoding Cliffords are related.

\begin{lemma}\label{lem:clifftobitprob}
Let $\cR\subset \cliff^n\times \pauli^n$ be a relation.   Let  $F:\cliff^n\rightarrow\pauli^n$ and $\epsilon>0$ be arbitrary. Then there is an $n$-tuple~$\enc=(\enc^{(0)},\ldots,\enc^{(n-1)})\in\cE^n$ 
 of valid encoding maps such that 
 the composition~$F\circ \enc: (\{0,1\}^5)^n  \rightarrow   \pauli^n$
 satisfies 
\begin{align}
    \Pr_{x\in  (\{0,1\}^5)^n}\left[(\enc(x),(F\circ \enc)(x)))\in\cR\right] &> (1-\epsilon) \Pr_{C\in  \cliff^n}\left[(C,F(C))\in\cR\right] \ .
\end{align}
Here $x\in (\{0,1\}^5)^n$ and~$C\in \cliff^n$, respectively, are chosen uniformly at random. 
\end{lemma}
\begin{proof}
We define a distribution over the set~$\cE^n$ of random~$n$-tuples $\enc$~of valid encoding maps as follows.
Consider random variables
\begin{align}\{C^{(j)}_x\}_{\substack{j\in\mathbb{Z}_n\\
x\in \{0,1\}^5\backslash \iota(\cliff)}}
\end{align} obtained by picking $C^{(j)}_x\in\cliff$ uniformly and independently at random for every $j\in \mathbb{Z}_n$ and every~$x\not\in \iota(\cliff)$. 
Then, for each $j\in \mathbb{Z}_n$, let $\cE^{(j)}=\cE_{\{C^{(j)}_x\}_{x\in \{0,1\}^5\backslash \iota(\cliff)}}$
be the valid encoding map defined by~\eqref{eq:encodingmapiotacx}.
Equivalently, this means that the $n$-tuple~$\enc=\{\enc^{(j)}\}_{j\in\mathbb{Z}_n}$ is a random variable with uniform distribution on~$\cE^n$.

    Now consider the probability taken over the choice of~$\enc$ and the choice of~$x\in (\{0,1\}^5)^n$ (chosen independently and uniformly) that the composition~$F\circ\enc$ provides a valid solution to the relation problem~$\cR$ on instance~$\enc(x)\in\cliff^n$. 
    Because  the induced distribution of the random variable~$\enc(x)$  is the uniform distribution on~$\cliff^n$, it follows that 
    \begin{align}
        \Pr_{\substack{
        \enc\in \cE\\
        x\in (\{0,1\}^5)^n
        }} \left[(\enc(x),(F\circ\enc)(x))\in\cR\right]&=\Pr_{C\in  \cliff^n}\left[(C,F(C))\in\cR \right]\ .\label{eq:identityinequality}
    \end{align}
    The claim then follows immediately from Markov's inequality: By~\eqref{eq:identityinequality} we have 
    \begin{align}
        \ExpE_{\enc} [p(\enc)]&\geq \Pr_{C\in  \cliff^n}\left[(C,F(C))\in\cR\right]\ ,
    \end{align}
    where
    \begin{align}
         p(\enc):=\Pr_{x\in (\{0,1\}^5)^n}\left[(\enc(x),(F\circ \enc)(x)))\in\cR\right]\ .
    \end{align}
Applying Markov's inequality~$\Pr[X\geq a]\leq \frac{\ExpE[X]}{a}$, $a>0$ to the random variable~$X:=1-p(\enc)$ 
and using the abbreviation~$q:=\Pr_{C\in  \cliff^n}\left[(C,F(C))\in\cR\right]$
gives 
\begin{align}
\Pr_{\enc}\left[1-p(\enc)\geq a\right] & \leq \frac{1-q}{a}\ .
\end{align}
With the choice 
\begin{align}
a&=1-(1-\eps)q
\end{align}
we obtain 
\begin{align}
\Pr_{\enc}\left[p(\enc)\leq (1-\eps)q\right] & \leq \frac{1-q}{(1-q)+\eps q}<1
\end{align}
for any $\epsilon \in (0,1)$. 
    This shows the existence of~$\enc\in\cE^n$ with the desired property, that is, $p(\enc)>(1-\eps)q$. 

\end{proof}

\subsection{Average-case hardness of $\restrictedtelepR$ with bit encoded input}\label{sec:NC0restrictedproblembitencoding}
In the previous section, we defined mappings between bitstrings and Cliffords and used blocks of $5$ bits to encode sequences of Clifford gates.
To show our $\AC^0$ circuit size lower bound we rely on a restriction based method which works with circuits having binary inputs and outputs. One of the building blocks is an extension of the $\NC^0$-result for the restricted relation problem $\restrictedtelepR$ to  the setting of circuits with binary inputs and corresponding restrictions fixing individual bits (rather than Clifford group elements). This is what we provide in this section. Specifically, we encode each Clifford by a bitstring~$\{0,1\}^n$ and consider restrictions of the form~$\xi:\mathbb{Z}^n\rightarrow \{0,1\}^5\cup \{*\}$.  We will then define an associated relation problem~$\telepR\restriction_{\enc(\xi)}$.

To show a result for $\NC^0$-circuits with bits (strings) as input we need a bound on the success probability of a bit encoded version of the generalized variant of the magic-square game with random Cliffords.
The following is a bit-string encoded version of Lemma~\ref{lem:mainCliff}.
\begin{lemma}\label{lem:mainCliff5bits}
Let $\enc^{(1)},\enc^{(2)}\in\cE$ be valid encodings. Let $F:\{0,1\}^5 \rightarrow\pauli$ and $G:\{0,1\}^5 \rightarrow\pauli$ be arbitrary functions. Then
\begin{align}
\Pr_{(x^{(1)},x^{(2)})\in (\{0,1\}^5)^2} \left[\tr\left(\enc^{(1)}(x^{(1)})F(x^{(1)})\enc^{(2)}(x^{(2)})G(x^{(2)})\right)=0\right]\ge \frac{1}{81}\ .
\end{align}
\end{lemma}

\begin{proof}
First observe that
\begin{align}
\Pr_{(x^{(1)},x^{(2)})\in (\{0,1\}^5)^2} [(x^{(1)},x^{(2)})\in  (\iota (\cliff) )^2]=\frac{4}{9} \ .
\end{align}
The bound then follows by applying Lemma~\ref{lem:mainCliff}: We have
\begin{align}
&\Pr_{(x^{(1)},x^{(2)})\in (\{0,1\}^5)^2} \left[\tr\left(\enc^{(1)}(x^{(1)})F(x^{(1)})\enc^{(2)}(x^{(2)})G(x^{(2)})\right)=0\right]\\
\ge& \Pr_{(x^{(1)},x^{(2)})\in (\iota (\cliff) )^2} \left[\tr\left(\enc^{(1)}(x^{(1)})F(x^{(1)})\enc^{(2)}(x^{(2)})G(x^{(2)})\right)=0\right] \cdot \frac{4}{9}\\
 \ge&  \frac{1}{36}\cdot \frac{4}{9}= \frac{1}{81} \ .\qedhere
\end{align}
\end{proof}

The following is an ``encoded'' version of Theorem~\ref{thm:maintheoremrestrictednc0}. 
We consider a restriction of the form~$\xi:\mathbb{Z}^n\rightarrow \{0,1\}^5\cup \{*\}$, i.e., it associates either a $5$-tuple of bits or the symbol~$*$ to each~$i\in\mathbb{Z}_n$. We  call~$\xi^{-1}(\{*\})\subset\mathbb{Z}_n$ the set of active inputs. Using an encoding function~$\enc$, such a restriction~$\xi$ gives rise to a restriction which associates a Clifford or a symbol~$*$ to every~$i\in\mathbb{Z}_n$. By a slight abuse of notation, we denote this new restriction as~$\enc(\xi)$. It is defined as 
\begin{align}
\begin{matrix}
\enc(\xi):&\mathbb{Z}_n & \rightarrow & \cliff\cup \{*\}&\\
& i & \mapsto & \enc(\xi)(i)&:=\begin{cases}
        * \qquad &\textrm{ if } \xi(i)=*\\
        \enc(\xi(i)) &\textrm{ otherwise }\ 
    \end{cases}\qquad\textrm{ for every }i\in\mathbb{Z}_n\ .
    \end{matrix}
    \end{align}
We may then consider the associated restricted problem for the restriction~$\enc(\xi)$, that is, the problem~$\telepR\restriction_{\enc(\xi)}$. It is illustrated in Fig.~\ref{fig:teleprestrictionencxi}. We have the following statement:
\begin{figure}
\centering
\centering
\begin{tikzpicture}[scale=0.2]
\centering

\usetikzlibrary{decorations.pathreplacing}

\draw[semithick] (0,0) rectangle (8,33);
\node[] at (4,16) {$\Utelep_n$};

\draw[dashed, myred, thick] (-25,1) rectangle (-1.75,32);
\node[anchor=west,myred] at (-25,32.75) {$\enc(\xi)$};

\foreach \y in {1,5,3,4,6} {
\draw[double, double distance=1pt] (11,5*\y-1) -- (8,5*\y-1) ;
}

\foreach \y in {0,...,3} {
\node[anchor=west] at (11,34-5*\y-5) {$P_{\y}$};
}
\node[anchor=west] at (12,10) {$\vdots$};
\node[anchor=west] at (11,4) {$P_{n-1}$};

\foreach \y in {1,5,3,4,6} {
\draw[double, double distance=1pt] (0,5*\y-1) -- (-7,5*\y-1) ;
}

\foreach \y in {0,...,3} {
\node[anchor=east,rectangle,draw,minimum width=1.8cm,minimum height=0.6cm,fill=white] at (-7,34-5*\y-5) {$\enc^{(\y)}$};
}
\node[anchor=east] at (-10,10) {$\vdots$};
\node[anchor=east,rectangle,draw,minimum width=1.8cm,minimum height=0.6cm,fill=white] at (-7,4) {$\enc^{(n-1)}$};

\foreach \y in {0,...,3} {
\node[] at (-4,35.5-5*\y-5) {$C_{\y}$};
}
\node[] at (-4,5.5) {$C_{n-1}$};

\foreach \y in {1,5,3,4,6} {
\draw[double, double distance=1pt] (-16,5*\y-1) -- (-19,5*\y-1) ;
}

\node[anchor=east] at (-19,29) {$x_{0}$};
\node[anchor=east] at (-19,24) {$*$};
\node[anchor=east] at (-19,19) {$*$};
\node[anchor=east] at (-19,14) {$x_{2}$};
\node[anchor=east] at (-19,4) {$x_{n-1}$};

\draw[double, double distance=1pt] (-28,19) -- (-22,19) ;
\draw[double, double distance=1pt] (-28,24) -- (-22,24) ;

\draw[thick,decoration={brace},decorate,anchor=east]
  (-29,18) -- node {$\times N(\xi)$} (-29,25);

\end{tikzpicture}

\caption{The bit-string encoded version of
the restricted gate-teleportation problem.
Starting from a restriction~$\xi:\mathbb{Z}^n\rightarrow \{0,1\}^5\cup \{*\}$,
the task can be understood as that of producing, on a given input~$x\in (\{0,1\}^5)^m$, 
a Pauli~$P$ such that the pair~$(\widehat{\enc}(x),P)$ appears in the output distribution of the illustrated circuit. Here $\widehat{\enc}\in\cE^m$ is obtained by suitably restricting (according to active inputs of~$\xi$) an $n$-tuple~$\enc\in\cE^n$ of valid encodings of Cliffords.
\label{fig:teleprestrictionencxi}}
\end{figure}
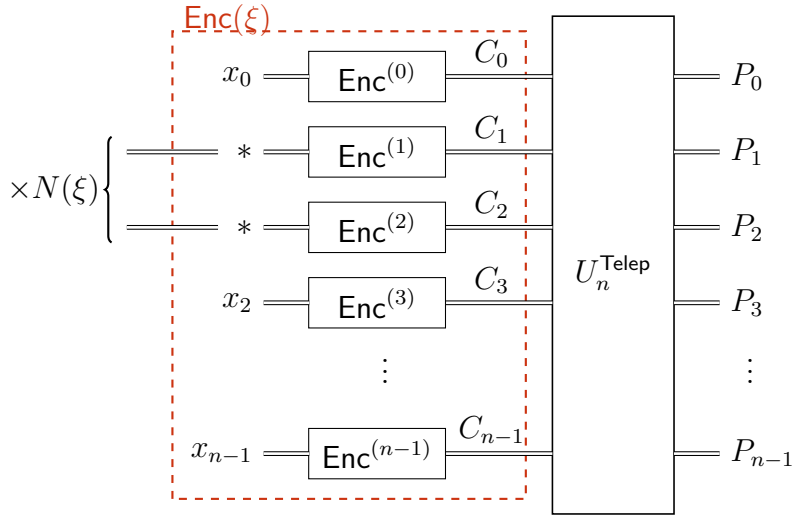

\begin{theorem}\label{thm:maintheoremrestrictednc0bitstring}
Let $0<\gamma\le 1$ and $1/2<\delta \le 1$ be constants. 
Let $\enc=(\enc^{(0)},\ldots,\enc^{(n-1)})\in\cE^n$ be an $n$-tuple of valid encoding maps.
Let $\xi:\mathbb{Z}_n\rightarrow \{0,1\}^5\cup\{*\}$ be a 
restriction that keeps $m\ge\gamma n^\delta$ inputs active. Assume that 
the set of active inputs is~$\xi^{-1}(\{*\})=\{i_0<\ldots < i_{m-1}\}$.
Then any $\NC^0$ circuit~$\cC:(\{0,1\}^5)^m\rightarrow \pauli^n$ satisfies
\begin{align}
\Pr_{x\in (\{0,1\}^5)^m}\left[
(\widehat{\enc}(x),\cC(x))\in \telepR\restriction_{\enc(\xi)}
\right] < \frac{80}{81}
\end{align}
for a uniformly randomly chosen input~$x=(x^{(0)},\ldots,x^{(m-1)})\in (\{0,1\}^5)^m$, where we introduced 
\begin{align}
    \widehat{\enc}(x)=(\enc^{(i_0)}(x^{(0)}),\ldots,\enc^{(i_{m-1})}(x^{(m-1)}))\in\cliff^m\qquad\textrm{ for }\qquad x\in (\{0,1\}^5)^m\ .
    \end{align}
\end{theorem}

\begin{proof} The proof is an immediate consequence of the proof of Theorem~\ref{thm:maintheoremrestrictednc0}. The only difference is that we use forward light cones of blocks of input bits encoding Clifford gates, similar to Lemma~\ref{lem:nosignalingbiasmain}.  
To reach a contradiction, we use use Lemma~\ref{lem:mainCliff5bits} instead of
Lemma~\ref{lem:mainCliff}. It gives a bound on the success probability of the bitstring-encoded variant of the  magic-square game with random Cliffords. 
\end{proof}

\section{Quantum advantage against $\AC^0$\label{sec:AC0}}
One of the most important methods for circuit size lower bounds going beyond the light-cone arguments are the so-called restriction-based methods and in particular switching lemmas~\cite{furst_parity_1984,Ajtai1983,Yao,Hastad86,HastadThesis,Rossman2017AnEP} (for a nice exposition see Beame's article \cite{Beame1994ASL}). Informally, a strong variant of this technique asserts that if we randomly fix all but $O(n^\delta)$ (for $\delta\le 1$) inputs of an $\AC^0$-circuit, then with overwhelming probability, it is possible to fix a few additional bits to ensure that the restricted circuit can be implemented by an $\NC^0$ circuit.

\newcommand*{\block}{\mathsf{block}}
\subsection{Bit-restrictions and block-restrictions}
We will be interested in circuits~$\cC^\block_n:(\{0,1\}^5)^n\rightarrow (\{0,1\}^2)^n$
whose input-strings~$x=(x^{(0)},\ldots,x^{(n-1)})$ consist of $5n$ bits, and are naturally grouped into blocks (strings) of length~$5$: Each five-tuple~$x^{(j)}\in\{0,1\}^5$ represents a Clifford group element~$C_j\in\cliff$, for $j\in\mathbb{Z}_n$. (The output can be interpreted similarly, but this is less important for this section.)

We will often interpret the circuit~$\cC^\block_n$ as a function $\cC_n:\{0,1\}^{n'}\rightarrow \{0,1\}^{m'}$ where $n'=5n$ and $m'=2n$. (Here we omit the superscript~$\block$ to distinguish this interpretation from the case where we block the of input bits into groups of~$5$.)
We may  may then consider a ``standard'' restriction~$\xi:\mathbb{Z}_{n'}\rightarrow\{0,1,*\}$ of the input bits. Applying such a restriction to the circuit~$\cC_n$, i.e., considering the circuit~$\cC_n\restriction_\xi$ (defined analogously as in Section~\ref{sec:averagecasehardnessrestrict} but for individual bits) may, however, not be compatible with the actual block-structure of the input of~$\cC_n$.  In the encoded bit representation of Clifford group elements, a general restriction can mean that individual Clifford inputs are partially fixed (rather than being completely free or completely fixed), i.e., the restriction may mean that we have to consider a subset of~$\cliff$ in some entries. 

Nevertheless, it will be convenient in this section to work with restrictions~$\xi:\mathbb{Z}_{n'}\rightarrow\{0,1,*\}$ of individual bits. Let us call~$\xi$ a {\em block restriction} if it is compatible with the block structure. That is, we say that
$\xi:\mathbb{Z}_{5n}\rightarrow \{0,1,*\}$ a block restriction if and only if 
\begin{align}
    (\xi(5j),\xi(5j+1),\ldots,\xi(5j+4))\in \{0,1\}^5\cup \{(*,*,*,*,*)\}\qquad\textrm{ for every }\qquad j\in\mathbb{Z}_n\ .\label{eq:blockfixingrestrictiondefintion}
\end{align}
Condition~\eqref{eq:blockfixingrestrictiondefintion} expresses the fact that for every block of $5$~input bits, the restriction assigns  a $5$-bit string  or leaves the entire block (i.e., each bit in the block) free (i.e., active). We will interchangeably refer to a block restriction by the function~$\xi:\mathbb{Z}_{5n}\rightarrow\{0,1,*\}$ (satisfying~\eqref{eq:blockfixingrestrictiondefintion}) or the associated function~$\xi^\block:\mathbb{Z}_5\rightarrow \{0,1\}^5\cup \{*^\block\}$. 

Naturally, any block restriction 
$\xi:\mathbb{Z}_{5n}\rightarrow \{0,1,*\}$
determines a function 
\begin{align}
\begin{matrix}
\xi^\block:&\mathbb{Z}_{n}&\rightarrow &\{0,1\}^5\cup \{*^\block\}&\\
 & j & \mapsto & \xi^\block(j):=&
 \begin{cases}
  *^\block &\textrm{ if } \xi(j)=(*,*,*,*,*)\\
  (\xi(5j),\xi(5j+1),\ldots,\xi(5j+4))\qquad&\textrm{ otherwise }
 \end{cases}
 \end{matrix}
\end{align}
as considered in Section~\ref{sec:NC0restrictedproblembitencoding}. To connect to the corresponding result, our analysis will involve, in particular, the question of how likely it is that a randomly chosen restriction~$\xi:\mathbb{Z}_{n'}\rightarrow\{0,1,*\}$  is a block restriction. We will also need to study the number~$N(\xi^\block)=|(\xi^\block)^{-1}\{*^\block\}|$ of free blocks in a block restriction arising in this way.

\subsection{From an $\AC^0$ circuit to an $\NC^0$ circuit by random restrictions}\label{sec:fromac0tonc0}
Let $p_*\in [0,1]$. We consider the distribution~$\mathsf{R}_{p_*}$ over random restrictions~$\rho:\mathbb{Z}_n\rightarrow \{0,1,*\}$ defined as follows: To draw $\rho\sim\mathsf{R}_{p_*}$ according to~$R_{p_*}$, independently set
\begin{align}
	\rho(\ell)&=\begin{cases}
		*\qquad\textrm{ with probability } & p_*\\
		0 & (1-p_*)/2\\
		1 & (1-p_*)/2 \ 
	\end{cases}\ 
\end{align}
for each input $\ell\in \mathbb{Z}_n$ at random.

Let us define concatenation of restrictions. Let $\rho : \mathbb{Z}_n\to \{0,1,*\}$ be restriction such that $N(\rho)=|\rho^{-1}(\{*\})|$ and $\rho^{-1}(*)=\{\iota(1)<\cdots<\iota(N(\xi))\}$. Let $\eta:\mathbb{Z}_{N(\rho)}\to \{0,1,*\}$ be another restriction. We define the concatenated restriction $\rho\restriction\eta$, i.e, first restricting by $\rho$ and subsequently restricting by $\eta$, as
\begin{align}
    \begin{matrix}
    \rho\restriction\eta:& \mathbb{Z}_n & \rightarrow &\{0,1,*\}&\\[0.5cm]
    &\ell&\mapsto & (\rho\restriction\eta)(\ell):=&\begin{cases}
        \rho(\ell) \qquad & \textrm{ if }\  \rho(\ell)\not=*\\
        \eta(\iota^{-1}(\ell))  \qquad &\textrm{ if }\  \rho(\ell)=* \  \textrm{ and }\  \eta(\iota^{-1}(\ell))\not=*\\
        * \qquad &\textrm{ otherwise }\ .
    \end{cases}\ .
\end{matrix}
\end{align}

The key result we build on here is the multi-switching lemma established in~\cite{BeneWattsKothariSchaefferTalAC0}, which is itself adapted from~\cite{Rossman2017AnEP} and extended to circuits outputting multiple bits. For any $\AC^0$-circuit~$\cC$, this  gives a bound on the probability that for a randomly chosen restriction~$\rho$, the restricted function~$\cC\restriction_\rho$  has a certain structure (namely, that it is a composition of a decision tree with tuples of decision trees computing output bitstring associated to each leaf). Here we do not need this  structural statement,  but only one consequence of it (which has also been discussed in~\cite{BeneWattsKothariSchaefferTalAC0}): the restricted circuit~$\cC\restriction_\rho$ can be restricted further by an additional restriction~$\eta$ to yield a circuit~$(\cC\restriction_\rho)\restriction_\eta$ which can be implemented by an $\NC^0$ circuit. While the additional restriction~$\eta$ is a function of~$\rho$ and the circuit~$\cC$, this dependence only specifies which additional input bits should be fixed but the values can be arbitrary and hence also random.

To state this concisely, suppose~$\cC\colon\{0,1\}^n\rightarrow \{0,1\}^m$  and~$\rho\colon \mathbb{Z}_n\rightarrow \{0,1,*\}$ are given. Let $N(\rho)=|\rho^{-1}(\{*\})|$
be the number of input bits that are free when the restriction~$\rho$ is applied to~$\cC$, i.e., the restricted circuit~$\cC\restriction_\rho$ is a function~$\cC\restriction_\rho:\{0,1\}^{N(\rho)}\rightarrow\{0,1\}^m$.
We then say that a restriction~$\eta:\mathbb{Z}_{N(\rho)}\rightarrow\{0,1,*\}$ fixes additional $t$~bits of~$\cC$
if $N(\rho)-t$~bits are active when applying this restriction to~$\cC\restriction_\rho$, i.e., if 
$(\cC\restriction_\rho)\restriction_\eta$ is a function $(\cC\restriction_\rho)\restriction_\eta:\{0,1\}^{N(\rho)-t}\rightarrow \{0,1\}^m$.

\begin{lemma}[Multi-switching lemma of~\cite{BeneWattsKothariSchaefferTalAC0}, paraphrased]\label{lem:multiswitching}
    Let $\cC:\{0,1\}^n\rightarrow\{0,1\}^m$ be an $\AC^0$ circuit of size $s=\csize(\cC)$ and depth $d=\cdepth(\cC)$. Let $q,t$ be arbitrary parameters. Set
    \begin{align}
        p_*:=\frac{1}{m^{1/q}O(\log s)^{d-1}}\ .
    \end{align}
    Suppose a random restriction~$\rho:\mathbb{Z}_n\rightarrow \{0,1,*\}$ is chosen according to~$\rho\sim\mathsf{R}_{p_*}$. 
Then the following holds  with probability at least~$1-s2^{-t}$   taken over the choice of~$\rho$. 
There is a subset~$T\subset \rho^{-1}(\{*\})$  of the $N(\rho)=|\rho^{-1}(\{*\})|$~remaining free input bits (after application of~$\rho$) of size at most~$|T|\leq 2t$ such that for any restriction~$\eta:\mathbb{Z}_{N(\rho)}\rightarrow \{0,1,*\}$ fixing exactly the inputs  belonging to~$T$, the restricted circuit $(\mathsf{C}\restriction_\rho)\restriction_\eta$ can be implemented by an $\NC^0$ circuit of depth at most~$q$.
\end{lemma}
In our problem we deal with inputs that represent Clifford gates and outputs that represents Pauli operators. With a binary encoding, we represent the inputs Clifford gates by blocks of $5$~bits and the output Pauli operators by blocks of $2$~bits. That is, an $\AC^0$ circuit of interest have the block structure~$\cC_n:(\{0,1\}^{5})^n\rightarrow (\{0,1\}^2)^n$. We will assume that the circuit depth is a constant~$d=\cdepth(\cC_n)$. 

We will use the multi-switching lemma (Lemma~\ref{lem:multiswitching}) to show that for suitably chosen parameters, the following holds: There is a distribution~$P$ over block restrictions~$\xi^\block:\mathbb{Z}_n\rightarrow \{0,1\}^5\cup \{*^\block\}$ 
with the following properties, for  
a constant
$c(d)$ only depending on the circuit depth $d$ of $\cC_n$. 
With high probability over the choice of~$\xi^\block\sim P$, we have:
\begin{enumerate}[(i)]
    \item\label{it:propfirstdes} The number $N(\xi^\block)=|\xi^{-1}(\{*^\block\})|$ of active blocks is greater than $c(d)n^{3/4}$.
    \item\label{it:propseconddes} The block restricted circuit $ \cC\restriction_{\xi^\block}: (\{0,1\}^5)^{N(\xi^\block)}\rightarrow (\{0,1\}^2)^n$  is implementable by an~$\NC^0$-circuit of depth~$20d$.
\end{enumerate}
Furthermore, the distribution~$P$ satisfies the following, for a sample~$\xi^\block\sim P$ drawn according to~$P$.
\begin{enumerate}[(i)]\setcounter{enumi}{2}
    \item\label{it:propthirddes} for any subset~$S\subset \mathbb{Z}_n$, the distribution of~$\xi$ conditioned on~$\xi^{-1}(\{*^\block\})=S$
is uniform over the set of all block restriction~$\tilde{\xi}:\mathbb{Z}_n\rightarrow \{0,1\}^5\cup \{*^\block\}$ having $S=\tilde{\xi}^{-1}(\{*^\block\})$ as the set of active blocks. 
\end{enumerate}

\begin{theorem}\label{thm:ac0tonc0} Let $\cC_n:(\{0,1\}^{5})^n\rightarrow (\{0,1\}^2)^n$ be an $\AC^0$ circuit of size $s=\csize(\cC_n)$ and depth $d=\cdepth(\cC_n)$. Assume that
\begin{align}
s\le e^{n^{1/(20d)}}\ .\label{eq:sizeassumptioncircuitlowerbound}
\end{align}
Let $n$ be sufficiently large. Then
there is a probability distribution~$P$ over block restrictions
\begin{align}\xi^\block:\mathbb{Z}_n\rightarrow \{0,1\}^5\cup \{*^\block\}
\end{align}such that  $P$ satisfies~\eqref{it:propthirddes} and for $\xi^\block\sim P$ chosen at random, properties~\eqref{it:propfirstdes} and~\eqref{it:propseconddes} are satisfied with probability at least
\begin{align}
1-\exp(-c_2(d) n^{1/2})\ ,
\end{align}
where $c_2(d)$ only depends on~$d$. 
\end{theorem}
\begin{proof}
We consider~$\cC_n$ as a circuit~$\cC:\{0,1\}^{n'}\rightarrow\{0,1\}^{m'}$ where
	\begin{align}
		n'&=5 n \ ,\\
		m'&=2 n
	\end{align}
	and apply the multi-switching lemma (Lemma~\ref{lem:multiswitching}) 
	with the following choice of parameters:
	\begin{align}
		q&=20d \ ,\\
		p_*&=\frac{1}{(m')^{1/q}\cdot (\log s)^{d-1}}=\frac{1}{(2n)^{1/(20d)}\cdot (\log s)^{d-1}}\ , \\
		t&=\frac{p_*^5 n'}{5}\cdot \frac{1}{5} = \frac{p_*^5 n}{5}\ .\label{eq:tchoice} 
	\end{align}
	This choice of parameters and the size assumption~\eqref{eq:sizeassumptioncircuitlowerbound} means that the probability~$p_*$ can be bounded from below by
	\begin{align}
		p_*\geq \frac{1}{(2n)^{1/(20d)}n^{\textfrac{d-1}{(20d)}}} =  \frac{1}{2^{1/(20d)}} \cdot  n^{-1/20}  \, , \label{eq:boundonstarp}
	\end{align}
	where we used that 
	\begin{align}
		(\log s)^{d-1} &\leq \left(n^{1/(20d)}\right)^{d-1}=n^{(d-1)/(20d)}  \, 
	\end{align}
	by Eq.~\eqref{eq:sizeassumptioncircuitlowerbound}.
	
	To estimate the probability~$1-s2^{-t}$ appearing in multi-switching lemma (Lemma~\ref{lem:multiswitching}), we argue that our choice of parameters also implies that
	\begin{align}
		s &\leq 2^{t/2}\ .\label{eq:sttwo}
	\end{align}
 Indeed, we have 
	\begin{align}
		\log(2^{t/2})&=\frac{t}{2}\log 2\\
		&=\frac{\log 2}{2\cdot 5}\cdot p_*^5 n\\
		&\geq\frac{\log 2}{10}\cdot\frac{1}{2^{1/(4d)}}\cdot n^{3/4} &&\textrm{ by~\eqref{eq:boundonstarp} }\\
		&\geq n^{1/(20d)}\qquad&&\textrm{ for sufficiently large~$n$}\\
		&\geq \log s\qquad &&\textrm{ by~\eqref{eq:sizeassumptioncircuitlowerbound}}\ ,
	\end{align}
	implying the claim~\eqref{eq:sttwo}.
	
The multi-switching lemma~\ref{lem:multiswitching} asserts that when a restriction $\rho:\mathbb{Z}_{n'}\rightarrow \{0,1,*\}$ is drawn randomly according to~$\rho\sim\mathsf{R}_{p_*}$, the probability that there is no restriction~$\eta:\mathbb{Z}_{N(\rho)}\rightarrow\{0,1,*\}$ of at most~$2t$ of the remaining~$N(\rho)$ free input bits such that $(\cC\restriction_\rho)\restriction_\eta$ is implementable by an~$\NC^0$-circuit is bounded by 
    \begin{align}
		\Pr_{\rho\sim\mathsf{R}_{p_*}}[\forall \eta: (\cC\restriction_\rho)\restriction_\eta \not\in (\NC^0 \text{ of depth } 20d)]&\leq s 2^{-t}\\
		&\leq 2^{-t/2}\, \label{eq:rectrictionprob}\qquad\textrm{by Eq.~\eqref{eq:sttwo}.}
	\end{align}
	More precisely, it says the following: Consider the set~$\cE(\rho)$ consisting of all subsets~$T\subset\rho^{-1}(\{*\})$ of size~$|T|\leq 2t$ and the property that any restriction~$\eta$ of all the bits of~$T$ gives a circuit~$(\cC\restriction_\rho)\restriction_\eta$ that can be implemented by an $\NC^0$-circuit. Then this set is non-empty with high probability for a randomly chosen~$\rho\sim \mathsf{R}_{p_*}$, i.e.,
	\begin{align}
	\Pr_{\rho\sim\mathsf{R}_{p_*}}\left[\cE(\rho)\neq \emptyset\right] &\geq 1-2^{-t/2}\ .\label{eq:bounderhoestimate}
	\end{align}

Now consider the following probabilistic process defining a distribution $P$ over restrictions~$\xi:\mathbb{Z}_{n'}\rightarrow \{0,1,*\}$, see Fig.~\ref{fig:rrs} for an illustration.
\begin{enumerate}[1.]
\item\label{step:rhozerostep}
Choose a restriction~$\rho:\mathbb{Z}_{n'}\rightarrow \{0,1,*\}$ randomly according to~$\rho\sim \mathsf{R}_{p_*}$.
\item\label{step:etachoistep}
If the corresponding set~$\cE(\rho)$ is non-empty, arbitrarily pick an element~$T\in\cE(\rho)$. Then randomly and uniformly pick a restriction~$\eta:\mathbb{Z}_{N(\rho)}\rightarrow \{0,1,*\}$ among all restrictions that fix exactly the inputs associated with~$T$. That is, choose~$|T|$ bits~$(b_1,\ldots,b_{|T|})\in \{0,1\}^{|T|}$ uniformly and independently at random, and then define~$\eta$ by~$T$ (the set of fixed bits) and these bits. 

If the set~$\cE(\rho)$ is empty, choose an additional restriction~$\eta$ of the remaining free bits  uniformly at random, among all such restrictions.
\item\label{it:turnintoblockrestriction}
Consider the concatenated restriction~$\eta\restriction\rho:\mathbb{Z}_{n'}\rightarrow \{0,1,*\}$. 
Partition 
\begin{align}
    \mathbb{Z}_{n'}=\bigcup_{j\in\mathbb{Z}_n} B_j \ , \qquad\text{where}\qquad  B_j=\{5j,5j+1,\ldots,5j+4\}\ ,
\end{align}
into blocks consisting of $5$~bits each. We call a block~$B_j$ partially fixed if $0<|B_j\cap (\eta\restriction\rho)^{-1}(\{*\})|<5$, i.e., if the corresponding blocks contains both bits fixed by the restriction~$\eta\restriction\rho$, as well as bits that are free.
Then randomly choose a restriction~$\tau$ that fixes (only) every free bit in every partially fixed block~$B_j$, for $j\in \mathbb{Z}_n$, among all restrictions with this property. This amounts to choosing uniformly random bits and assigning them to each free bit  located in a partially fixed block.
\item
Output the concatenated restriction~$\xi:=\tau\restriction\eta\restriction\rho:\mathbb{Z}_{n'}\rightarrow \{0,1,*\}$.
\end{enumerate}
We claim that the distribution~$P$ over restrictions~$\xi:\mathbb{Z}_{n'}\rightarrow\{0,1,*\}$ has the required properties.

Observe first that Step~\ref{it:turnintoblockrestriction} ensures that~$\xi$ is a block restriction, i.e.,
can be interpreted as a function~$\xi^\block:\mathbb{Z}_n\rightarrow \{0,1\}^5\cup \{*\}$ with certainty. Indeed, this step guarantees that each block is either completely fixed or completely free.
Furthermore, since in the process defining~$\xi$, every bit which is fixed is fixed to a uniformly chosen bit, property~\eqref{it:propthirddes} of the distribution~$P$ is satisfied by construction. 

It remains to argue that the constructed distribution~$P$ satisfies both Properties~\eqref{it:propfirstdes} and~\eqref{it:propseconddes}  with high probability. For property~\eqref{it:propseconddes} this follows immediately from the switching lemma because of the choice of~$\eta$ in step~\ref{step:etachoistep}, the fact that fixing bits of an $\NC^0$-circuit preserves $\NC^0$-implementability, and the bound~\eqref{eq:bounderhoestimate}: We have 
\begin{align}
\Pr_{\xi\sim P}\left[\xi\textrm{ does not satisfy~\eqref{it:propseconddes}}\right] & \leq 2^{-t/2}\\
&=2^{-p_*^5 n/10}\qquad&&\textrm{ by the choice~\eqref{eq:tchoice}}\\
&\leq 2^{-\frac{1}{10\cdot 2^{1/(4d)}}n^{3/4}}\qquad&&\textrm{ by inequality~\eqref{eq:boundonstarp}}\ .\label{eq:upperboundunitionbound}
\end{align}
For property~\eqref{it:propfirstdes}, we need to estimate the number $N(\xi^\block):=|\xi^{-1}(\{*^\block\})|$ of active blocks. Let~$N(\rho)$ be the number of blocks that are (completely) free in~ the restriction~$\rho\sim \mathsf{R}_{p_*}$ chosen in step~\eqref{step:rhozerostep}. Since the proposed procedure involves fixing at most~$2t$ additional bits
in Step~\ref{step:etachoistep}, the number of free blocks in $\rho\restriction\eta$
can be lower bounded by
\begin{align}
N(\eta\restriction\rho)\geq N(\rho)-2t\ .
\end{align}
The subsequent step~\eqref{it:turnintoblockrestriction} only affects partially fixed block and thus does not change the number of free blocks: We have~$N(\xi^\block)=N(\eta\restriction\rho)$ and thus
\begin{align}
N(\xi^\block)\geq N(\rho)-2t\ .\label{eq:lowerboundxirhot2}
\end{align}
It thus remains to study the random variable~$N(\rho)$. By definition of the distribution  $\mathsf{R}_{p_*}$, the random variable~$N(\rho)$ is distributed according to the binomial distribution~$\mathsf{BIN}(n'/5, p_*^5)$ with  expectation
	\begin{align}
		\operatorname{\mathbb{E}} [\mathsf{BIN}(n'/5, p_*^5)]= \frac{n'}{5}p_*^5 &=np_*^5\ . 
	\end{align}
    In particular,
	using Hoeffding's inequality 
	$\Pr\left[Z\leq \mathbb{E}[Z]-\gamma\right]\leq \exp\left(-(2\gamma^2)/n\right)$
	for a random variable~$Z\sim \mathsf{BIN}(n,q)$ and $\gamma\leq \mathbb{E}[Z]$, we obtain
	\begin{align}
		\Pr_{\rho\sim\mathsf{R}_{p_*}}\left[N(\rho)\leq \frac{np_*^5}{2}\right]\leq 
		\exp\left(-\frac{np_*^{10}}{2}\right)\, ,\label{eq:freeblocksm}
	\end{align}
or, with~\eqref{eq:lowerboundxirhot2},
\begin{align}
\Pr_{\xi\sim P}\left[N(\xi^\block) \geq \frac{np_*^5}{2}-2t\right]&\geq 1-\exp\left(-\frac{np_*^{10}}{2}\right)\ .\label{eq:inequalityprobabilitym}
\end{align}
Since by the choice~\eqref{eq:tchoice}  of~$t$  and inequality~\eqref{eq:boundonstarp} we have 
\begin{align}
\frac{np_*^5}{2}-2t=\frac{3np_*^5}{10}\geq \frac{3}{10\cdot 2^{1/(20d)}}n^{3/4}\ ,
\end{align}
and similarly
\begin{align}
\frac{1}{2}np_*^{10} &\geq \frac{1}{2\cdot 2^{1/(2d)}}n^{1/2}\ ,
\end{align}
inequality~\eqref{eq:inequalityprobabilitym} implies
\begin{align}
\Pr_{\xi\sim P}\left[N(\xi^\block) \geq 
 \frac{3}{10\cdot 2^{1/(20d)}}n^{3/4}
\right]&\geq 1-\exp\left(-\frac{1}{2\cdot 2^{1/(2d)}}n^{1/2}\right)\ .\label{eq:inequalitysimpler}
\end{align}
Combining~\eqref{eq:inequalitysimpler} and~\eqref{eq:upperboundunitionbound} with the union bound shows that for randomly chosen~$\xi\in P$, the event 
\begin{align}
N(\xi) \geq 
 \frac{3}{10\cdot 2^{1/(20d)}}n^{3/4}\ \textrm{ and }\  
 (\cC\restriction_\xi \in (\NC^0 \text{ of depth } 20d)
 \end{align}
 occurs with probability at least
 \begin{align}1-\left(\exp\left(-\frac{1}{2\cdot 2^{1/(2d)}}n^{1/2}\right)+2^{-\frac{1}{10\cdot 2^{1/(4d)}}n^{3/4}}\right)\ .
\end{align}
This completes the proof.
\end{proof}

\begin{figure}[h]
\centering
\resizebox{0.45\textwidth}{!}{
\begin{tikzpicture}[scale=0.2]
\centering

\draw[fill=mygrey,draw=none, opacity=0.5] (-13,-3) rectangle (7,39);
\node[anchor=east] at (7,37) {\footnotesize Implementable in $\NC^0$};

\draw[dashed, myred, thick] (-19,-2) rectangle (-1,34);
\node[anchor=west, myred] at (-19,32.5) { $\xi$};

\draw[black, semithick,fill=white] (0,-2) rectangle (6,35);
\node[anchor=east] at (6,34) {$\mathcal{C}$};

\draw[semithick,fill=white] (-6,0) rectangle (-2,33);
\node[anchor=east] at (-2,32) {$\eta$};

\draw[semithick,fill=white] (-12,5) rectangle (-8,30);
\node[anchor=east] at (-8,29) {$\rho$};

\draw[semithick,fill=white] (-18,6) rectangle (-14,25);
\node[anchor=east] at (-14,24) {$\tau$};

\foreach \y in {2,...,6} {
\foreach \yy in {0,...,4} {
\draw[] (-3,\y+\yy*6) -- (2,\y+\yy*6) ;
}}

\foreach \y in {3.5,4.5} {
\foreach \yy in {0,...,4} {
\draw[] (4,\y+\yy*6) -- (9,\y+\yy*6) ;
}}

\node[anchor=east] at (-3,2) {\tiny $1$};
\node[anchor=east] at (-3,3) {\tiny $0$};
\node[anchor=east] at (-3,4) {\tiny $0$};
\node[anchor=east] at (-3,5) {\tiny $1$};
\node[anchor=east] at (-3,6) {\tiny $0$};

\node[anchor=east] at (-3,8) {\tiny $*$};
\node[anchor=east] at (-3,9) {\tiny $0$};
\node[anchor=east] at (-3,10) {\tiny $*$};
\node[anchor=east] at (-3,11) {\tiny $1$};
\node[anchor=east] at (-3,12) {\tiny $0$};

\node[anchor=east] at (-3,14) {\tiny $*$};
\node[anchor=east] at (-3,15) {\tiny $*$};
\node[anchor=east] at (-3,16) {\tiny $*$};
\node[anchor=east] at (-3,17) {\tiny $*$};
\node[anchor=east] at (-3,18) {\tiny $*$};

\node[anchor=east] at (-3,26) {\tiny $0$};
\node[anchor=east] at (-3,27) {\tiny $*$};
\node[anchor=east] at (-3,28) {\tiny $1$};
\node[anchor=east] at (-3,23) {\tiny $1$};
\node[anchor=east] at (-3,24) {\tiny $1$};

\node[anchor=east] at (-3,20) {\tiny $*$};
\node[anchor=east] at (-3,21) {\tiny $*$};
\node[anchor=east] at (-3,22) {\tiny $*$};
\node[anchor=east] at (-3,29) {\tiny $1$};
\node[anchor=east] at (-3,30) {\tiny $0$};

\foreach \y in {20,21,22,27,14,15,16,17,18,10,8} {
\draw[] (-9,\y) -- (-5,\y) ;
}

\node[anchor=east] at (-9,8) {\tiny $*$};
\node[anchor=east] at (-9,10) {\tiny $0$};
\node[anchor=east] at (-9,18) {\tiny $*$};
\node[anchor=east] at (-9,17) {\tiny $*$};
\node[anchor=east] at (-9,16) {\tiny $*$};
\node[anchor=east] at (-9,15) {\tiny $*$};
\node[anchor=east] at (-9,14) {\tiny $*$};
\node[anchor=east] at (-9,22) {\tiny $*$};
\node[anchor=east] at (-9,21) {\tiny $1$};
\node[anchor=east] at (-9,20) {\tiny $*$};
\node[anchor=east] at (-9,27) {\tiny $1$};

\foreach \y in {8,20,22,14,15,16,17,18} {
\draw[] (-15,\y) -- (-11,\y) ;
}

\node[anchor=east] at (-15,8) {\tiny $1$};
\node[anchor=east] at (-15,20) {\tiny $0$};
\node[anchor=east] at (-15,22) {\tiny $0$};

\node[anchor=east] at (-15,14) {\tiny $*$};
\node[anchor=east] at (-15,15) {\tiny $*$};
\node[anchor=east] at (-15,16) {\tiny $*$};
\node[anchor=east] at (-15,17) {\tiny $*$};
\node[anchor=east] at (-15,18) {\tiny $*$};

\foreach \y in {14,15,16,17,18} {
\draw[] (-21,\y) -- (-17,\y) ;
}

\end{tikzpicture}
}
\caption{Schematic representation of the multiple layers of restriction maps used in the proof of Theorem~\ref{thm:ac0tonc0}.}\label{fig:rrs}
\end{figure}

\subsection{An $\AC^0$ circuit size lower bound for the single-qubit gate-teleportation problem}
In this section, we complete the proof for the average-case $\AC^0$ size lower bound on the single-qubit gate-teleportation problem over uniformly random instances.

Since the multi-switching lemma works with circuits taking binary inputs, we show an average case result for~$\AC^0$-circuits for uniformly random bitstrings encoding Clifford gates by maps defined in Section~\ref{sec:FromCliffordsToBits}, see Eq.~\eqref{eq:encodingmapiotacx}. We  then rely on results from Section~\ref{sec:FromCliffordsToBits} to obtain our average case circuit size lower bound for $\AC^0$ circuits: Here we consider uniformly random sequences of Cliffords as input.

\begin{theorem}\label{thm:ac0boundbits}
Let $\enc\in\cE^n$ be an arbitrary $n$-tuple of valid encoding maps as described in Section~\ref{sec:FromCliffordsToBits}. 
Suppose $\cC_n:(\{0,1\}^5)^n\rightarrow \pauli^n$ is 
an $\AC^0$ circuit which solves the bit encoded single-qubit gate-teleportation problem with probability at least~$0.988$ on average for a bitstrings $x\in (\{0,1\}^5)^n$ sampled uniformly at random, i.e.,
    \begin{align}
\Pr_{x\in(\{0,1\}^5)^n}\left[
\left(\enc(x), \cC_n(x)\right)\in \telepR
\right]\geq 0.988\ .
\end{align}
Then the circuit has size at least
\begin{align}
    \csize(\cC_n)>e^{n^{1/(20 \cdepth(\cC_n))}}\ .
\end{align}
\end{theorem}
\begin{proof} For the sake of contradiction, assume there exists an $n$-tuple~$\enc\in\cE^n$ of valid encoding maps and an $\AC^0$ circuit $\cC_n:(\{0,1\}^5)^n\rightarrow \pauli^n$ having depth $d=\cdepth(\cC_n)$ and $\csize(\cC_n) = s \le e^{n^{1/(20 d)}}$ such that 
\begin{align}
\Pr_{x\in(\{0,1\}^5)^n}\left[
\left(\enc(x), \cC_n(x)\right)\in \telepR
\right]\ge 0.988\ .
\end{align}
By Theorem~\ref{thm:ac0tonc0}  there is a probability distribution $P$ over block restrictions $\xi^\block:\mathbb{Z}_n\to\{0,1\}^5\cup\{*^\block\}$ such  the distribution~$P$ satisfies the condition~\eqref{it:propthirddes} and a randomly chosen~$\xi^\block\sim P$ satisfies~\eqref{it:propfirstdes} and~\eqref{it:propseconddes} with probability at least
\begin{align}
1-\exp(-c_2(d) n^{1/2})\ ,
\end{align}
where $c_2(d)$ only depends on~$d$. Let us restate these conditions here for convenience:
\begin{enumerate}[(i)]
    \item The number of active blocks satisfies $N(\xi^\block)\ge c(d)n^{3/4}$ for $c(d)$ depending only on the circuit depth~$d$.
    \item The block restricted circuit $\cC_n\restriction_{\xi^\block}$ can be implemented by an $\NC^0$ circuit $\cC'_n$ of depth~$20d$.
    \item For a restriction $\xi^\block\sim P$ drawn according to the distribution~$P$ the following holds. For any subset $S\subset\mathbb{Z}_n$, the distribution of $\xi^\block$ conditioned on $S=(\xi^\block)^{-1}(\{*^\block\})$ is uniform over all block restrictions $\tilde{\xi}^\block:\mathbb{Z}_n\to \{0,1\}^5\cup\{*^\block\}$ having $S$ as the set of active blocks, i.e.,  $S=(\tilde{\xi}^\block)^{-1}(\{*^\block\})$.\label{it:ac0third}
\end{enumerate}
Let us use $(\xi^\block)^{-1}(\{*^\block\})=\{i_0< i_1<\ldots< i_{N(\xi^\block)-1}\}$ to denote the active blocks of $\xi^\block$. 
We omit the superscript in~$\xi^\block$ and~$*^\block$ in the following, and simply write~$\xi:\mathbb{Z}_{5n}\rightarrow (\{0,1\}^5)^n\times \{*\}$. 
Let us define the bit-encoding
\begin{align}
    \widehat{\enc}_\xi(y)=(\enc^{(i_0)}(y^{(0)}),\ldots,\enc^{(i_{N(\xi)-1})}(y^{(N(\xi)-1)}))\in\cliff^{N(\xi)}\qquad\textrm{ for }\qquad y\in (\{0,1\}^5)^{N(\xi)}\ .
\end{align}
for input instances of the $\telepR\restriction_{\enc(\xi)}$. From Condition~\eqref{it:ac0third} we have
\begin{align}
   \Pr_{x\in(\{0,1\}^5)^n}\left[
\left(\enc(x), \cC_n(x)\right)\in \telepR
\right] 
&= \Pr_{\substack{
        \xi\sim P\\
        y\in (\{0,1\}^5)^{N(\xi)}
        }} 
        \left[
\left(\widehat{\enc}_\xi(y), \cC_n\restriction_\xi(y)\right)\in \telepR\restriction_{\enc(\xi)}
\right] \ .\label{eq:proboverrestrictionandinputs}
\end{align}
Recall that with probability at least~$1-\exp(-c_2(d)n^{1/2})$, the block-restriction~$\xi$ is such that there exists an~$\NC^0$ circuit~$\cC'_n$ implementing the restricted circuit~$\cC_n\restriction\xi$, and this circuit solves instances of size at least~$c(d)n^{3/4}$. By Theorem~\ref{thm:maintheoremrestrictednc0bitstring} we have
\begin{align}
    \Pr_{y\in(\{0,1\}^5)^{N(\xi)}}\left[
\left(\widehat{\enc}_{\xi}(y), \cC'_n(y)\right)\in \telepR\restriction_{\enc(\xi)}
\right]< \frac{80}{81}\ ,
\end{align}
over the uniform choice of instances from $(\{0,1\}^5)^{N(\xi)}$.

Therefore we can bound~\eqref{eq:proboverrestrictionandinputs} by using the law of total probability, obtaining
\begin{align}
   \Pr_{x\in(\{0,1\}^5)^n}\left[
\left(\enc(x), \cC_n(x)\right)\in \telepR
\right] 
&< \exp(-c_2(d) n^{1/2})+\left(1-\exp(-c_2(d) n^{1/2})\right)\frac{80}{81}\\
&< 0.988\ ,
\end{align}
for sufficiently large $n$. This is a contradiction.
\end{proof}

By using Lemma~\ref{lem:clifftobitprob}, we obtain the following corollary.
\begin{corollary}\label{cor:ac0main}
Suppose $\cC_n:\cliff^n\rightarrow\pauli^n$ is an $\AC^0$-circuit  which solves the  single-qubit gate-teleportation problem with probability at least $0.9888$ on average over uniformly randomly chosen instances, i.e., 
    \begin{align}
\Pr_{C\in\cliff^n}\left[
\left(C, \cC_n(C)\right)\in R_{\Utelep_n}
\right]\geq 0.9888\ .
\end{align}
Then the circuit size of~$\cC_n$ is at least
\begin{align}
\csize(\cC_n) >e^{n^{1/(20 \cdepth(\cC_n))}}\ .
\end{align}
\end{corollary}

This together with Theorem~\ref{thm:mainquantumcircuit} implies the following corollary. 
\begin{corollary}The relation problem~$R_{\Utelep_n}$ separates the (relational) complexity classes~$\AC^0$ and $\QNC^0$ of constant-depth classical circuit with unbounded AND and OR gates as well as NOT gates and constant-depth quantum circuits, i.e., we have~$\QNC^0\not\subseteq \AC^0$.
\end{corollary}

\section{Quantum advantage with noisy shallow 3D-local circuits against $\AC^0$}
In~\cite{BGKT}, an advantage of ideal shallow  quantum circuits against $\NC^0$-circuits (based on the so-called 1$D$ magic-square problem) was lifted to an advantage of noisy $3D$-local shallow quantum circuits.  This was achieved by incorporating fault tolerance in a non-standard manner into the considered computational problems. We give a brief overview of the general method developed in~\cite{BGKT} in Section~\ref{sec:liftingmethods}. This technique was subsequently applied in~\cite{grier2021interactiveNoisy} to establish a number of complexity-theoretic separations (in both non-interactive and interactive settings) between noisy shallow quantum circuits and classical shallow circuits, but without locality considerations. 

\subsection{Noisy quantum circuits: Local stochastic noise\label{sec:localstochasticnoise}}
To state our result in detail, we use the following notions formalizing (certain) noisy quantum circuits. These  were first proposed in pioneering work by Gottesman~\cite{gottesmanoverhead}, see also e.g.,~\cite{fawzi2018constant}. A Pauli error~$E$ on $n$~qubits (in the following indexed by~$\mathbb{Z}_n=\{1,\ldots,n\}$) is a random variable taking values in~$\pauli(n):=\{I,X,Y,Z\}^{\otimes n}$. 
We denote by~$\supp(E)\subseteq [n]$ the (random) subset of qubits acted on non-trivially by~$E$ and call this the support of~$E$. The Pauli error~$E$ is called {\em $p$-local stochastic noise} for $p\in [0,1]$ if 
\begin{align}
\Pr\left[F\subseteq \supp(E)\right] &\leq p^{|F|}\qquad\textrm{ for every subset }\qquad F\subseteq [n]\ .
\end{align}
That is, $p$-local stochastic noise non-trivially affects any subset of qubits with a probability exponentially suppressed in the size of the subset. The parameter~$p$ will be called the {\em noise strength}, and we will use the notation~$E\sim\cN(p)$ to express that~$E$ is $p$-local stochastic noise. 

In the following, we consider noisy implementations of  processes where a Clifford circuit~$U=C_D\cdots C_1$ on $n$~qubits of constant depth $D=O(1)$ is applied to an initial product state~$\ket{0^n}$, and the resulting final state is measured in the computational basis. Here each $C_j$ is a gate layer consisting of one- and two-qubit Clifford gates acting on disjoint subsets of qubits. We assume that local stochastic noise $E_j$ acts after each gate layer~$C_j$. For $j\in [D-1]$ this models errors occurring during execution of the gates inside~$C_j$. The error~$E_D$ accounts 
for both errors in~$C_D$ as well as readout  (measurement) errors. To 
model  errors in the state preparation, we additionally include local stochastic noise~$E_0$  before the application of the first gate layer~$C_1$. Overall, this noisy process results in a sample drawn from the distribution
\begin{align}
\tilde{p}^U(z)&=|\langle z| E_DC_D\cdots E_2C_2E_1C_1E_0\ket{0^n}|^2\cdot P(E_0,\ldots,E_D)\label{eq:noisysamplingprocess}
\end{align}
by Born's rule. We constrain the distribution~$P(E_0,\ldots,E_D)$ of errors as follows: We make the assumption each of the errors $E_j$, $j=0,\ldots,D$ is local stochastic with parameter~$p$, i.e., $E_j\sim\cN(p)$. Importantly, this is an assumption on the marginal distribution $P_{E_j}$ of each error only, and the joint distribution $P(E_0,\ldots,E_D)$ may include dependencies. We refer to this process, i.e., sampling from the distribution~\eqref{eq:noisysamplingprocess}, as an noisy implementation of~$U$ with error strength~$p$.

\subsection{Fault-tolerance building blocks\label{sec:faulttolerancebuildingblocks}}
The construction of~\cite{BGKT} makes use of a  CSS-type~\cite{calderbank1996good,steane1996multiple} quantum code~$\cQ_m$ encoding a single logical qubit into $m$~physical qubits. This code needs to have the following properties, where we use the term single-shot to refer to fault-tolerance  protocols that do not rely on repeated syndrome measurements, following seminal work by Bombin~\cite{bombin2015single}. 
\begin{enumerate}[(1)]
\item {\bf \label{it:conditionone}Condition 1: Transversal Clifford gates.} Both the logical Hadamard gate~$H$ and the logical phase gate~$S=\mathsf{diag}(1,i)$ are realized ``transversally,'' i.e., by a depth-$1$ circuit consisting of one- and two-qubit Clifford gates.

By the CSS-nature of the code~$\cQ_m$ and the fact that the $n$-qubit Clifford group is generated by $H$, $S$ and~$\mathsf{CNOT}$-gates, Condition~$1$ ensures that the entire logical Clifford group can be realized transversally when dealing with~$n$ qubits each individually encoded in~$\cQ_m$.

\item {\bf \label{it:conditiontwo} Condition 2: Single-shot state preparation.} There is a constant-depth Clifford circuit~$W$ which uses~$\manc$~auxiliary qubits with the following property: If~$W$ is applied to the state~$\ket{0^m}\otimes\ket{0^\manc}$ and the auxiliary qubits are measured in the computational basis, the post-measurement state  is the logical state~$\ket{\overline{0}}$ (the~$+1$ eigenstate of the logical Pauli operator~$\overline{Z}$) up to a Pauli correction~$\rec(s)$ determined by the measurement outcome~$s\in\{0,1\}^\manc$. That is, there is a function~$\rec:\{0,1\}^\manc\rightarrow\pauli(m)$ such that 
\begin{align}
(\rec(s)\otimes \proj{s})W(\ket{0^m}\otimes\ket{0^{\manc}}) \propto \ket{\overline{0}}\otimes\ket{s}\qquad\textrm{ for all }\qquad s\in \{0,1\}^{\manc}\ .
\label{eq:recoverymapconditionone}
\end{align}
In other words, this can be seen as a low-depth state preparation algorithm (involving measurements) that produces the desired logical state~$\ket{\overline{0}}$ up to a known Pauli correction~$\rec(s)$.

Importantly, this state preparation needs to be fault-tolerant. To express this, consider an error~$E\in\pauli(m+\manc)$ occurring before the measurement. (The consideration of such errors is sufficient because Pauli errors can be propagated to this circuit location.) Then the post-measurement state of the~$m$ qubits after applying the correction~$\rec(s)$ is no longer proportional to~$\ket{\overline{0}}$, but a Pauli-corrupted version thereof. That is, there is a function~$\rep:\pauli(m+\manc)\rightarrow\pauli(m)$ such that 
\begin{align}
(\rec(s)\otimes \proj{s})EW(\ket{0^m}\otimes\ket{0^{\manc}}) \propto \rep(E)\ket{\overline{0}}\otimes\ket{s}\qquad\textrm{ for all }\qquad s\in \{0,1\}^{\manc}\ . \label{eq:recoverymapconditiontwo}
\end{align}
``Single-shot'' fault tolerance of the procedure is now expressed by the following requirement, which states that local stochastic noise before the measurement results in the residual error~$\rep(E)$ being local stochastic: There are constants $c'$ and $c''$ such that 
\begin{align}
E\sim\cN(p)\qquad\textrm{ implies that }\qquad \rep(E)\sim \cN(c'p^{c''})\qquad\textrm{ for all noise rates }p\in [0,1]\ .
\end{align}

\item {\bf \label{it:conditionthree}Condition 3: Single-shot readout.} Logical information can be read out fault-tolerantly in a single-shot manner by measuring every qubit in the computational basis. To formalize this, first consider the case of an uncorrupted encoded state~$\ket{\overline{\Psi}}\in\cQ_m$. In this case, a (destructive) measurement of the logical Pauli-$\overline{Z}$-observable can be realized by first measuring every qubit in the computational basis, obtaining an outcome~$x\in \{0,1\}^m$, and computing $(-1)^{\parity(x)}$ for a certain function~$\parity:\{0,1\}^m\rightarrow \{0,1\}$. Indeed, this is achieved by the function
\begin{align}
\parity(x)=\left(\sum_{j\in\supp(\overline{Z})} x_j\pmod 2\right)
\end{align} which 
 takes the parity of outcome bits belonging to the support of (one realization of) the logical operator~$\overline{Z}$. We note that for a code state~$\ket{\overline{\Psi}}\in\cQ_m$, the measurement outcome~$x\in\{0,1\}^m$ belongs with certainty to a subspace~$\cL\subset \{0,1\}^n$ of strings which obey all~$Z$-stabilizers. 
 
If instead, a state~$E\ket{\overline{\Psi}}$ corrupted by a Pauli error~$E=X(v)Z(w)$, $v,w\in \{0,1\}^m$ is measured in the computational basis, the measurement outcome is of the form~$x\oplus v$ for $x\in\cL$, where $\oplus$ denotes bitwise addition modulo two. Determination of the expection of~$\langle \overline{\Psi},\overline{Z}\overline{\Psi}\rangle$ from the measurement outcome is possible if~$\parity(x)$ can be computed from~$x\oplus v$. Condition~\eqref{it:conditionthree} results from requiring that this is the case with high probability when~$E$ is local stochastic with a sufficient small noise strength.  This leads to the following condition: There are constants~$c,c',q_{\textrm{th}}$ and  a function~$\dec:\{0,1\}^m\rightarrow\{0,1\}$ such that 
\begin{align}
\Pr\left[\dec(x\oplus v)=\parity(x)\right] &\geq 1-\exp(-c'm^c)\qquad\textrm{ for all }\qquad x\in \cL
\end{align}
whenever~$v$ is a random variable taking values in~$\{0,1\}^m$ such that $X(v)\sim\cN(q)$.
\end{enumerate}
{\bf Fault-tolerant implementation of a classicaly controlled Clifford circuit.} 
Consider a classically controlled Clifford circuit~$U$ acting on~$n$ qubits (and classical control bits). Assume that~$U$ has constant depth~$D$. Following~\cite{BGKT}, we write describe the action of such a circuit in terms of a family~$\{C_b\}_{b\in\{0,1\}^v}$ of $n$-qubit Clifford circuits indexed by the classical input~$b\in \{0,1\}^v$. Each Clifford circuit~$C_b$ is composed of Clifford gates that act non-trivially on at most $k=O(1)$ qubits and has depth at most~$D$. The action of~$U$ is then given by
\begin{align}
U(\ket{\Psi}\otimes \ket{b})&=(C_b\ket{\Psi})\otimes\ket{b}\qquad\textrm{ for all }\qquad \Psi\in(\mathbb{C}^2)^{\otimes n}\textrm{ and }b\in \{0,1\}^v\ .
\end{align}
We consider the situation where on input~$b\in \{0,1\}^v$, the circuit is applied to the initial state~$\ket{0^n}\otimes\ket{b}$, and the output state of the $n$~qubits is measured in the computational basis. This results in a sample~$z$ drawn from the distribution
\begin{align}
p^U_{b}(z)&=|\langle z|C_b0^n\rangle|^2\qquad\textrm{ for }\qquad z\in \{0,1\}^n\ .\label{eq:originaldistributionideal}
\end{align}
A quantum code~$\cQ_m$ satisfying Conditions~\eqref{it:conditionone}-\eqref{it:conditionthree}  gives rise to an implementation of this  process
which is fault-tolerant against local stochastic noise below some threshold strength.
This fault-tolerant implementation relies on an  extended classically controlled circuit~$\Uext$ on~$n\cdot (m+\manc)$ qubits and $v$~classical control bits, as well as classical post-processing of the measurement outcome.    
 The circuit~$\Uext$ (cf.~\cite[Fig.~8]{BGKT}) 
 is applied to an initial state of the form~$(\ket{0^m}\otimes\ket{0}^{\manc})^{\otimes n})\otimes\ket{b}$. Here in each of the~$n$ factors, the first~$m$ qubits correspond to an encoded logical qubit of the initial circuit~$U$ (excluding classical controls), whereas the latter $\manc$~qubits are auxiliary qubits for single-shot logical state preparation. The remaining~$v$ qubits in the state~$\ket{b}$ correspond to the classical control. Let us describe the action of~$\Uext$ for a classical input~$b\in \{0,1\}^v$. It  proceeds by  applying the state preparation unitary~$W$ to each of the~$m$ factors. Subsequently, a logical implementation~$\overline{C_b}$ of the Clifford~$C_b$ is applied to the~$\cQ_m$-encoded logical qubits. This is realized by replacing each (classically controlled) Clifford gate in~$C_b$ by the corresponding transversal implementation. 
 
 Now consider the process of 
 \begin{enumerate}[(i)]
 \item applying~$\Uext$ to the initial state~$(\ket{0^m}\otimes\ket{0}^{\manc})^{\otimes n})\otimes\ket{b}$,
 \item subsequently measuring each auxiliary qubit in the computational basis obtaining outcomes~$s=(s^1,\ldots,s^n)\in (\{0,1\}^{\manc})^n$, and
 \item measuring  each of the remaining $n\cdot m$ qubits obtaining outcomes~$y=(y^1,\ldots,y^n)\in (\{0,1\}^m)^n$.
 \end{enumerate}
 This results in a sample $(s,y)$ drawn from the distribution
 \begin{align}
p^{\Uext}_b(s,y)&= |(\langle y|\otimes\langle s|)(\overline{C}_b\otimes I)W^{\otimes n} (\ket{0^{mn}}\otimes\ket{0^{\manc n}}|^2\ .
 \end{align}
The process described here involving the extended unitary~$\Uext$ can be used in a black-box manner to generate a sample from the original distribution~\eqref{eq:originaldistributionideal} by post-processing the measurement result~$(s,y)$. Furthermore, even if its realization is affected by local stochastic noise acting in between gate layers, the resulting distribution is close to the ideal distribution~\eqref{eq:originaldistributionideal}, see Theorem~\ref{thm:noisetolerance} below.  The correct post-processing map can be described as follows. It is obtained by propagating the Pauli operator~$\rec(s^1)\otimes \cdots \otimes \rec(s^n)$ associated with the syndrome $s=(s^1,\ldots,s^n)$ through the logical gate~$\overline{C_b}$. More precisely, only the Pauli-$X$-part of the corresponding operators matters. That is, consider the functions $f,h:\{0,1\}^{n\manc}\times \{0,1\}^v\rightarrow (\{0,1\}^m)^n$ defined  by
\begin{align}
\overline{C_b} (\rec(s^1)\otimes\cdots\otimes\rec(s^n))\overline{C_b}\propto X(f(s,b))Z(h(s,b))\label{eq:fdefinitioneq}
\end{align}
 for all  $s=(s^1,\ldots,s^n)\in (\{0,1\}^{\manc})^n$ and $b\in \{0,1\}^v$. Where we used the notation $X,Z:\{0,1\}^n\to\pauli^n$ to denote the strings of respective Paulis $X(s):=X_1^{s^1}\cdots X_n^{s^n}$ and $Z(s):=Z_1^{s^1}\cdots Z_n^{s^n}$.
 Denoting by~$f^i$ the restriction of the output of~$f$ to the $i$-th codeblock, a sample~$(s,y)$
 produced by a potentially noisy implementation of~$\Uext$ on input~$b\in \{0,1\}^v$ should be then postprocessed as follows: Compute
 \begin{align}
 z_i&=\dec(y^i\oplus f^i(s,b))\qquad\textrm{ for }\qquad i\in [n]\label{eq:postprocessingdefinition}
 \end{align}
 and output $z=(z_1,\ldots,z_n)\in \{0,1\}^n$.

 This post-processing leads to an outcome as described in the following Theorem~\ref{thm:noisetolerance}. This statement is a slight reformulation of
~\cite[Theorem 17]{BGKT} to which we refer to for the proof. It is a consequence of Conditions~\ref{it:conditionone}--\ref{it:conditionthree} and the fact that propagating local stochastic errors past through low-depth Clifford circuits composed of one- and two-qubit gates preserves local stochasticity (albeit increasing the noise strength in a controlled manner, see e.g.,~\cite[Lemma~11]{BGKT}).
\begin{theorem}\label{thm:noisetolerance}
Let $U$ be a constant-depth, classically controlled Clifford circuit on $n$~qubits, with~$v$ classical control bits. 
Consider a noisy implementation of the extended circuit~$\Uext$ on input~$b\in \{0,1\}^v$ resulting in an output~$(s,y)\in (\{0,1\}^{\manc})^n\times (\{0,1\}^m)^n$. We assume that the noise strength is~$p\in [0,1]$.  Let $z\in \{0,1\}^n$ be the output of the post-processing map defined by~\eqref{eq:postprocessingdefinition}. Then there is a choice of $m=O(\mathsf{poly}(\log n))$ and a threshold value~$\pthres$ (depending only on the depth of~$U$) such that the following holds for any $p<\pthres$. The sample~$z$ produced by this noisy process is drawn from a distribution~$q_b$ that is close in $L^1$-norm to the 
the distribution~\eqref{eq:originaldistributionideal}, i.e.,
\begin{align}
\left\|q_b-p^U_b\right\|_1&<0.01\ .\label{eq:faultycircuitsimulationweak}
\end{align}
\end{theorem}
The construction discussed in Theorem~\ref{thm:noisetolerance} is not exactly what is needed to show a fault-tolerant quantum advantage, for two reasons: First, the given construction does not involve locality considerations: We have made no assumptions, e.g., about a spatial arrangement of the qubits. In particular, we have not considered the question of whether or not the gates in the single-shot preparation circuit~$W$ are (geometrically) local.  Even if this is the case (for some geometry), this does not necessarily mean that the extended circuit~$\Uext$ is also geometrically local because the latter involves $n$~logical qubits.  As discussed in the introduction, non-local circuits mean that the local stochastic noise model is not physically well-motivated. We address this issue in Section~\ref{sec:3dlocalbuilding}. There we argue that the circuits relevant to our work are geometrically local in~$3D$. 

The second reason for which Theorem~\ref{thm:noisetolerance} is not directly applicable to quantum advantage considerations is the post-processing involved. As stated, Theorem~\ref{thm:noisetolerance} constitutes what can be considered a ``standard'' use of fault tolerance: Here a noisy implementation is used to weakly (approximately) simulate an ideal quantum circuit as expressed by Eq.~\eqref{eq:faultycircuitsimulationweak}. In addition to the overhead in the number of qubits and gates, this comes at the cost of having to classically compute the ``correction map'' given by the post-processing operation~\eqref{eq:postprocessingdefinition}. The functions $f^i:(\{0,1\}^{\manc})^n\rightarrow\{0,1\}^m$ used here generally fall outside the confines of constant-depth (classical) circuits, and should therefore not appear in a quantum circuit for demonstrating an advantage of shallow circuits.  We address how this issue is resolved  (following~\cite{BGKT})  in Section~\ref{sec:liftingmethods}: The post-processing can be incorporated into a new computational problem while preserving the hardness for classical circuits. This is what we mean by ``non-standard'' use of fault-tolerance techniques.

\subsection{3D-local fault-tolerant building blocks\label{sec:3dlocalbuilding}}
 In~\cite{BGKT}, it was shown that the so-called folded surface code~\cite{moussa2016transversal} (a variant of the original surface code~\cite{bravyi1998quantum}) is a quantum code~$\cQ_m$ for which Conditions~\ref{it:conditionone}--\ref{it:conditionthree} can be satisfied: Transversality of the relevant single-qubit gates (i.e., condition~\ref{it:conditionone}) was shown by Moussa in earlier work~\cite{moussa2016transversal}. A proof that single-shot decoding is possible below a threshold error strength was given in~\cite{BGKT}  following ideas from Refs.~\cite{fowler2012proof}. Condition~\ref{it:conditiontwo} is the most non-trivial aspect of the analysis of~\cite{BGKT}. Here a  slight variant of the condition was shown: namely, the given circuit prepares a logical Bell state (up to Pauli errors) encoded in two copies of the folded surface code. It is possible to obtain a single-shot state preparation procedure for the logical   state~$\ket{\overline{0}}$, i.e.,  Condition~\ref{it:conditiontwo}, from this construction, but we will not need this here. In fact, this single-shot logical Bell state preparation is key to making our circuit (as well as the one considered in~\cite{BGKT}) local.

 The  relevant variant of Condition~\ref{it:conditiontwo} is obtained by replacing the logical state~$\ket{\overline{0}}\in(\mathbb{C}^2)^{\otimes m}$ by
 the logical Bell state~$\ket{\overline{\Phi}}=\frac{1}{\sqrt{2}}(\ket{\overline{00}}+\ket{\overline{11}})\in (\mathbb{C}^2)^{\otimes m}\otimes (\mathbb{C}^2)^{\otimes m}$ encoded in two copies of the code~$\cQ_m$. It reads as follows:

\begin{description}
\item[{\bf Condition 2':}] There is a constant-depth Clifford circuit~$W$ which uses~$\manc$~auxiliary qubits such and a function~$\rec:\{0,1\}^\manc\rightarrow\pauli(2m)$ such that 
\begin{align}
(\rec(s)\otimes \proj{s})W(\ket{0^{2m}}\otimes\ket{0^{\manc}}) \propto \ket{\overline{\Phi}}\otimes\ket{s}\qquad\textrm{ for all }\qquad s\in \{0,1\}^{\manc}\ .
\label{eq:recoverymapconditiononebell}
\end{align}
Furthermore, 
there is a function $\rep:\pauli(2m+\manc)\rightarrow \pauli(2m)$ such that for any Pauli error~$E\in\pauli(2m+\manc)$
we have \begin{align}
(\rec(s)\otimes \proj{s})EW(\ket{0^m}\otimes\ket{0^{\manc}}) \propto \rep(E)\ket{\overline{\Phi}}\otimes\ket{s}\qquad\textrm{ for all }\qquad s\in \{0,1\}^{\manc}\ . \label{eq:recoverymapconditiontwobell}
\end{align}
Finally, there are constants  $c'$ and $c''$ such that 
\begin{align}
E\sim\cN(p)\qquad\textrm{ implies that }\qquad \rep(E)\sim \cN(c'p^{c''})\qquad\textrm{ for all noise rates }p\in [0,1]\ .
\end{align}
 \end{description}
 {\bf A realization by a local circuit.} 
 The construction given  in~\cite{BGKT} satisfies Conditions~\ref{it:conditionone},~\ref{it:conditiontwo}' and~\ref{it:conditionthree}, with the additional property that the circuit~$W$ is local in~$3D$. Here we only give a brief overview of the geometric arrangement of qubits, see~\cite{BGKT} for details.  Locality of the circuit is  defined in terms of a corresponding interaction graph (lattice): Qubits are associated with vertices of the graph, and a two-qubit gate is a nearest-neighbor gate if it acts on two qubits associated with adjacent vertices. We say that a circuit is $3D$-local if it consists of single-qubit and nearest-neighbor gates, and if the graph is embedded in~$\mathbb{R}^3$ in such a way that there is a constant~$\kappa>0$ such that  the Euclidean distance between any two adjacent vertices is upper bounded by~$\kappa$, and any ball of radius~$\kappa$ contains only~$O(1)$ qubits. For example, this is the case when the interaction graph is a regular lattice.

The construction of~\cite{BGKT} has polylogarithmic dependence
\begin{align}
    m,\manc\in \Theta(\poly(\log n))\label{eq:mmancscaling}
\end{align}
of the parameters~$m,\manc$ on~$n$. 
 The~$2m$ qubits (corresponding to two copies of the code~$\cQ_m$) in the definition of~$W$ are located on the ``left'' and ``right'' faces of a wedge (prism) in~$3D$, see Fig.~\ref{fig:wedge}.  Each of these two faces is tesselated by a square lattice and has~$m$ qubits sitting on a sublattice. The remaining~$\manc$ auxiliary qubits constitute the ``bulk'' of the prism; they are arranged on a regular lattice embedded in the wedge. If $d$~is 
 the linear size of the prism (which is proportial to the distance of the underlying surface code), then  $m\in \Theta(d^2)$ and $\manc\in \Theta(d^3)$. 
 The scaling~\eqref{eq:mmancscaling}
 is a consequence of the choice $d=\Theta(\poly(\log n))$.

 To describe the action of the circuit~$W$, it is helpful to recall how the wedge-like  geometric arrangement arises. The construction is based on the fact that the cluster state on a 3D~cubic lattice of size~$L\times L\times L$ has localizable entanglement on two opposing faces even if corrupted by noise. This was first established in seminal work by Raussendorf, Bravyi, and Harrington~\cite{raussendorf2005long}: they showed that measuring bulk qubits in the cluster state results in a surface-code encoded Bell state on the two opposing faces, up to a Pauli correction determined by the measurement outcomes. This was shown to be the case even for a certain thermal noise model. The analysis of~\cite{BGKT} generalizes this result to local stochastic noise by establishing Condition~\ref{it:conditiontwo}' for the surface code. The corresponding cluster state can be generated from a product state by a depth-$6$ $3D$-local Clifford circuit. In the context of quantum advantage experiments, one actually  needs the folded surface code instead in order to satisfy  Condition~\ref{it:conditionone}. This  is obtained by folding the cubic lattice along a diagonal plane, resulting in the described wedge. Thus~$W$ is a depth-$6$ Clifford circuit generating a folded version of the cluster state.

 \begin{figure}
 \centering
 \includegraphics[width=0.23\textwidth]{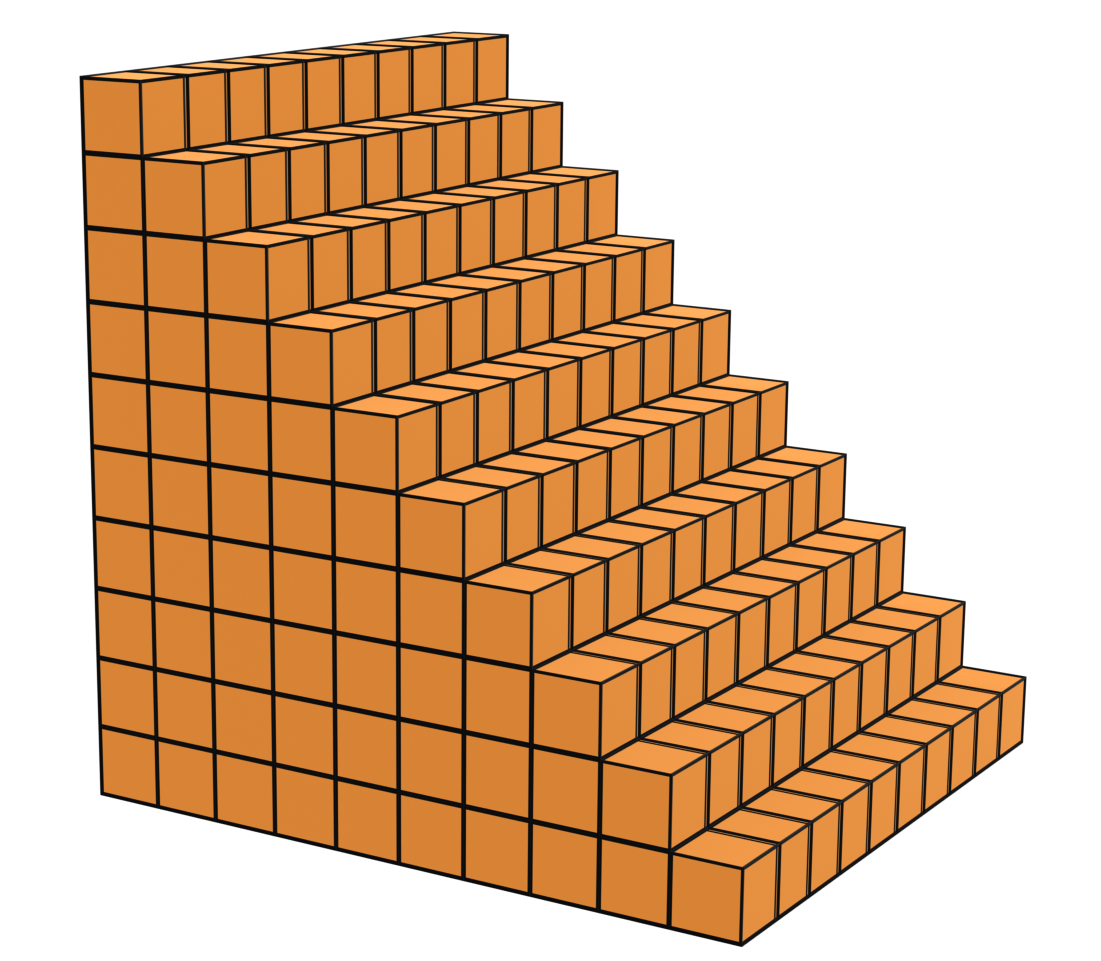}
 \caption{Geometric arrangement of qubits in the definition of the single-shot Bell state preparation procedure.\label{fig:wedge}}
 \end{figure}

 {\bf Geometric locality of the extend circuit.}
 Consider 
 a classically controlled Clifford circuit~$U$ specified by a family~$\{C_b\}_{b\in \{0,1\}^k}$ of Clifford circuits. As explained in~\cite[p.~32]{BGKT}, the associated fault-tolerant extension~$\Uext$ built using a quantum code~$\cQ_m$ satisfying Conditions~\ref{it:conditionone}--\ref{it:conditionthree} is generally not geometrically local. Non-local operations may be required as a consequence of any of the following:
\begin{enumerate}[(i)]
 \item\label{it:statepreparationWm}
The state preparation unitary~$W$ involves non-local gates.
\item\label{it:originalcliffordcircuitscliffgates}
The Clifford circuit~$C_b$ involves non-local (two-qubit) Clifford gates.
\item
The transversal implementation of a   logical one- or two-qubit Clifford gate may act  on all qubits of one or two codeblocks, respectively. If this is the case and the gate is classically controlled, the classical control may need to be geometrically non-local.\label{it:thirdissuenonlocality}
 \end{enumerate} 
 
 Non-locality in classical controls as in~\eqref{it:thirdissuenonlocality} are the least problematic here because they can be addressed by copying corresponding classical input (control) bits. This is achieved by suitably modifying the computational problem, see Section~\ref{sec:liftingmethods}.

 Fortunately, for certain circuits of interest including the single-qubit gate-teleportation circuit considered in our work, the potential non-locality arising from~\eqref{it:statepreparationWm} can be avoided by using the $3D$-local Bell state preparation circuit~$W$. Furthermore, the issue~\eqref{it:originalcliffordcircuitscliffgates} does not arise in our setup because the single-qubit gate-teleportation circuit is geometrically $1D$-local (with qubits arranged on a ring). This  is similar to the circuit for the  extended magic-square game considered in~\cite{BGKT}, which was geometrically $1D$-local (with qubits arranged on a line). In particular, this means that  the fault-tolerant extension~$\Uext$ associated with the gate-teleportation circuit~$\Utelep_n$ 
 can be made $3D$-local following the arguments of~\cite[p.~52]{BGKT}. We state this as follows, where we identify the single-qubit Clifford group (corresponding to inputs) with a subset of~$\{0,1\}^5$ (and similarly, identify Paulis (outputs) with a subset of~$\{0,1\}^2$.  
 To achieve locality also of the classical controls, each $5$-tuple  of the input~$b\in (\{0,1\}^5)^n$ needs to be provided at $m$~locations, i.e., copies of the corresponding bits are provided in several locations as input.
  \begin{theorem}\label{thm:faulttolerantegateteleportation}
 Let $\Utelep_n$ be the gate-teleportation circuit on~$n$ qubits. Then an associated fault-tolerant circuit~$\Uext$ can be realized in a $3D$-local manner using a total of~$O(n\mathsf{\poly}(\log n))$ qubits.  Assuming that copies of the input bits~$b\in (\{0,1\}^5)^n$ are provided in the relevant locations, the associated interaction graph is obtained by 
 gluing together wedges along a circle, resulting in a graph that tesselates a Colosseum-shaped $3D$~manifold that is homeomorphic to a solid torus in~$\mathbb{R}^3$, see Fig.~\ref{fig:torus}.
   There is a constant threshold~$p_{\textrm{th}}>0$ such that for all $p<p_{\textrm{th}}$, the input-output pair~$(b,z)$ obtained by post-processing
   the measurement result~$(s,y)$ for input~$b\in (\{0,1\}^5)^n$ according to~\eqref{eq:postprocessingdefinition} satisfies
   \begin{align}
   \Pr\left[(b,z)\in R_{\Utelep_n}\right]>0.99
   \end{align}
   for any noisy implementation of~$\Uext$ with error strength~$p$.
 \end{theorem}
 
 \begin{figure}
 \centering
 \includegraphics[width=0.5\textwidth]{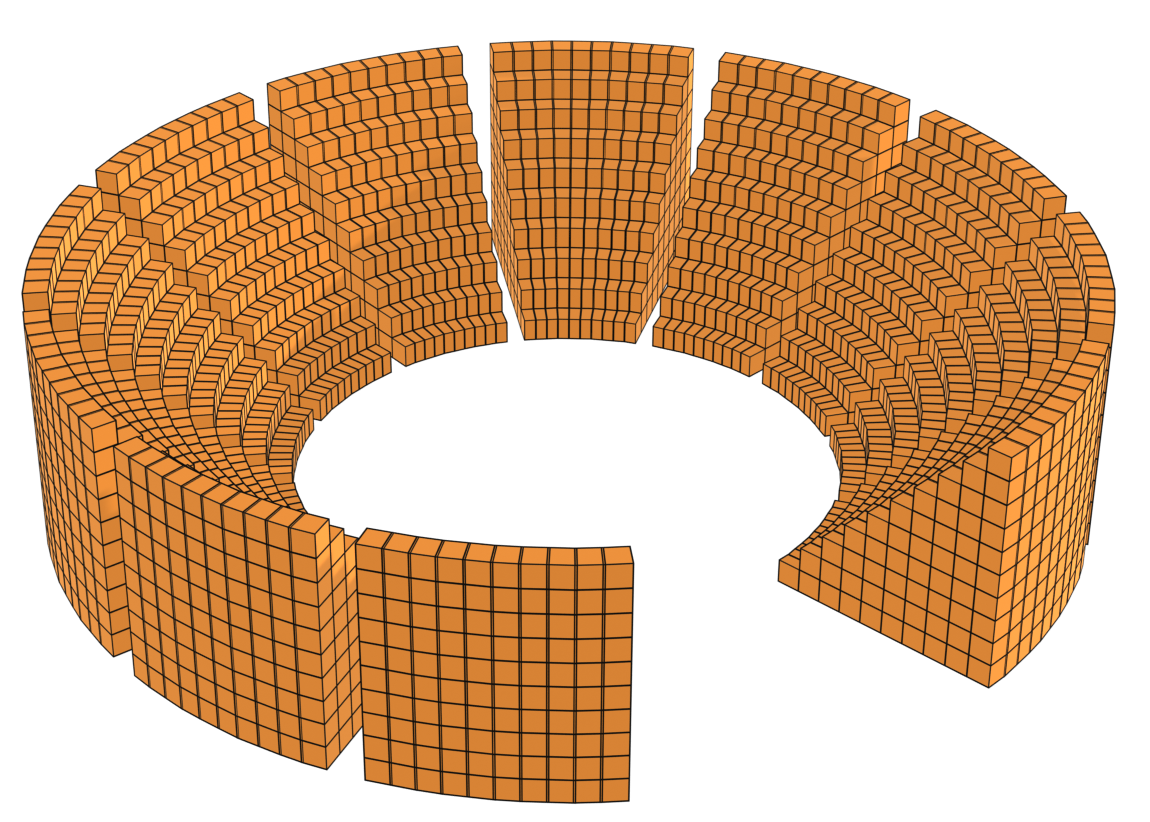}
  \caption{The Colosseum-shaped embedding in~$\mathbb{R}^3$ of the  
  interaction graph associated with the $3D$-local circuit realizing the fault-tolerant circuit~$\Uext$. By suitable post-processing, the circuit's output can be used to fault-tolerantly compute a valid solution to the relation problem~$R_{\Utelep_n}$ defined by the single-qubit gate-teleportation circuit~$\Utelep_n$.  
  \label{fig:torus}}
 \end{figure}

$3D$-locality of the circuit is a consequence of the fact 
that there are $n$~wedges of polylogarithmic size: Each wedge has~$\Theta(L)$  sites 
on each of its~$9$ bounding edges (when considered as a prism), with $L=\poly(\log n)$, see Eq.~\eqref{eq:mmancscaling}. Let $r_{in}$ and $r_{out}$ be the inner and outer radius of the Colosseum, i.e., $r_{in}$ is the radius of the arena, whereas~$r_{out}$ is the radius of the outer wall.
Let $\Delta_{out}$ be the geodesic distance between two neighboring sites at the base of the outer wall, i.e., on a circle of radius~$r_{out}$. We similarly define~$\Delta_{in}$. Let $\Delta_{R}$ be the ``radial'' (Euclidean) distance between two neighboring sites on the ground as one moves radially outwards from the arena.
Then 
\begin{align}
    r_{in}=r_{out}-L\cdot \Delta_R\ \label{eq:rinrout}
\end{align}
by definition. 
Since there are 
\begin{align}
    N_{out}=N_{in}=L\cdot n
\end{align}
points on the outer, as well as inner circle, we have 
\begin{align}
\Delta_{out}&=\frac{2\pi r_{out}}{N_{out}}=\frac{2\pi}{L\cdot n}r_{out}\\
\Delta_{in}&=\frac{2\pi r_{in}}{N_{in}}=\frac{2\pi}{L\cdot n}r_{in}\ .
\end{align}
Inserting~\eqref{eq:rinrout}
it follows that
\begin{align}
 \Delta_{in}&=\Delta_{out}-\frac{2\pi \Delta_R}{n}\ .   
\end{align}
In particular, choosing e.g., $\Delta_R=\Delta_{out}=:\Delta$ constant, it follows $\Delta_{in}$ is lower bounded by a constant (for any sufficiently large~$n$). Since the geodesic distance 
on a circle of radius~$r$ is equivalent to the Euclidean distance in~$\mathbb{R}^2$ (with constants independent of~$r$), it follows that
for this choice, the Euclidean distance between any pair of neighboring sites is both upper and lower bounded by a constant.

\subsection{Rendering quantum advantages noise-resilient~\label{sec:liftingmethods}}
Here we discuss how a quantum advantage can be made fault-tolerant. Our arguments follow~\cite{BGKT}. Consider a classically controlled Clifford circuit~$U$ that demonstrates a quantum advantage against a classical circuit class, e.g., consider the single-qubit gate-teleportation circuit~$\Utelep$ and $\AC^0$. As discussed in Section~\ref{sec:faulttolerancebuildingblocks}, any  classically controlled Clifford circuit~$U$ such as $\Utelep$ gives rise to a circuit~$\Uext$ with the property that the output of the latter can be post-processed to emulate the action of~$U$. Even though this provides a fault-tolerant solution of the problem~$R_{U}$ of possibilistically simulating~$U$, one cannot directly conclude that~$\Uext$ demonstrates a quantum advantage against the same classical circuit class. This is because of the additional classical post-processing that is required in the process. Instead, the idea of~\cite{BGKT} is to argue that the circuit~$\Uext$ solves a computational problem~$\cR_U$ different from the original problem~$R_U$, which is also computationally intractable for the classical circuit class considered. 

As with the original problem~$R_U$, the computational problem~$\cR_U$ is a relation problem, and we can specify it as a subset of input/output-pairs. Concretely, it is the subset
\begin{align}
\cR_U\subset \{0,1\}^v \times \left(\{0,1\}^{n\manc}\times \{0,1\}^{nm}\right)
\end{align}
defined as follows: We have $(b,(s,y))\in\cR_U$ if and only if 
\begin{align}
(b,z)\in R_U\qquad\textrm{ where }\qquad z_i=\dec(y^i\oplus f^i(s,b))\qquad \textrm{ for }\qquad i\in [n]\ .
\end{align}
The definition of this problem and Theorem~\ref{thm:noisetolerance} immediately imply that even a noisy implementation of~$\Uext$ provides a solution to the problem~$\cR_U$ with high probability, for any input. That is, we have 
\begin{theorem}(\cite[Theorem~17]{BGKT})
There is a constant threshold~$p_{\textrm{thres}}$ such that any noisy implementation of the circuit~$\Uext$ with error strength~$p<p_{\textrm{thres}}$ provides, on an arbitrary input~$b\in \{0,1\}^v$, an output~$(s,y)$ such that
\begin{align}
\Pr\left[(b,(s,y))\in\cR_U\right]>0.99\ .
\end{align}
\end{theorem}
\noindent It remains to argue that the problem~$\cR_U$ is hard for classical circuits if this is the case for~$R_U$.  This is a consequence of the following result, where we specialize~\cite[Theorem~18]{BGKT} to the setting of interest here. It shows that a solution to the computational problem~$\cR_U$ can be ``decoded'' by an $\AC^0$-circuit to one for the original problem~$R_U$. We note that the observation that this post-processing function can be realized by an $\AC^0$-circuit of quasiexponential size was first made in~\cite[Lemma 16]{grier2021interactiveNoisy}.
\begin{theorem}(\cite[Theorem~18]{BGKT} and \cite[Lemma 16]{grier2021interactiveNoisy}) 
\label{thm:relatingextendedandoriginalproblem} 
There is an $\AC^0$-circuit~$\cC_{\mathsf{post}}$
of depth
\begin{align}
    \cdepth(\cC_{\mathsf{post}})\in O(\cdepth(U))
\end{align} and size
\begin{align}
    \csize(\cC_{\mathsf{post}})\in O(\exp(\poly( \log n)))
\end{align}
such that the following holds. 
Given any valid solution~$(b,(s,y))\in\cR_U$ of the relation problem~$\cR_U$, the output
$(b,z)=\cC_{\mathsf{post}}((b,(s,y)))$ of the circuit~$\cC_{\mathsf{post}}$ is a valid solution of  the problem~$R_U$, i.e., $(b,z)\in R_U$.
\end{theorem}
\begin{proof}
As argued in~\cite[Proof of Theorem~18]{BGKT},
the output~$z\in \{0,1\}^n$ can be computed from~$(s,y)$ from the post-processing defined by~\eqref{eq:postprocessingdefinition} by a classical circuit~$\cC_{\mathsf{post}}$. This circuit first computes the usual $2n$-bit representation of the $n$-qubit Pauli operator~$\rec(s)=\bigoplus_{j=1}^n \rec(s^j)$. Here $s^j$ consists of $\manc$ bits of~$s$ associated with the $j$-th logical qubit.  This is achieved by considering each of the $n$~copies of the function~$\rec:\{0,1\}^{\manc}\rightarrow\pauli(m)$. As noted in \cite[Lemma 16]{grier2021interactiveNoisy}, this function can be implemented by an $\AC^0$-circuit of depth $4$ and size~$O(\exp(\manc))$ using a truth-table.
 
Once~$\rec(s)\in \pauli(m)^{\otimes n}\subset \pauli(mn)$ has been computed, we can obtain~$f^i(s,b)$ for any fixed~$i\in [n]$ as follows.  By relying on the standard stabilizer formalism and working in a gate-by-gate fashion, the result~$P\mapsto CPC^\dagger$ of propagating a single Pauli operator~$P$ through a Clifford circuit~$C$ composed of $1$- and $2$-qubit gates can be computed by a depth-$1$ classical circuit of size~$D$ whose gates have fan-in at most~$4$. Recalling the definition of the logical circuit~$\overline{C}_b$ in terms of transversal Clifford gates, we conclude that for any Pauli operator~$P\in\pauli(n\cdot m)$, the function~$b\mapsto \overline{C}_b P\overline{C}_b^\dagger$ can be computed by a circuit of depth~$\cdepth(U)$ and fan-in~$4$.  By definition of~$f$ (see Eq.~\eqref{eq:fdefinitioneq}), this means that $f^i(s,b)$ can be computed from~$(s,b)$ by an $\AC^0$-circuit of depth~$O(\cdepth(U))+O(1)$  and size~$O(\exp(\manc))$.

 Once~$f^i(s,b)$ has been computed, the output~$z_i$ (see Eq.~\eqref{eq:postprocessingdefinition}) can be obtained by first taking the bitwise XOR ~$y^i\oplus f^i(s,b)$ of two~$\manc$-bits strings. This can be achieved by a depth-$1$ circuit  of $\manc$ XOR-gates of fan-in~$2$. Finally, we apply the function~$\dec:\{0,1\}^m\rightarrow \{0,1\}$. We again note that $\dec$ can be computed by an $\AC^0$-circuit of size $O(\exp(m))$ by using a truth-table.
 
 In summary, for each~$i\in [n]$, the bit~$z_i$ is computed by an $\AC^0$-circuit
  of constant depth and size at most $O(\exp(\max\{\manc, m\}))$.
 Recalling that in our case, both $\max\{\manc,m\}=\manc$ has a polylogarithmic dependence  on~$n$ and the claim follows.

\end{proof}

\subsection{Noise-resilient quantum advantage against $AC^0$ with a $3D$-local architecture}
We can now show our main result:

\begin{theorem}[Colossal quantum advantage]\label{thm:mainresult}
There  are functions $n_i(n), n_o(n)\in \Theta(n\poly(\log n))$, a
relation
\begin{align}
\tilde{\cR}_n\subset \{0,1\}^{n_i(n)}\times 
\{0,1\}^{n_o(n)}\ ,
\end{align}
 a set of instances $\tilde{S}_n\subset \{0,1\}^{n_i(n)}$, a threshold~$p_{\textrm{thres}}>0$ and a constant $\delta>0$ 
 such that the following holds for all sufficiently large~$n$.
\begin{enumerate}[(i)]
\item\label{it:firstclaimquantumb}
There is a classically controlled, Colosseum shaped (Fig.~\ref{fig:torus}), $3D$-local constant-depth quantum circuit~$\tilde{U}$ such that for any noisy implementation of~$\tilde{U}$ with error strength~$p<p_{\textrm{thres}}$, the output~$o\in \{0,1\}^{n_o(n)}$ produced by measurement on an arbitrary input~$i\in \tilde{S}_n$ satisfies $(i,o)\in\tilde{\cR}_n$ with probability at least~$0.99$.
\item\label{it:aczeroclaimmain} Suppose~$\cC$ is an $\AC^0$-circuit which
--- with an average probability greater than~$0.9888$ over a uniformly chosen instance~$i\in\tilde{S}_n$ ---
produces an output $o\in \{0,1\}^{n_o(n)}$ such that $(i,o)\in\tilde{\cR}_n$. Then its size is at least $\csize(\cC)> e^{n^{1/(\delta \cdepth(\cC))}}$.
\end{enumerate}
\end{theorem}
We remark that the constant~$0.99$ can be replaced by any other constant by modifying the parameters of the construction, and is merely chosen for concreteness. 

\begin{proof}
Recall that the circuit~$\Utelep_n$ takes  elements of~$\cliff^n$  as input  and produces (upon measurement in the computational basis), an output from~$\pauli^n$.  Representing each output Pauli by $2$~bits, the fault-tolerance construction
of interest is essentially that given by a relation
\begin{align}
\cR_{\Utelep_n}\subset \cliff^n\times 
\left(\{0,1\}^{(2n)\manc}\times \{0,1\}^{(2n)m}
\right)
\end{align}
where $m,\manc=\poly(\log n)$, see Eq.~\eqref{eq:mmancscaling}.  The only modification we need to make to~$\cR_{\Utelep_n}$
is to provide the relevant input  entry (Clifford)  at each of the~$n$ pairs of the $2m$~pairwise neighboring physical qubits (associated with two boundaries of two neighboring wedges) encoding neighboring (logical)  qubits. This  allows for a local classical  control of the transversal Clifford gates in the associated quantum circuit. 
We also embed the set~$\cliff$ into~$\{0,1\}^5$ with an arbitrary (but fixed) injective map. This leads to a trivially  modified relation
\begin{align}
\tilde{\cR}_n&\subset \{0,1\}^{5n m}\times \left(
\{0,1\}^{(2n)\manc}\times \{0,1\}^{(2n)m}
\right)\\
&\subset \{0,1\}^{\Theta(n\poly(\log(n)))}\times \left(
\{0,1\}^{\Theta(n\poly(\log n))}\times \{0,1\}^{\Theta(n \poly(\log n))}
\right)\\
&=\{0,1\}^{\Theta(n\poly(\log(n)))}\times \{0,1\}^{\Theta(n\poly(\log n))}
\end{align}
as claimed. The described way of copying information about input Cliffords naturally defines an injective map~$\mathsf{copy}:\cliff^n\rightarrow \{0,1\}^{5nm}$. We define the set~$\tilde{S}_n$ of instances as the image~$\tilde{S}_n:=\mathsf{copy}(\cliff^n)$ of the set~$\cliff^n$ of all instances of the original problem~$R_{\Utelep_n}$.

The claim~\eqref{it:firstclaimquantumb} is shown in  Theorem~\ref{thm:relatingextendedandoriginalproblem}, see the arguments before Theorem~\ref{thm:faulttolerantegateteleportation} for the $3D$-locality of the corresponding circuit.

Let us now turn to the proof of Claim~\eqref{it:aczeroclaimmain}.
Theorem~\ref{thm:relatingextendedandoriginalproblem} tells us that there is an $\AC^0$-circuit $\cC_{\mathsf{post}}$ that, given a pair~$(b,(s,y))\in\cR_{U}$, produces a pair~$(b,z)\in R_U$. Furthermore, $\cC_{\mathsf{post}}$ is of constant depth~$\cdepth(\cC_{\mathsf{post}})\in O(1)$
(since this is the case for~$\Utelep$) and its size is at most $O(\exp\poly(\log(n)))$. 
Define 
\begin{align}
    \delta&:=40(1+\cdepth(\cC_{\mathsf{post}}))\ .
\end{align}
The proof of Claim~\eqref{it:aczeroclaimmain} is now by contradiction:
Assume that there is an $\AC^0$-circuit $\cC$ 
satisfying
\begin{align}
    \csize(\cC)<e^{n^{1/(\delta \cdepth(\cC))}}\label{eq:contrdictingassumpt}
\end{align}
such that (for sufficiently large $n$)
\begin{align}
    \Pr_{b\in \tilde{S}_n}\left[ (b,\cC(b))\in\mathcal{R}_n \right] \ge 0.9888\ .
\end{align}
Concatenating both circuits we obtain an $\AC^0$-circuit $\cC_{\mathsf{post}}\circ\cC$ of depth 
\begin{align}
\cdepth(\cC_{\mathsf{post}}\circ\cC)&=\cdepth(\cC)+\cdepth(\cC_{\mathsf{post}})\\
&\leq (1+\cdepth(\cC_{\mathsf{post}}))\cdot \cdepth(\cC)\\
&\leq \frac{\delta}{40}\cdepth(\cC) \label{eq:CpostCDepth}
\end{align}
and size
\begin{align}
    \csize(\cC_{\mathsf{post}}\circ\cC)&
    =\csize(\cC)+\csize(\cC_{\mathsf{post}}) \\
    &<e^{n^{1/(\delta \cdepth(\cC))}}+\exp(\poly(\log n))\\
    &<
    e^{n^{2/(\delta \cdepth(\cC))}} \\
    &\le e^{n^{1/(20 \cdepth(\cC_{\mathsf{post}}\circ\cC))}} && \textrm{ by Eq.~\eqref{eq:CpostCDepth}}
\end{align}
which satisfies 
\begin{align}
    \Pr_{s\in \tilde{S}_n}\left[ (\mathsf{copy}^{-1}(s),\cC_{\mathsf{post}}\circ\cC(s))\in\telepR \right] \ge 0.9888\ .
\end{align}
Therefore we can use the concatenated circuit $\cC_{\mathsf{post}}\circ\cC$, which satisfies
\begin{align}
    \csize(\cC_{\mathsf{post}}\circ\cC) < e^{n^{1/(20\cdepth(\cC_{\mathsf{post}}\circ\cC))}}
\end{align}
to solve the original problem $\telepR$ with average probability at least $0.9888$ over uniformly random instances from $\tilde{S}_n$. By definition, $\tilde{S}_n$ is a set of all images of $n$-tuples~$C\in \cliff^n$ obtained by the $\mathsf{copy}$ function. This means that
\begin{align}
    \Pr_{s\in \tilde{S}_n}\left[ (\mathsf{copy}^{-1}(s),\cC_{\mathsf{post}}\circ\cC(s))\in\telepR \right] = \Pr_{C\in\cliff^n}\left[ (C,\cC_{\mathsf{post}}\circ\cC\circ \mathsf{copy}(C))\in\telepR \right]\ ,
\end{align}
where we slightly abused the notation and allowed the function $\mathsf{copy}$ to accept inputs in the form of a Clifford gate as well as binary representation of Clifford gate. This gives us a circuit $\cC'=\cC_{\mathsf{post}}\circ\cC\circ \mathsf{copy}$ that solves $\telepR$ with average probability at least $0.9888$ over uniformly random instances from $\cliff^n$. This circuit has size at most 
\begin{align}
    \csize(\cC')<e^{n^{1/(20\cdepth(\cC'))}}\ ,
\end{align}
since $\mathsf{copy}$ can be realised as fan-out of the input nodes and does not contribute to the standard definition of circuit size.

On the other hand, Corollary~\ref{cor:ac0main} tells us that any $\AC^0$ circuit $\cC''_n$  that solves the $\telepR$ problem with average success probability
\begin{align}
    \Pr_{C\in\cliff^n}\left[
\left(C, \cC''_n(C)\right)\in \telepR
\right]\geq 0.9888\
\end{align}
has size at least
\begin{align}
    \csize(\cC''_n)>e^{n^{1/(20 \cdepth(\cC''_n))}}\ .
\end{align}
This is a contradiction to~\eqref{eq:contrdictingassumpt}.

\end{proof}

\subsection*{Acknowledgments} We thank Hjalmar Rall for help in making Fig.~\ref{fig:wedge} and Fig.~\ref{fig:torus}, and Michael de Oliveira for useful discussions. LC and RK gratefully acknowledge support by the European Research Council under grant agreement no.~101001976 (project EQUIPTNT),
as well as the Munich Quantum
Valley, which is supported by the Bavarian state government
with funds from the Hightech Agenda Bayern Plus.  XCR thanks the Swiss National Science Foundation (SNSF) for their support.


\begin{thebibliography}{10}

\bibitem{Aharonov1997}
Dorit~Aharonov and Michael~Ben-Or.
\newblock Fault-tolerant quantum computation with constant error.
\newblock In {\em Proceedings of the twenty-ninth annual {ACM} symposium on
  Theory of computing - {STOC} {\textquotesingle}97}. {ACM} Press, 1997.

\bibitem{Ajtai1983}
Mikl{\'o}s Ajtai.
\newblock {$\Sigma^1_1$-Formulae on finite structures}.
\newblock {\em Ann. Pure Appl. Log.}, 24:1--48, 1983.

\bibitem{Beame1994ASL}
Paul Beame.
\newblock A switching lemma primer.
\newblock Technical report, Technical Report UW-CSE-95-07-01, Department of
  Computer Science and Engineering, University of Washington, 1994.

\bibitem{BeneWattsKothariSchaefferTalAC0}
Adam Bene~Watts, Robin Kothari, Luke Schaeffer, and Avishay Tal.
\newblock Exponential separation between shallow quantum circuits and unbounded
  fan-in shallow classical circuits.
\newblock In {\em Proceedings of the 51st Annual ACM SIGACT Symposium on Theory
  of Computing}, STOC 2019, page 515–526, New York, NY, USA, 2019.
  Association for Computing Machinery.

\bibitem{bombin2015single}
H{\'e}ctor Bomb{\'\i}n.
\newblock Single-shot fault-tolerant quantum error correction.
\newblock {\em Physical Review X}, 5(3):031043, 2015.

\bibitem{BGKT}
Sergey Bravyi, David Gosset, Robert K{\"o}nig, and Marco Tomamichel.
\newblock Quantum advantage with noisy shallow circuits.
\newblock {\em Nature Physics}, 16(10):1040--1045, 2020.

\bibitem{BGK}
Sergey Bravyi, David Gosset, and Robert König.
\newblock Quantum advantage with shallow circuits.
\newblock {\em Science}, 362(6412):308--311, 2018.

\bibitem{bravyi1998quantum}
Sergey~B. Bravyi and Alexei~Yu. Kitaev.
\newblock Quantum codes on a lattice with boundary, 1998.
\newblock arXiv:9811052.

\bibitem{teleppaper}
Libor Caha, Xavier Coiteux-Roy, and Robert Koenig.
\newblock Single-qubit gate teleportation provides a quantum advantage, 2022.
\newblock arXiv preprint arXiv:2209.14158.

\bibitem{calderbank1996good}
A~Robert Calderbank and Peter~W Shor.
\newblock Good quantum error-correcting codes exist.
\newblock {\em Physical Review A}, 54(2):1098, 1996.

\bibitem{fawzi2018constant}
Omar Fawzi, Antoine Grospellier, and Anthony Leverrier.
\newblock Constant overhead quantum fault-tolerance with quantum expander
  codes.
\newblock In {\em 2018 IEEE 59th Annual Symposium on Foundations of Computer
  Science (FOCS)}, pages 743--754. IEEE, 2018.

\bibitem{fowler2012proof}
Austin~G Fowler.
\newblock Proof of finite surface code threshold for matching.
\newblock {\em Physical review letters}, 109(18):180502, 2012.

\bibitem{furst_parity_1984}
Merrick Furst, James~B. Saxe, and Michael Sipser.
\newblock Parity, circuits, and the polynomial-time hierarchy.
\newblock {\em Mathematical systems theory}, 17(1):13--27, December 1984.

\bibitem{LeGallAveragecase}
Fran{\c{c}}ois~Le Gall.
\newblock {Average-Case Quantum Advantage with Shallow Circuits}.
\newblock In Amir Shpilka, editor, {\em 34th Computational Complexity
  Conference (CCC 2019)}, volume 137 of {\em Leibniz International Proceedings
  in Informatics (LIPIcs)}, pages 21:1--21:20, Dagstuhl, Germany, 2019. Schloss
  Dagstuhl--Leibniz-Zentrum fuer Informatik.

\bibitem{gottesmanoverhead}
Daniel Gottesman.
\newblock Fault-tolerant quantum computation with constant overhead.
\newblock {\em Quantum Info. Comput.}, 14(15–16):1338–1372, nov 2014.

\bibitem{GottesmanChuangNature}
Daniel Gottesman and Isaac~L. Chuang.
\newblock Demonstrating the viability of universal quantum computation using
  teleportation and single-qubit operations.
\newblock {\em Nature}, 402:390--393, 1999.

\bibitem{grier2021interactiveNoisy}
Daniel Grier, Nathan Ju, and Luke Schaeffer.
\newblock Interactive quantum advantage with noisy, shallow clifford circuits,
  2021.
\newblock arXiv:2102.06833.

\bibitem{grier2020interactive}
Daniel Grier and Luke Schaeffer.
\newblock {Interactive shallow Clifford circuits: Quantum advantage against
  NC$^1$ and beyond}.
\newblock In {\em Proceedings of the 52nd Annual ACM SIGACT Symposium on Theory
  of Computing}, pages 875--888, 2020.

\bibitem{hasegawalegallArbitrarycorruption}
Atsuya Hasegawa and Fran\c{c}ois Le~Gall.
\newblock {Quantum Advantage with Shallow Circuits Under Arbitrary Corruption}.
\newblock In Hee-Kap Ahn and Kunihiko Sadakane, editors, {\em 32nd
  International Symposium on Algorithms and Computation (ISAAC 2021)}, volume
  212 of {\em Leibniz International Proceedings in Informatics (LIPIcs)}, pages
  74:1--74:16, Dagstuhl, Germany, 2021. Schloss Dagstuhl -- Leibniz-Zentrum
  f{\"u}r Informatik.

\bibitem{Hastad86}
Johan H{\aa}stad.
\newblock {Almost Optimal Lower Bounds for Small Depth Circuits}.
\newblock In {\em Proceedings of the Eighteenth Annual ACM Symposium on Theory
  of Computing}, STOC '86, page 6–20, New York, NY, USA, 1986. Association
  for Computing Machinery.

\bibitem{hoyerspalek}
Peter H{\o}yer and Robert {\v S}palek.
\newblock Quantum fan-out is powerful.
\newblock {\em Theory of Computing}, 1(5):81--103, 2005.

\bibitem{mermin}
David Mermin.
\newblock Extreme quantum entanglement in a superposition of macroscopically
  distinct states.
\newblock {\em Physical Review Letters}, 65(15):1838, 1990.

\bibitem{moussa2016transversal}
Jonathan~E. Moussa.
\newblock Transversal clifford gates on folded surface codes.
\newblock {\em Physical Review A}, 94(4):042316, 2016.

\bibitem{PERES1990107}
Asher Peres.
\newblock Incompatible results of quantum measurements.
\newblock {\em Physics Letters A}, 151(3):107--108, 1990.

\bibitem{raussendorf2005long}
Robert Raussendorf, Sergey Bravyi, and Jim Harrington.
\newblock Long-range quantum entanglement in noisy cluster states.
\newblock {\em Physical Review A}, 71(6):062313, 2005.

\bibitem{Rossman2017AnEP}
Benjamin Rossman.
\newblock {An entropy proof of the switching lemma and tight bounds on the
  decision-tree size of AC0}.
\newblock 2017.
\newblock \url{http://www.math.toronto.edu/rossman/logsize.pdf}.

\bibitem{HastadThesis}
Johan~H\aa stad.
\newblock {\em {Computational limitations for small-depth circuits}}.
\newblock PhD thesis, MIT, 1987.

\bibitem{steane1996multiple}
Andrew Steane.
\newblock Multiple-particle interference and quantum error correction.
\newblock {\em Proceedings of the Royal Society of London. Series A:
  Mathematical, Physical and Engineering Sciences}, 452(1954):2551--2577, 1996.

\bibitem{terhal2002adaptive}
Barbara~M. Terhal and David~P. DiVincenzo.
\newblock Adaptive quantum computation, constant depth quantum circuits and
  {A}rthur-{M}erlin games.
\newblock {\em Quant. Inf. Comp.}, 4(2):134--145, 2004.

\bibitem{wang2021possibilistic}
Daochen Wang.
\newblock {Possibilistic simulation of quantum circuits by classical circuits}.
\newblock {\em Phys. Rev. A}, 106:062430, Dec 2022.

\bibitem{Yao}
Andrew Chi-Chih Yao.
\newblock {Separating the Polynomial-time Hierarchy by Oracles}.
\newblock In {\em {26th Annual Symposium on Foundations of Computer Science
  (SFCS 1985)}}, pages 1--10, 1985.

\end{thebibliography}
\end{document}